\documentclass{CSML}
\pdfoutput=1

\usepackage{lastpage}
\newcommand{\dOi}{13(4:5)2017}
\lmcsheading{}{1--\pageref{LastPage}}{}{}%
{Mar.~02, 2016}{Nov.~06, 2017}{}

\usepackage{graphicx}
\usepackage{mathrsfs, amssymb, xspace,url}
\usepackage[utf8]{inputenc}
\usepackage{floatflt}
\usepackage{latexsym}
\usepackage{hyperref}
\hypersetup{hidelinks}

\theoremstyle{plain}
\newtheorem{theorem}[thm]{Theorem}

\newtheorem{proposition}[thm]{Proposition}
\newtheorem{lemma}[thm]{Lemma}
\newtheorem{corollary}[thm]{Corollary}
\newtheorem{remark}[thm]{Remark}

\newcommand{\no}[1]{}
\newcommand{\stress}[1]{}

\newcommand{\note}[1]{}

\newcommand{\mnote}[1]{}

\usepackage{mathrsfs}
\DeclareFontFamily{OT1}{pzc}{}
\DeclareFontShape{OT1}{pzc}{m}{it}%
             {<-> s * [1.130] pzcmi7t}{}
\DeclareMathAlphabet{\mathpzc}{OT1}{pzc}%
                                 {m}{it}

\newcommand{\msf}[1]{\ensuremath{\mathsf{#1}}\xspace}

\newcommand{\aA}{\mathbf{a}} 
\newcommand{\bB}{\mathbf{b}} 
 
\newcommand{\dD}{\mathbf{d}} 
\newcommand{\tT}{\mathbf{t}} 

\newcommand{\data}{\mathit{data}} % concatenation
\newcommand{\conc}{{\cdot}} % concatenation
\newcommand{\midd}{\mathrel{\;\mid\;}}
\newcommand{\pos}{\msf{POS}}
\newcommand{\tpos}{\msf{pos}}
\newcommand{\talph}{\msf{alph}}

\newcommand{\Trees}{\textit{Trees}}

\newcommand{\xml}{\textsc{xml}\xspace}
\newcommand{\tup}[1]{\langle #1 \rangle}
\newcommand{\dbracket}[1]{[\![ #1 ]\!]}
\newcommand{\abracket}[1]{[\![ #1 ]\!]}
\newcommand{\mbb}[1]{\ensuremath{\mathbb{#1}}\xspace}
\newcommand{\N}{\mbb{N}}
\newcommand{\A}{\mbb{A}}
\newcommand{\B}{\mbb{B}}

\newcommand{\E}{\mbb{E}}
\newcommand{\Ealt}{\mathbb{F}}

\newcommand{\exptime}{\textsc{ExpTime}\xspace}
\newcommand{\expspace}{\textsc{ExpSpace}\xspace}

\newcommand{\D}{\ensuremath{\mathbb{D}}\xspace}

\newcommand{\eps}{\varepsilon}
\newcommand{\coloneqq}{\mathrel{\mathop:}=}
\newcommand{\Coloneqq}{\mathrel{{\mathop:}{\mathop:}}=}

\newcommand{\cl}[1]{\ensuremath{\mathpzc{#1}}\xspace}

\newcommand{\rtrig}{{\rightarrow^*}}

\newcommand{\rtlef}{{{}^*\!\!\leftarrow}}

\newcommand{\upw}{{\uparrow}}

\newcommand{\rtupw}{{\uparrow^*}}

\newcommand{\dow}{{\downarrow}}

\newcommand{\rtdow}{{\downarrow_*}}

\newcommand{\opguess}{\msf{guess}}
\newcommand{\opspread}{\msf{spread}}
\newcommand{\atra}{\msf{ATRA}}

\newcommand{\nsubf}{\msf{nsub}}
\newcommand{\psubf}{\msf{psub}}

\newcommand{\down}{\downarrow}

\newcommand{\td}{\downarrow^{\!*}\!}

\newcommand{\xpath}{\msf{XPath}}
\newcommand{\rxpath}{\msf{regXPath}}
\newcommand{\vxpath}{\ensuremath{\xpath(\frak V, =)}\xspace}

\newcommand{\ara}{\msf{ARA}}

\newcommand{\wsts}{\textsc{wsts}\xspace}
\newcommand{\wqo}{\msf{wqo}}%

\newcommand\BUDTA{\msf{BUDA}}
\newcommand{\budta}{\BUDTA}
\newcommand\set[1]{\ensuremath{\{#1\}}\xspace}
\newcommand\subsets{\ensuremath{\wp}\xspace}

\newcommand{\prd}{\hspace{1.3pt}{\otimes}\hspace{1.3pt}}

\newcommand{\DataMonoid}{\ensuremath{\Gamma}\xspace}
\newcommand{\DataState}{\ensuremath{\Delta}\xspace}

\newcommand{\paragr}[1]{\vspace{.25em}\par\noindent\textbf{#1}\hspace{3pt}}

\def\ie{{\em i.e.}}
\def\cf{{\em cf.}}

\begin{document}

\title{Bottom-up automata on data trees and vertical XPath}

\author{Diego Figueira}	%required
\address{CNRS, LaBRI, and University of Edinburgh}	%required

\author{Luc Segoufin}	%optional
\address{INRIA and ENS Cachan, LSV}	%optional

%\thanks{This research was funded by the ERC research project FoX under
%          grant agreement FP7-ICT-233599.}
%\thanks{We acknowledge the financial support of the Future and Emerging
%Technologies (FET) programme within the Seventh Framework Programme for
%Research of the European Commission, under the FET-Open grant agreement FOX,
%number FP7-ICT-233599.} 

\begin{abstract}
  \noindent A data tree is a finite tree whose every node carries a label from a
  finite alphabet and a datum from some infinite domain. We introduce a new
  model of automata over unranked data trees with a decidable emptiness
  problem.  It is essentially a bottom-up alternating automaton with one
  register that can store one data value and can be used to perform equality
  tests with the data values occurring within the subtree of the current
  node. We show that it captures the expressive power of the vertical fragment
  of \xpath{} ---containing the child, descendant, parent and ancestor axes---
  obtaining thus a decision procedure for its satisfiability problem.
\end{abstract}

%\category{F.4.1}{Mathematical logic and formal languages}{Mathematical logic}
%\category{H.2.3}{Database management}{Languages}[Query languages]

%\terms{Theory,Languages}
%\keywords{Data Values, Tree Automata, XPath}
%\acmformat{}
%\begin{bottomstuff}
%This research was funded by the ERC research project FoX under grant agreement FP7-ICT-233599.\\
%This is the journal version of the conference paper \cite{FigueiraS11}.
%\end{bottomstuff}
\maketitle

 \section{Introduction}
We study formalisms for data trees.  A data tree is a finite tree where each position
carries a label from a finite alphabet and a \emph{datum} from some infinite
domain.  This structure has been considered in the realm of semistructured
data, timed automata, program verification, and generally in systems
manipulating data values. Finding decidable logics or automata models over
data trees is an important quest when studying data-driven systems.

In particular data trees can model \xml documents. There exist many formalisms
 to specify or query \xml documents. For static analysis or optimization
purposes it is often necessary to test whether two properties or queries over
\xml documents expressed in some formalism are equivalent. This problem usually
boils down to a satisfiability question. One such formalism to express
properties of \xml documents is the logic \xpath ---the most widely used node selection language for \xml. Although satisfiability of
\xpath in the presence of data values is undecidable, there are some known
decidable data-aware
fragments~\cite{F09,Fig10,Fig11,Fig13,BFG08,BMSS09:xml:jacm}. Here, we investigate a rather big fragment that
nonetheless is decidable. \emph{Vertical} \xpath is the fragment that contains
all downward and upward axes, but no \emph{horizontal} axis is allowed.

We introduce a novel automaton model that captures vertical \xpath. We show
that the automaton has a decidable emptiness problem and therefore that the
satisfiability problem of vertical \xpath is decidable. The {\bf B}ottom-{\bf
  U}p {\bf D}ata {\bf A}utomata (or \budta) are bottom-up alternating tree
automata with one register to store and compare data values.  Further, these
automata can compare the data value currently stored in the register with the
data value of a descendant node, reached by a downward path satisfying a given
regular property. Hence, in some sense, it has a two-way behavior. However,
they cannot test horizontal properties on the siblings of the tree, like ``the
root has exactly three children''.

Our main technical result shows the decidability of the emptiness problem of
this automaton model. We show this through a reduction to the coverability
problem of a well-structured transition system (\wsts~\cite{FS01}), that is,
the problem of whether, given two elements $x,y$, an element greater or equal
to $y$ can be reached starting from $x$. Each \BUDTA automaton is associated
with a transition system, in such a way that a derivation in this transition
system corresponds to a run of the automaton, and vice-versa. The domain of the
transition system consists in the \emph{extended configurations} of the
automaton, which contain all the information necessary to preserve from a
(partial) bottom-up run of the automaton in a subtree in order to continue the
simulation of the run from there.  On the one hand, we show that \BUDTA can be
simulated using an appropriate transition relation on sets of extended
configurations. On the other hand, we exhibit a well-quasi-order (\wqo) on
those extended configurations and show that the transition relation is
``monotone'' relative to this \wqo. This makes the coverability problem (and
hence the emptiness problem) decidable.

Our decision algorithm is not primitive recursive. However it follows
from~\cite{FS09} that there cannot be a primitive recursive decision algorithm
for vertical \xpath.

In terms of expressive power, we show that \BUDTA can express any node
expression of the vertical fragment of \xpath{}.  \msf{Core\text{-}XPath} (term
coined in~\cite{GKP05}) is the fragment of \xpath~1.0 that captures its
navigational behavior, but cannot express any property involving data. It is
easily shown to be decidable. The extension of this language with the
possibility to make equality and inequality tests between data values is named
\msf{Core\text{-}Data\text{-}XPath} in \cite{BMSS09:xml:jacm}, and it has an
undecidable satisfiability problem~\cite{GF05}.  By ``vertical \xpath'' we
denote the fragment of \msf{Core\text{-}Data}-\xpath that can only use the
downward axes \textsc{child} and \textsc{descendant} and the upward axes
\textsc{parent} and \textsc{ancestor} (no navigation among siblings is
allowed).  It follows from our work that vertical \xpath is decidable, settling
an open question \cite[Question~5.10]{BK08}.

\bigskip

\paragr{\bf Related work.}
A model of \emph{top-down} tree automata with one register and alternating
control (\atra) is introduced in~\cite{JL11}, where the decidability of its
emptiness problem is proved. \atra are used to show the decidability of
temporal logics extended with a ``freeze'' operator.  This model of automata
was extended in~\cite{Fig10} with the name $\atra(\opguess,\opspread)$ in order
to prove the decidability of the \emph{forward} fragment of \xpath, allowing
only axes navigating downward or rightward (\textsc{next-sibling} and
\textsc{following-sibling}). The \BUDTA and \atra automata models are incomparable: \atra
can express all regular tree languages, but \BUDTA cannot; while \BUDTA can
express unary inclusion dependency properties (like ``the data values labeled
by $a$ is a subset of those labeled by $b$''), but \atra cannot. In order to
capture vertical \xpath, the switch from top-down to bottom-up seems necessary
to express formulas with upward navigation, and this also makes the
decidability of the emptiness problem considerably more difficult.
In~\cite{Fig10}, the decidability of the forward fragment of \xpath is also
obtained using a \wsts. However, the automata model and therefore also the
transition system derived from it, are significantly  different from \BUDTA and
the transition system we derive from it. In particular they cannot traverse a
tree in the same way.

Another decidable fragment of \xpath on data trees is the \emph{downward} fragment of XPath, strictly contained the vertical fragment treated here, where navigation can be done only through the child and descendant axes. This fragment is known to be decidable, \exptime-complete \cite{F09,Fig12b}. In \cite{Fig13} it is shown the decidability of the satisfiability problem for XPath where navigation can be done going downwards, rightwards or leftwards in the XML document but using only reflexive-transitive axes. That is,  where navigation is done using the \xpath axes $\rtdow$, $\rtrig$, and $\rtlef$. The complexity is of 3\expspace, and this in sharp contrast with the fact that having strict (non-reflexive) transitive axes makes the satisfiability problem undecidable.

The paper~\cite{BK08} contains a comprehensive survey
of the known decidability results for various fragments of \xpath, most of
which cannot access data values. In the presence of data values, the notable
new results since the publication of \cite{BK08} are the downward \cite{F09}
and the forward \cite{Fig10} fragments, as well as the fragment containing only
the successor axis~\cite{BMSS09:xml:jacm} (the latter is closely related to
first-order logic with two variables), or containing reflexive-transitive relations (such as descendant, or the reflexive-transitive closure of the next/previous sibling relation) \cite{Fig11,Fig13}. As already mentioned, this paper solves
one of the remaining open problems of~\cite{BK08}.

%  over data trees is also a
% logic which is close to $\xpath(\frak V,=)$ and \budta automata since in some
% sense it is also \emph{two-way}. This logic can express horizontal properties,
% like the property that restricts the tree to be linear which cannot be
% expressed by \budta, but cannot test any property that requires a non-local
% test, like ``for every $a$ labeled node there is a descendant with the same
% data value''.

\bigskip

\paragr{\bf Organization.} In
Section~\ref{section-automata} we introduce the \BUDTA model and we show  that it
captures vertical \xpath in Section~\ref{section-xpath}. The associated well-structured transition system and
the proof to show the decidability of its reachability is in
Section~\ref{sec:wsts-budta}. 

This paper is a journal version of~\cite{FigueiraS11}. Compared to the conference paper, we have modified and simplified significantly the automata model and the associated \wsts.

%%% Local Variables: 
%%% TeX-PDF-mode: t
%%% TeX-master: "main"
%%% ispell-local-dictionary: "american"
%%% End:  

 \section{Preliminaries}\label{sec-prelim}

\newcommand{\rdc}{\msf{rdc}}%reflexive-downwards-compatible
\newcommand{\aSet}{S}
\newcommand{\aSetk}{K}
\newcommand{\aSetb}{T}
\newcommand{\aSetc}{U}
\newcommand{\aSetd}{V}
\newcommand{\lqdom}{\leq_{\subsets}}
\newcommand{\lqemb}{\sqsubseteq}
\newcommand{\subsetsf}{\wp_{<\infty}} %finite powerset

\newcommand{\attrxpath}{\msf{attrXPath}}

% SPECIAL LEMMAS
\newtheorem{dicksonsl}{Dickson's Lemma}
\newtheorem{higmansl}{Higman's Lemma}
\newcommand{\Np}{\N_{+}} %positive numbers without zero
\newcommand{\anAut}{\ensuremath{\cl{A}}\xspace}
%never used \newcommand{\pto}{\rightharpoonup} %partial function

\paragr{\bf Basic notation.} Let $\subsets(S)$ denote the set of subsets of
$S$, and $\subsetsf(S)$ be the set of \emph{finite} subsets of $S$.  Let
$\N=\set{0,1,2,\dotsc}$, $\Np=\set{1,2,3,\dotsc}$, and let $[n] \coloneqq
\set{1, \dotsc, n}$ for any $n \in \Np$. We fix once and for all $\D$ to be any
infinite domain of data values; for simplicity in our examples we will consider
$\D = \N$. In general we use letters $\A$, $\B$ for finite alphabets, the
letter $\D$ for an infinite alphabet and the letters $\E$ and $\Ealt$ for any
kind of alphabet.  By $\E^*$ we denote the set of finite sequences over $\E$,
by $\E^+$ the set of finite sequences with at least one element over $\E$, and
by $\E^\omega$ the set of infinite sequences over $\E$. We write $\epsilon$ for
the empty sequence and `$\conc$' as the concatenation operator between
sequences.  We write $| S |$ to denote the length of $S$ (if $S$ is a finite
sequence), or its cardinality (if $S$ is a set).

\paragr{\bf Regular languages.}
We denote by $\text{REG}(\A)$ the set of regular expressions over the finite
alphabet \A.  We make use of the many characterizations of regular languages
over a finite alphabet \A. 
In particular, we use that a word language $\mathscr{L} \subseteq
\A^*$ is regular if{f}  it satisfies one the following equivalent
properties:
\begin{itemize}
\item there is a deterministic (or non-deterministic) finite automaton   recognizing $\mathscr{L}$,
\item it is described by a regular expression,
\item there is a finite monoid $(M, \conc )$ with a distinguished subset $T \subseteq  M$, and a monoid homomorphism $h : \A^* \to M$ such that $w \in \mathscr{L}$ if{f} $h(w)\in T$,
\item there is a finite semigroup $(S, \conc )$ with a distinguished subset $T \subseteq  S$, and a semigroup homomorphism $h : \A^* \to S$ such that for all $w$ with $|w|>0$, $w \in \mathscr{L}$ if{f} $h(w)\in T$.
\end{itemize}

Depending on the section, in order to clarify the presentation, we will use the characterization that fits the best
our needs.

\newcommand{\descendant}{\succeq} \newcommand{\descendantstr}{\succ}
\paragr{Unranked finite trees.}
By $\Trees(\E)$ we denote the set of finite ordered and unranked trees over an
alphabet $\E$.  We view each \emph{position} in a tree as an element of
$(\Np)^*$. Formally, we define $\pos \subseteq \subsetsf((\Np)^*)$ as the set
of sets of finite tree positions, such that: $X \in \pos$ if{f} (a) $X
\subseteq (\Np)^*, |X| < \infty$; (b) $X$ is prefix-closed; and (c) if $n\conc
(i+1) \in X$ for $i\in\Np$, then $n \conc i \in X$ for $n \in \Np$.  A tree is then a mapping
from a set of positions to letters of the alphabet $\Trees(\E) \coloneqq \set{
  \tT : P \to \E \mid P \in \pos}$. By $\tT|_x$ we denote the subtree of $\tT$
at position $x$: $\tT|_x( y) = \tT(x \conc y)$.  The root's position is the
empty string and we denote it by `$\epsilon$'. The position of any other node
in the tree is the concatenation of the position of its parent and the node's
index in the ordered list of siblings. In this work we use $v,w,x, y, z$
as variables for positions, and $i,j,k,l,m,n$ as variables for numbers. Thus two positions $x,y$ are sibling if they are of the form $x = z\cdot i$ and $y = z \cdot j$ for some $z,j$; whereas $x$ is the parent of $y$ (resp.\ $y$ is the child of $x$) if $y$ is of the form $x \cdot i$ for some $i$. Note that from the notation $x \conc i$ one knows that it is a position which is not the root, that has $x$ as parent position, and that has $i-1$ siblings to the left.

Given a tree $\tT \in \Trees(\E)$, $\tpos(\tT)$ denotes the domain of $\tT$,
which consists of the set of positions of the tree, and $\talph(\tT)=\E$
denotes the alphabet of the tree. From now on, we informally refer by `node' to
a position $x$ together with the value $\tT(x)$. 

Given two trees $\tT_1 \in \Trees(\E)$, $\tT_2 \in \Trees(\Ealt)$ such that
$\tpos(\tT_1)=\tpos(\tT_2)=P$, we define $\tT_1 \otimes \tT_2 : P \to (\E
{\times} \Ealt)$ as $(\tT_1 \otimes \tT_2) (x) = (\tT_1(x),\tT_2(x))$.

\smallskip

\begin{figure}[h]
\centering
    \includegraphics[scale=0.55]{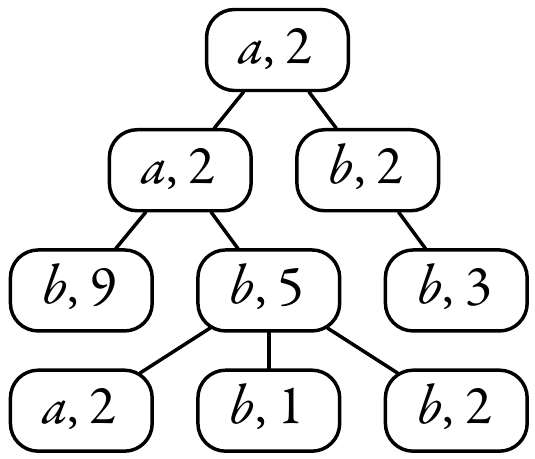}
  \caption{A data tree.}
\label{fig:data-tree}
\end{figure}

The set of \emph{data trees} over a finite alphabet $\A$ and an infinite
domain $\D$  is defined as $\Trees(\A {\times} \D)$. Note that every tree $\tT
\in \Trees(\A {\times} \D)$ can be decomposed into two trees $\aA \in
\Trees(\A)$ and $\dD \in \Trees(\D)$ such that $\tT = \aA \otimes \dD$.
Figure~\ref{fig:data-tree} shows an example of a data tree.
The notation for the set of data values used in a data tree is
$\data(\aA \otimes\dD) \coloneqq \set{\dD(x) \mid x \in \tpos(\dD)}$.
With an abuse of notation we write $\data(X)$ to denote all the elements of $\D$ contained in $X$, for whatever object $X$ may be.

\paragr{\bf Downward path.}
A \emph{downward path} starting at a node $x$ of a tree $\tT$ is the sequence
of labels of a simple path whose initial node is $x$ and going to a descendant of $x$. In other words, it is the word of the form $a_1\cdots
a_n$ where, for all $1\leq i \leq n$, $a_i=\tT(x_i)$ with $x_1=x$ and $x_{i+1}$
is a child of $x_i$.

\paragr{\bf {XP}ath on data trees.}
\label{xpath-main-definition}
Finally we define vertical \xpath, the fragment of \xpath where no horizontal
navigation is allowed.  We actually consider an extension of \xpath allowing
the Kleene star on \emph{any} path expression and we denote it by \rxpath.
Although here we define XPath (a language conceived for XML documents) over \emph{data trees} instead of over \emph{XML documents}, the main decidability result can be easily transfered to XPath over XML documents through a standard translation (see for instance~\cite{BMSS09:xml:jacm}).
% Although we define this logic over data trees, our decidability result
% also holds for the class of \xml documents through a standard
% reduction (see for instance~\cite{BDMSS06:xml}).

Vertical \rxpath is a two-sorted language, with \emph{path} expressions (denoted by $\alpha, \beta, \gamma$) and \emph{node} expressions (denoted by $\varphi, \psi,
\eta$). Path expressions are binary relations resulting from composing the
child and parent relations (which are denoted
respectively by $\dow$ and $\upw$), and node expressions. Node
expressions are boolean formulas that test a property of a node, like for
example, that it has a certain label, or that it has a child labeled $a$ with
the same data value as an ancestor labeled $b$, which is expressed by $\tup{
  \dow[a] = \rtupw[b]}$. We write $\rxpath(\frak V,=)$ to denote this logic.  A
\emph{formula} of $\rxpath(\frak V,=)$ is either a node expression or a path
expression of the logic.  Its syntax and semantics are defined in
Figure~\ref{fig:xpath-semantics}.

\noindent
\begin{figure}
\newcommand{\smallinterlinespace}{\\[0pt]}
\newcommand{\indspace}{\mathrel{\phantom{=}}\mathop{\phantom{\{}}}

\noindent
\begin{align*}
  \alpha, \beta \; &\Coloneqq \; o \midd \alpha[\varphi] \midd [\varphi]\alpha \midd
  \alpha\beta \midd \alpha \cup \beta \midd \alpha^* &&o \in
  \{\varepsilon,\dow,\upw\}\ ,\\[-1pt]
  \varphi, \psi \; &\Coloneqq \; a \midd \lnot \varphi \midd \varphi \land \psi \midd \varphi \lor \psi
  \midd \tup{\alpha} \midd \tup{\alpha= \beta} \midd \tup{\alpha \not= \beta} &&
  a \in \A \ .
\end{align*}

\centerline{The syntax of vertical $\xpath$}

  \begin{align*} 
    \abracket{\dow}^\tT & = \{(x,x\conc i) \mid x \conc i \in \tpos(\tT)\} &
    \abracket{\upw}^\tT & = \{(x\conc i,x) \mid x \conc i \in \tpos(\tT)\} \smallinterlinespace
     \abracket{[\varphi]}^\tT & = \{(x,x) 
     \mid x \in \tpos(\tT), x \in \abracket{\varphi}^\tT\}
 &
 \abracket{\alpha^*}^\tT & =
    \textrm{\scriptsize the reflexive transitive closure of
    }\abracket{\alpha}^\tT\smallinterlinespace
    \abracket{\varepsilon}^\tT & = \{(x,x) \mid x \in \tpos(\tT)\}
    &
\abracket{\alpha \beta}^\tT & = \{(x,z) \mid
    \exists y ~.~(x,y) \in
    \abracket{\alpha}^\tT,
\smallinterlinespace
\abracket{\alpha \cup \beta}^\tT & =
    \abracket{\alpha}^\tT \cup \abracket{\beta}^\tT&&\indspace  (y,z) \in \abracket{\beta}^\tT\}    \smallinterlinespace
%\displaybreak[0] 
\abracket{a}^\tT & = \{ x \in \tpos(\tT) \mid
    \aA(x) = a \} &
    \abracket{\tup{\alpha}}^\tT & = \{ x \in \tpos(\tT) \mid \exists
    y ~.~ (x,y) \in \abracket{\alpha}^\tT \}\smallinterlinespace
    \abracket{\lnot \varphi}^\tT & =  \tpos(\tT) \setminus \abracket{\varphi}^\tT &
    \abracket{\varphi \land \psi}^\tT & = \abracket{\varphi}^\tT \cap
    \abracket{\psi}^\tT \smallinterlinespace
 \abracket{\tup{\alpha {=} \beta}}^\tT 
& = \{ x \in \tpos(\tT) \mid \exists y,\!z ~.~ (x,y) \in
    \abracket{\alpha}^\tT, &
    \abracket{\tup{\alpha {\not=} \beta}}^\tT & =  \{ x \in \tpos(\tT) \mid \exists y,\!z ~.~ (x,y) \in \abracket{\alpha}^\tT,\smallinterlinespace
&\indspace (x,z) \in
    \abracket{\beta}^\tT, \dD(y)=\dD(z)\}
    &&\indspace (x,z) \in \abracket{\beta}^\tT, \dD(y)\not=\dD(z)\}
  \end{align*}
\centerline{The semantics of vertical $\xpath$ over a data tree
  $\tT = \aA \otimes \dD$}

\caption{The syntax and semantics of vertical $\xpath$.}
\label{fig:xpath-semantics}
\end{figure}

As another example, we can select the nodes that have a descendant labeled $b$
with two children also labeled by $b$ with different data values by a formula
$\varphi = \tup{\td\![\, b \, \land \,
  \tup{\down\![b]\not=\down\![b]}\,]}$. Given a tree $\tT$ as in
Figure~\ref{fig:data-tree}, we have $\abracket{\varphi}^\tT=\{\epsilon, 1, 12\}$.

The satisfiability problem for $\rxpath(\frak V, =)$ is the problem of, given a
formula $\varphi$, whether there exists a data tree $\tT$ such that
$\abracket{\varphi}^\tT\neq \emptyset$.

Our main result on \xpath is the following.
\begin{theorem}\label{thm:rvxpath-decidable}
The satisfiability problem for $\rxpath(\frak V,=)$ is decidable.
\end{theorem}

The proof of Theorem~\ref{thm:rvxpath-decidable} goes as follows.  We define a
model of automata running over data trees. This model of automata is
interesting on its own and the second main result of this paper shows that they
have a decidable emptiness problem. Finally we show that formulas of
$\rxpath(\frak V,=)$ can be translated into a \budta.

%%% Local Variables: 
%%% TeX-PDF-mode: t
%%% TeX-master: "main"
%%% ispell-local-dictionary: "american"
%%% End:  
 \newcommand{\morphism}{h}
\newcommand{\Monoid}{\+S}
\newcommand{\up}{\textit{up}}
\newcommand{\aMonoid}{\mu}
\newcommand{\bMonoid}{\nu}
\newcommand{\automatonop}[1]{\msf{#1}}
\newcommand{\opunique}{\automatonop{unique}}
\newcommand{\opaccept}{\automatonop{accept}}
\newcommand{\opdup}{\automatonop{spread}}
\newcommand{\opuniv}{\automatonop{univ}}
\newcommand{\opnothing}{\automatonop{keep}}
\newcommand{\opdisjoin}{\automatonop{disjoint}}
\newcommand{\optestroot}{\automatonop{root}}
\newcommand{\optestnroot}{\overline{\automatonop{root}}}
\newcommand{\optestleaf}{\automatonop{leaf}}
\newcommand{\optestnleaf}{\overline{\automatonop{leaf}}}
\newcommand{\opstore}{\automatonop{store}}
\newcommand{\optesteq}{\automatonop{eq}}
\newcommand{\optestneq}{\overline{\automatonop{eq}}}
\newcommand{\deltaup}{\delta_\up}
\newcommand{\tranbudta}{\rightarrowtail}
\newcommand\Tests{\textsc{Tests}\xspace}
\newcommand\Actions{\textsc{Actions}\xspace}
\newcommand{\uptrans}{\ensuremath{\tranbudta_\up}}
\section{The automata model}\label{section-automata}

In this section we introduce our automata model. It is essentially a bottom-up
tree automaton with one register to store a data value and an alternating
control. We will see in Section~\ref{section-xpath} that these automata are
expressive enough to capture vertical \rxpath. In Section~\ref{sec:wsts-budta}
we will show that their emptiness problem is
decidable. Theorem~\ref{thm:rvxpath-decidable} then follows immediately.

A Bottom-Up Data Automaton (\BUDTA)  $\anAut$ runs over data trees of $\Trees(\A
\times\D)$ and it is defined as a tuple
$\anAut=(\A,\B,Q,q_0,\delta)$ where $\A$ is the finite
alphabet of the tree, $\B$ is an internal finite alphabet of the automaton
(whose purpose will be clear later), $Q$ is a finite set of states, $q_0$ is
the initial state, and $\delta$ is the transition function which is a finite set of
pairs of the form \emph{(test,action)} that will be described below.

Before we present the precise syntax and semantics of our automaton model, we
first give the intuition.  The automaton has one register, where it can store
and read a data value from $\D$, and it has alternating control. Hence, at any
moment several independent \emph{threads} of the automaton may be running in
parallel. Each thread has one register and consists of a state from $Q$ and a
data value from \D stored in the register. The automaton first guesses a finite
internal label from $\B$ for every node of the tree and all threads share
access to this finite information. This internal information can be viewed as a
synchronization feature between threads and will be necessary later for
capturing the expressive power of vertical \rxpath.  The automaton is
bottom-up, and it starts with one thread with state $q_0$ at every leaf of the
tree with an arbitrary data value in its register. From there, each thread
evolves independently according to the transition function $\delta$: If the
test part of a pair in $\delta$ is true then the thread can perform the
corresponding action, which may trigger the creation of new threads. We first
describe the set of possible tests the automata may perform and then the set of
their possible actions.

\newcommand{\rexp}{\mathit{exp}} The tests may consist of any conjunction of
the basic tests described below or their negation.  The automata can test the
current state, the label (from \A) and internal label (from \B) of the current
node and also whether the current node is the root, a leaf or an internal
node. The automata can test equality of the current data value with the one
stored in the register (denoted by $\optesteq$). Finally the automata can test
the existence of some downward path, starting from the current node and leading
to a node whose data value is (or is not) equal to the one currently stored in
the register, such that the path matches some regular expression on the labels.
For example, for a regular expression $\rexp$ over the alphabet $\A \times \B$,
the test $\tup{\rexp}^=$ checks the existence of a downward path that matches
$\rexp$, which starts at the current node and leads to a node whose data value
matches the one currently stored in the register. Similarly,
$\tup{\rexp}^{\not=}$ tests that it leads to a data value different from the
one currently in the register.

The precise set of possible basic tests is:
\begin{equation*}
\textup{B}\Tests=\set{p,\optesteq, \tup{\rexp}^= , \tup{\rexp}^{\not=}, \optestroot,
  \optestleaf, a , b~~|~~ \rexp \in \text{REG}(\A \times \B), p \in Q, a\in\A, b \in \B}.
\end{equation*}
If $x$ is a basic test, we will write $\overline{x}$ to denote the test
corresponding to the negation of $x$. For instance $\overline{\optesteq}$ tests
whether the current data value differs from the one stored in the register. The
possible set of tests is then:
\begin{equation*}
\Tests=\textup{B}\Tests \cup \overline{\textup{B}\Tests}.
\end{equation*}

Based on the result of a test the thread can perform an action. A basic action
either accepts (\opaccept) and the corresponding thread terminates, or
specifies a new state $p$ and a new content for the register for each thread it
generates, each of them moving up in the tree to the parent node. The possible
updates of the register are: keep the register's data value unchanged (denoted
by $\opnothing$), store the current data value in the register (denoted by
$\opstore$), store an arbitrary data value non-deterministically chosen
($\opguess$), or start a new thread for every data value of the subtree
($\opuniv$) of the current node.  Note that this last action creates
unboundedly many new threads. Altogether the precise set of possible basic
actions is:
\begin{equation*}
\Actions=\set{\opaccept, \opnothing(p), \opstore(p), \opguess(p), \opuniv(p)~~~|~~~ p \in Q}
\end{equation*}
and the set of actions is any conjunction of those. As usual, conjunction
corresponds to universality.
For example with an action of the form $a_1 \land a_2$ the automaton starts two
new threads, one specified by $a_1$ and one specified by $a_2$. If $a_2$ would
be $\opuniv(p)$ then it actually starts one new thread in state $p$ and data
value $d$ per data data value $d$ occurring in the subtree of the current node.

A transition is therefore a pair \emph{(test,action)} where \emph{test} is a
conjunction of basic tests in \Tests and \emph{action} is a conjunction of
basic actions in \Actions. There might be several rules involving the same
tests, corresponding to non-determinism. 

Before we move on to the formal definition of the language accepted by a
\BUDTA, we stress that the automaton model is not closed under complementation
because its set of actions are not closed under complementation: \opguess is a
form of existential quantification while \opuniv is a form of universal
quantification, but they are not dual. Actually, we will show in
Proposition~\ref{prop:undec-dual} that adding their dual would yield
undecidability of the model.

We now turn to the formal definition of the semantics of a \BUDTA. A data tree
$\aA \prd \dD \in \Trees(\A \times \D)$ is accepted by $\anAut$ if{f} there
exists an internal labeling $\bB \in \Trees(\B)$ with $\tpos(\bB) = \tpos(\aA
\prd \dD)$ such that there is an accepting run of $\anAut$ on $\tT=\aA \prd \bB \prd \dD$. We
focus now on the definition of a run.

We say that a thread $(q,d)$ makes a basic test $t\in\textup{B}\Tests$ true at a position
$x$ of $\tT$, and write $\tT,x,(q,d) \models t$, if:
\begin{itemize}
\item $t$ is one of the tests $p$, $a$, $b$, $\optestroot$, $\optestleaf$,
  $\optesteq$ and we have respectively $q=p$, $\aA(x)=a$, $\bB(x)=b$, $x$ is
  the root of $\tT$, $x$ is a leaf of $\tT$, $\dD(x)=d$,

\item $t$ is $\tup{\rexp}^=$ and there is a downward path in $\tT$ matching
  $\rexp$, starting at $x$ and ending at $y$ where $\dD(y) = d$. The case of
  $\tup{\rexp}^{\not=}$ is treated similarly replacing $\dD(y) = d$ by $\dD(y)
  \neq d$.
\end{itemize}

This definition and notation lifts to arbitrary Boolean combination of basic
tests in the obvious way. Note that $\overline{\tup{\rexp}^=}$ is not
$\tup{\rexp}^{\not=}$. The former is true if there no downward path matching
$\rexp$ and reaching the current data value, while the latter requires the
existence of a downward path matching $\rexp$ and reaching a
data value different from the current one.

\medskip

\label{def:run}
A \emph{configuration} of a \BUDTA \anAut is a set $\+C$ of threads, viewed
as a finite subset of $Q \times \D$.  A configuration $\+C$ is said to be
\emph{initial} if{f} it is the singleton $\{(q_0,e)\}$ for some $e \in \D$.%changed definition of intial to be a singleton

\newcommand\action{\textup{Ac}}
A \emph{run} $\rho$ of $\anAut$ on $\tT=\aA\prd\bB\prd\dD$ is a function
associating a configuration to any node $x$ of $\tT$ such that
\begin{itemize}
\item for any leaf $x$ of $\tT$, $\rho(x)$ is initial, 
\item for any inner position $x$ of $\tT$, whose parent is the position $y$,
  and for any $(q,d) \in \rho(x)$ there exists $(t,\action) \in \delta$ with
  $\action=\bigwedge_{j \in J} \action_{j}$ such that $\tT,x,(q,d) \models t$ and for any
  $j\in J$ we have:
\begin{itemize}
\item if $\action_{j}$ is $\opnothing(p)$ then $(p,d)\in \rho(y)$,
\item if $\action_{j}$ is $\opstore(p)$ then $(p,\dD(x))\in \rho(y)$,
\item if $\action_{j}$ is $\opguess(p)$ then $(p,e)\in \rho(y)$ for some $e \in \D$,
\item if $\action_{j}$ is $\opuniv(p)$ then for all $e \in \data(\tT|_x)$,
  $(p,e) \in \rho(y)$. 
\end{itemize}
\end{itemize}

The run $\rho$ on $\tT$ is \emph{accepting} if moreover for the root $y$ of
$\tT$ and all $(q,d) \in \rho(y)$ there exists $(t,\opaccept) \in \delta$ such
that $\tT,y,(q,d) \models t$. 

% We have the following properties.
% \begin{proposition}\label{prop:budta-closure-properties}
%   The class of languages $\BUDTA$ definable by the \BUDTA class is:\diego{check this}
%   \begin{enumerate}[(i)]
%   \item closed under union,
%   \item closed under intersection, and
%   \item not closed under complementation.
%   \end{enumerate}
% \end{proposition}
% The proof is in Appendix~\ref{app:prop:budta-closure-properties}.

\subsection{Discussion}
\label{sec:discussion-BUDA}

\subsubsection{Semigroup notation}
For convenience of notation in the proofs, we shall use an equivalent
definition of $\BUDTA$ using \emph{semigroup homomorphisms} instead of regular
expressions. That is, we consider an automaton $\anAut \in\BUDTA$ as a tuple
$\anAut=(\A,\B,Q,q_0,\delta,\Monoid,\morphism)$ where $\Monoid$ is a finite
semigroup, $\morphism$ is a semigroup homomorphism from $(\A\times\B)^+$ to
$\Monoid$, and tests of the form $\tup{\mu}^=$ and $\tup{\mu}^{\neq}$ contain a
semigroup element $\mu \in \Monoid$. Hence, $\tup{\aMonoid}^=$ is true at $x$
in $\tT$ if there is a downward path in $\tT$ starting at $x$ and ending at a descendant
$y$, evaluating to $\aMonoid$ via $\morphism$ and such that $\dD(y) = d$.  The
case of $\tup{\aMonoid}^{\not=}$ is treated similarly replacing $\dD(y) = d$ by
$\dD(y) \neq d$. Note that since regular languages are exactly those recognized
by finite semigroups (recall Section~\ref{sec-prelim}) this is an equivalent automata model.

\subsubsection{Disjunction}
As mentioned earlier the automata model does not allow for disjunctions of
actions or tests. But in fact these can be added without changing the
expressivity of the automaton, by modifying the transition relation $\delta$:
\begin{itemize}
\item any automaton having transition with a disjunction of actions $(t, a_1
  \lor a_2)$ is equivalent to the automaton resulting from replacing $(t, a_1
  \lor a_2)$ with the transitions $(t,a_1)$ and $(t,a_2)$,
\item any automaton having transition with a disjunction of tests $(t_1 \lor
  t_2, a)$ is equivalent to the automaton resulting from replacing $(t_1 \lor
  t_2, a)$ with the transitions $(t_1,a)$ and $(t_2,a)$.
\end{itemize}
For this reason we will sometimes write disjunction of actions or of tests, as a shorthand for the equivalent automaton without disjunctions. 

We will also make use of a test $\tup{\mu}$ denoting $\tup{\mu}^= \lor \tup{\mu}^{\neq}$; and $\overline{\tup{\mu}}$ denoting $\overline{\tup{\mu}^=} \land \overline{\tup{\mu}^{\neq}}$.

\subsubsection{Closure properties}
We say that a model is closed under effective operation $O$, if it is closed
under $O$ and further the result of the application of the operation $O$ is
computable.
\begin{proposition}\label{prop:buda-iner-union}
The class \BUDTA is  closed under effective intersection and effective union.
\end{proposition}
\begin{proof}
  This is straightforward from the fact that  the model has alternation and non-determinism.
\end{proof}

As already mentioned, the closure under complementation of \BUDTA yields an undecidable model.
\begin{proposition}\label{prop:undec-dual}
Any extension of \BUDTA closed under effective complementation and effective intersection has an undecidable emptiness problem.
\end{proposition}
\begin{proof}
The automaton model $\ara(\opguess)$ introduced in \cite{Fig12b} is an alternating 1-register automaton over \emph{data words}, which are essentially data trees whose every position has at most one child. The model is equivalent to a restricted version of $\BUDTA$ without the tests $\tup{\rexp}^=$, $\tup{\rexp}^{\neq}$ and without the action $\opuniv$ running on data words. Indeed, the $\ara(\opguess)$ model is an alternating automaton with one register and a $\opguess$ operator just like the one of $\BUDTA$. Then, there is a simple reduction
  from the emptiness problem  for the automata model
  $\ara(\opguess)$ into the emptiness problem of \BUDTA. Indeed, there is a straightforward translation $f$ from $\ara(\opguess)$ to \BUDTA  so that for every $\anAut \in \ara(\opguess)$
  \begin{itemize}
  \item if a data word is accepted by $\anAut$, then it is accepted by $f(\anAut)$, and
  \item if a data tree is accepted by $f(\anAut)$, then any of its
    maximal branches (seen as data words, starting at a leaf and ending at the root) is accepted by $\anAut$.
  \end{itemize}

By means of contradiction, suppose that there is a class of automata $C$ that extends the class $\BUDTA$ and is closed under effective complementation and effective intersection. We show undecidability by reduction from the emptiness problem for Minsky machines~\cite{MM67}. In \cite[Proof of Proposition~3.2]{Fig12b}, it is shown that for every Minsky machine $M$ one can build two properties $P$ and $P'$ expressible by $\ara(\opguess)$ so that $P \land \lnot P'$ is satisfiable iff $M$ is non-empty. By the reduction above, the properties $\hat P = $ \emph{every branch satisfies $P$} and $\hat P' = $ \emph{every branch satisfies $P'$} are expressible with $C$. Since $C$ is closed under intersection and complementation, then $\hat P \land \lnot \hat P'$ is expressible with $C$. Note that a data tree satisfies $\hat P \land \lnot \hat P'$ iff it has a branch satisfying $P \land \lnot P'$, which happens iff $M$ is non-empty. Hence, $C$ has an undecidable emptiness problem.
\end{proof}

Since we will show that the emptiness problem for \BUDTA is decidable and $\BUDTA$ is closed under intersection, we then
have that they are not closed under complementation.
\begin{corollary}
 The class \BUDTA is not closed under effective complementation.
\end{corollary}

We do not know how to show that \BUDTA is not closed under complement. A
possible concrete example would be the property that on every branch of the
data tree the data values under a node of label $a$ are the same as those under
the nodes of label $b$. The complement of this property is expressible by a
\BUDTA, but it seems that a \BUDTA cannot express it.

\subsection{\bf Automata normal form}
We now present a normal form of \BUDTA{}s, removing redundancy in its
definition. This normal form simplifies the technical details in the
proof of decidability presented in the next section. We use the semigroup point
of view.

\begin{enumerate}[label=(NF\arabic*)]
\item\label{eq:semigroup-normal-form} 
The semigroup $\Monoid$ and homomorphism $h$ have the property that different
values are used for paths of length one: For all $w \in (\A\times\B)^+$ and $c \in \A\times\B$, $h(w)=h(c)$ if{f} $w=c$.

\item \label{eq:budta:normal-form:deltaeps-r}
All  transitions $(\textit{test},\textit{action})$ are such that $\textit{test}$ contains only conjuncts of the form
$p$, $\tup{\rexp}^= $, $\tup{\rexp}^{\not=}$, $\optestroot$ as well as their negated $\overline{~\cdot~}$ counterparts.
%no conjunct of the form $a$, $\bar a$, $b$, $\bar b$, $\optesteq$, $\optestneq$, $\optestleaf$ or $\optestnleaf$.

% \item \label{eq:budta:normal-form:tests-with-no-disjunction} In the
%   definition of $\delta$ of \anAut, there all transitions $(t,a)$ of $\anAut$
%   are such that $t$ is a conjunction $\bigwedge_{i \in I} t_i$ where $t_i$ is
%   a test from $\Tests$ or its negation.

% \item \label{eq:budta:normal-form:deltaup}
% In the definition of $\delta_{\text{up}}$ of \anAut, there is
% exactly one disjunct that contains exactly one conjunct. That is, for all $q\in Q$, $\delta_{\text{up}}(q)$ is undefined  or $
%  \delta_{\text{up}}(q)=p $ for some $ p\in Q$.

% \item \label{eq:budta:normal-form:deltaeps}
% For all $q\in Q$, $ \delta_{\epsilon}(q) $ is defined either as an atom, as $p \land p'$ or as $p \lor p'$ for some $p,p'\in Q$.
\end{enumerate}
An automaton $\anAut\in\budta$ is said to be in \emph{normal form} if it
satisfies \ref{eq:semigroup-normal-form},
\ref{eq:budta:normal-form:deltaeps-r}. 
%Note that only  root and path tests remain to be considered for normal form \BUDTA{}s.
%, and \ref{eq:budta:normal-form:tests-with-no-disjunction}.
%Notice that once \ref{eq:semigroup-normal-form} holds, then any test concerning a label ($a$, $\bar a$, $b$, or $\bar b$) can be simulated using tests of the form $\tup{\aMonoid}$ for some appropriate $\aMonoid$. Using similar ideas, it is not hard to check that:
\begin{proposition}\label{prop-normal-form}
  For any $\anAut \in \BUDTA$, there is an equivalent $\anAut' \in \BUDTA$ in normal form that can be effectively obtained.
\end{proposition}
\begin{proof}
First, given a finite semigroup we can easily compute another one that satisfies \ref{eq:semigroup-normal-form}, only by adding some extra elements to the domain in order to tell apart all the one letter words for each symbol of the finite alphabet.

To show that \ref{eq:budta:normal-form:deltaeps-r} can always be assumed, note
that any test for label can be simulated using $\tup{\aMonoid}$. Indeed, a test
$a$ with $a \in \A$ can be simulated with $\bigvee_{b \in \B}\tup{h(a,b)}$,
$\bar a$ with $\bigwedge_{b \in \B}\overline{\tup{h(a,b)}}$, and similar tests
can simulate $b$ and $\bar b$ for $b\in \B$.  Once this is done, a test
$\tup{\aMonoid}$ can be simulated using $\tup{\aMonoid}^{=} \lor
\tup{\aMonoid}^{\neq}$. The test $\optesteq$ can be simulated using
$\bigvee_{a \in \A \times \B} \tup{h(a)}^=$. Similarly, $\optestneq$ can be simulated
using $\tup{\aMonoid}^{\neq}$.
%Also, $\opstore$ can be simulated using $\opguess \land \optesteq$. 
Lastly, $\optestleaf$ and $\optestnleaf$ can
be tested with
$\bigwedge_{\mu \not\in\set{h(a) \mid a\in \A \times \B}} \overline{\tup{\mu}^=} \land \overline{\tup{\mu}^{\neq}}
$ and
$\bigvee_{\mu \not\in\set{h(a) \mid a\in \A \times \B}} \tup{\mu}^= \vee\tup{\mu}^{\neq}
$ respectively. Thus, we can suppose that the automaton $\anAut$ does not contain any
transition that uses tests for labels, $\optesteq$, $\optestneq$, $\optestleaf$ and $\optestnleaf$ without any loss
of generality.
\end{proof}

\subsection{Non-moving actions}
\newcommand{\anact}{\mathit{act}}

Here we introduce an extension to the automata model with non-moving
$\eps$-transitions. The automaton can now perform actions while remaining at
the same node. We show that this extension does not change the expressive power
of the model. This extension will prove useful when translating vertical
\rxpath into \BUDTA in Section~\ref{section-xpath}.

A \BUDTA automaton with $\eps$-transitions, noted $\budta^\eps$, is defined as
any \BUDTA automaton, with the only difference that its state of states $Q$ is
split into two disjoint sets $Q_\eps$ and $Q_{move}$. Moreover the set $Q$ is
partially ordered by $<$. Whenever an action $\anact(p)$ is taken where $p \in
Q_\eps$, then the automaton switches to state $p$, performs the action
$\anact$, but does not move to the parent node and stays instead at the same
node. In order to avoid infinite sequences of $\eps$-transitions we require
that when switching to a state $q$ from a state $p$ where $q\in Q_\eps$ then
$p< q$. Hence a thread can make at most $|Q_\eps|$ $\eps$-steps before moving
up in the tree.

Formally a $\budta^\eps$ is a tuple
$(\A,\B,Q_\eps,Q_{move},<,q_0,\delta,\Monoid,\morphism)$ where $Q=Q_\eps\cup Q_{move}$ is a
set partially ordered by $<$ and the set of transition $\delta$ contains only pairs
$(t,a)$ such that if $a$ contains a basic action $\anact(p)$ where $p\in
Q_\epsilon$ then $t$ contains a basic test of the form $q\in Q$ with $q<p$. The
rest is defined as for \budta.

A run $\rho$ of a $\budta^\eps$ $\anAut=(\A,\B,Q_\eps,
Q_{move},<,q_0,\delta,\Monoid,\morphism)$ on $\tT=\aA\prd\bB\prd\dD$ is defined
as for $\budta$:

\begin{itemize}
\item for any leaf $x$ of $\tT$, $\rho(x)$ is initial,
\item for any inner position $x$ of $\tT$, whose parent is the position $y$,
  and any $(p,d) \in \rho(x)$ there exists $(t,a) \in \delta$ with
  $a=\bigwedge_{j \in J} a_{j}$ such that $\tT,x,(p,d) \models t$ and for any
  $j\in J$ we have:
\begin{itemize}
\item if $a_{j}$ is $\opnothing(q)$ with $q \in Q_{move}$, then $(q,d)\in \rho(y)$,
\item if $a_{j}$ is $\opnothing(q)$ with $q \in Q_\eps$, then $(q,d)\in \rho(x)$,
\item if $a_{j}$ is $\opstore(q)$ with $q \in Q_{move}$, then $(q,\dD(x))\in \rho(y)$,
\item if $a_{j}$ is $\opstore(q)$ with $q \in Q_\eps$, then $(q,\dD(x))\in \rho(x)$,
\item if $a_{j}$ is $\opguess(q)$ with $q \in Q_{move}$, then $(q,e)\in \rho(y)$ for some $e \in \D$,
\item if $a_{j}$ is $\opguess(q)$ with $q \in Q_{\eps}$, then $(q,e)\in \rho(x)$ for some $e \in \D$,
\item if $a_{j}$ is $\opuniv(q)$ with $q \in Q_{move}$, then for all $e \in \data(\tT|_x)$,  $(q,e) \in \rho(y)$,
\item if $a_{j}$ is $\opuniv(q)$ with $q \in Q_{\eps}$, then for all $e \in \data(\tT|_x)$,  $(q,e) \in \rho(x)$. 
\end{itemize}
\end{itemize}

As expected we show that adding $\eps$-transitions does not increase the
expressive power.

\begin{proposition}\label{prop:budaeps-2-buda}
  There is an effective language-preserving translation from $\budta^\eps$ into $\budta$.
\end{proposition}
\begin{proof}
Let $\anAut$ be a $\budta^\eps$. We say that a basic action is an $\eps$-action
if it is of the form $\anact(q)$ for some $q\in Q_\eps$. The \emph{rank} of an
$\eps$-action $\anact(q)$ is the number of $p\in Q$ such that $p<q$.

We first assume without loss of generality that $\anAut$ contains no
$\eps$-action of the form $\opstore(q)$, since those can be simulated by a
$\opguess$ $\eps$-action followed by a $\optesteq$ test.

We also assume without loss of generality that all transitions $(t,a)$ of
$\anAut$ are such that $t$ contains exactly one conjunct of the form $p$ and no
conjunct of the form $\bar q$, for $p,q\in Q$.  This can be enforced using
non-determinism by always testing all the finitely many possibilities for the
states. All tests $t$ now contain exactly one test for a state $q$ and we call
$q$ \emph{the state associated to $t$}. We denote by $t\setminus q$ the new
test constructed from $t$ by removing the conjunct~$q$.

We now show how to remove the $\eps$-actions of the form $\opnothing(q)$.  We
remove them one by one computing at each step an equivalent automaton $\anAut'$
with the same set of states and the same partial order on it. $\anAut'$ is
essentially $\anAut$ with one transition removed (the one containing
$\opnothing(q)$) and several transitions added. The idea is classical and
consists in performing the actions executed by transitions associated to $q$
instead of executing $\opnothing(q)$. The added transitions may
contain new $\eps$-actions of the form $\opnothing(p)$ but those satisfy
$q<p$. Hence this process will eventually terminate. Assume that $\anAut$ has a
transition $(t,a)$ containing an $\eps$-action $\opnothing(q)$ where $q$ has
a minimal possible rank satisfying this property. We let $\hat a$ denote the action
computed from $a$ by removing the conjunct $\opnothing(q)$. $\anAut'$ is
copied from $\anAut$ with the following modifications:

We remove $(t,a)$ from the list of transitions and for each transition
$(t',a')$ of $\anAut$ such that $q$ is associated to $t'$ we add the new
transition $(t\land (t'\setminus q), \hat a \land a')$.  Notice that all the
$\eps$-actions of the form $\opnothing(p)$ occurring in these new transitions
come from $a'$ or $\hat a$ and therefore are such that $q < p$ as
desired. Notice that since the action $\opnothing$ do not change the data value
stored in the register, the test $t'\setminus q$ can be equivalently performed
before or after executing this action. Hence the new automaton is equivalent
to the old one.

We can now assume that $\anAut$ contains only $\eps$-actions of the form
$\opguess(q)$ and $\opuniv(q)$.

We next show how to remove the $\eps$-actions of the form $\opguess(q)$.  We
again remove them one by one computing an equivalent automaton $\anAut'$. This
automaton has new states but those do not modify the rank of the $\eps$-actions
of the form $\opguess(p)$. As before it is essentially $\anAut$ with one
transition removed and several new transitions. The idea is to perform the
guessing of the data value at one of the children of the current node, while
moving up in the tree. The new transitions needed to do this may introduce new
$\eps$-actions of the form $\opguess(p)$, the rank of $p$ being strictly
smaller than the rank of $q$.  Hence this process will eventually
terminate. Assume that $\anAut$ has a transition $(t,a)$ containing an
$\eps$-action $\opguess(q)$ where $q$ has a maximal possible rank satisfying
this property. Let $p$ be the state associated to $t$ (hence $p<q$) and $\hat
a$ be the action computed from $a$ by removing the conjunct $\opguess(q)$.
$\anAut'$ is copied from $\anAut$ with the following modifications:

We add two new states $q'$ and $q''$. If $p\in Q_{move}$ then so are $q'$ and
$q''$. If $p\in Q_{\eps}$ then so are $q'$ and $q''$. We modify $<$ by setting
$r<q'$ and $r<q''$ whenever $r<p$ and $q'<r$ and $q''<r$ whenever $p<r$. In
other words $q'$ and $q''$ play the same role as $p$ in the partial order but
are incomparable with $p$. Hence the rank of any state of $Q$ is not modified
by the addition of these two states. Moreover the ranks of $q'$ and $q''$ match the rank
of $p$ and are therefore strictly less than the rank of $q$.

We now update the list of transitions by first removing $(t,a)$. Then, consider
a transition $(t',a')$ of $\anAut$ of state $p'$ and such that $a'$ contains an
action of the form $\anact(p)$. We add a new transition $(t',a'')$ where
$a''$ is the action constructed from $a'$ by removing $\anact(p)$ and
adding $\anact(q')$ and $\opguess(q'')$ (this is consistent with the partial
order). Finally we add in $\anAut'$ a transition $(q'\land (t\setminus p), \hat
a)$ and a transition $(q''\land (t''\setminus q),b)$ for any transition
$(t'',b)$ such that $p$ is the state associated to $t''$, both being consistent
with the partial order.

The reader can now verify that $\anAut'$ is equivalent to $\anAut$.  If right
after switching to state $p$ using the transition $(t',a')$, where $\anact(p)$
is a basic action of $a'$, $\anAut$ decides to use the transition $(t,a)$ then
$\anAut'$ can simulate this as follows. First it uses the transition
$(t',a'')$. This generates two new basic actions, $\anact(q')$ and
$\opguess(q'')$ instead of $\anact(p)$. The thread generated by $\anact(q')$
will perform the tests and actions that $\anAut$ did by using the transition
$(t,a)$, except for $\opguess(q)$. But this latter action is simulated by
$\opguess(q'')$ that has been launched earlier. Hence $\anAut'$ does simulate
$\anAut$ also in this case. The other direction, showing that a run of
$\anAut'$ can be simulated by a run of $\anAut$ is proved similarly.

We can now assume that $\anAut$ contains only $\eps$-actions of the form
$\opuniv(q)$. Removing those requires more care. The idea is, as above, to
perform the action at the children of the current node while moving up in the
tree. As this is a universal move, all children are concerned and we therefore
need some synchronization. The internal alphabet will be used for this. In a
nutshell the new automaton will mark the nodes where the $\opuniv(q)$ should be
executed and mark all their children. Simple threads are launched at the leaves
making sure the marking is consistent. When encountering a node marked as a
child, $\opuniv(q)$ can be executed while moving up in the tree,
taking care of all the data values of the corresponding subtree. An extra
thread need to be executed at the parent node to take care of its data
value. Note that different tests can be launched in different threads by the use of actions like $\opnothing(q) \land \opnothing(p)$ with $p,q \in Q_\eps$.
 We now turn to the details.

Assume that $\anAut$ has a transition $(t,a)$ containing an $\eps$-action
$\opuniv(q)$ where $q$ has a maximal possible rank satisfying this
property. Let $p$ be the state associated to $t$ and $\hat a$ be the action
computed from $a$ by removing the conjunct $\opuniv(q)$. $\anAut'$ is
essentially a copy of $\anAut$ with the following modifications:

We use $\B'=\B\times\set{0,1,2}$ as internal alphabet for $\anAut'$. The nodes
containing $1$ in their label are expected to simulate an action $\opuniv(q)$,
and their children must contain $2$ in their label. On top of simulating $\anAut$,
$\anAut'$ starts a new thread at all the leaves of the tree, using a fresh new
state in $Q'_{move}$. These threads make sure that all nodes whose label
contain $2$ have a parent with $1$ in their label and that all nodes having $1$
in their label have no child without $2$ in their label (this can be done with
the appropriate $\langle \rexp \rangle$-test). Moreover, we assume yet a new
state $q'$ in $Q'_{move}$ and these new threads trigger an action $\opuniv(q')$
each time the symbol $2$ is found. In other words, we make sure that any node
$x$ with internal label $1$ has a thread $(q',d)$ for every data value
occurring strictly below $x$. We replace the transition $(t,a)$ with
$(t^1,\opstore(q) \land \hat a)$ where $t^1$ performs the same tests as $t$ but
also checks that the internal label contains $1$. Finally, to any transition
$(t',a')$ of $\anAut$ whose associated test is $q$ we add a new transition
replacing $q$ with $q'$ (note that we don't remove any transition, hence
$(t',a')$ is still a transition of $\anAut'$).

Note that this generates a new $\eps$-action of the form $\opstore(q)$. However
we have seen that those could be simulated using \opguess $\eps$-actions
and we have seen above that such $\eps$-actions could be removed. As this last step
does not introduce any new \opuniv $\eps$-actions, the resulting
automaton has strictly less $\eps$-actions of that rank.

For correctness, assume that $\anAut$ accepts a data tree $\tT$ with an accepting run $\rho$. $\anAut'$ marks all positions where $\rho$ uses transition $(t,a)$, with $1$, and the children of those nodes with $2$. The rest of the simulation of $\anAut$ is trivial. The run will be accepting because all the new generated threads will have exactly the same effect as the action $\opuniv(q)$ at the marked nodes: the action $\opstore(q)$ generates a thread with the current data value, while the actions $\opuniv(q')$ generate a thread for all the data values within the subtrees. Conversely, assume that $\anAut'$ has an accepting run $\rho'$ on a data tree $\tT$. Let $x$ be a node whose internal label has $1$. If $\rho'$ does use the transition $(t^1,\opstore(q) \land \hat a)$ at node $x$ then \anAut can simulate $\anAut'$ using the transition $(t,a)$ as they produce the same threads. If $\rho'$ does not use the transition $(t^1,\opstore(q) \land \hat a)$ at that node, then the initial internal labeling was incorrectly guessed. But this is not important as $\anAut$ can still simulate $\anAut'$ on a subset of the threads, hence will also accept $\tT$. In other words, $\anAut'$ will have executed an action $\opuniv(p)$ at $x$ that was not necessary for accepting the tree.
\end{proof}

%%% Local Variables: 
%%% TeX-PDF-mode: t
%%% TeX-master: "main"
%%% End: 

 \newcommand{\anWSTS}{\ensuremath{\cl{W}}\xspace}
\section{The emptiness problem for \BUDTA}
\label{sec:wsts-budta}\label{sec:prelim:wsts}
\newcommand{\eqconf}{\sim}
\newcommand{\mult}{\cdot}

This is the most technical section of the paper. Its goal is to show:

\begin{theorem}\label{thm-budta-decid}
The emptiness problem for \BUDTA is decidable.  
\end{theorem}

This result is shown through methods from the theory of \emph{well-structured
  transition systems}, or \wsts for short~\cite{FS01}. It is obtained by
interpreting the execution of a \BUDTA using a transition system compatible
with some well-quasi-ordering (\wqo). We start with the notions of the theory
of \wsts that we will need.

\bigskip

\newcommand{\sDom}{\cl{A}}
\newcommand{\sDomb}{\cl{B}}
\newcommand{\sDomc}{\cl{C}}
\newcommand{\sDomd}{\cl{D}}
\newcommand{\upclose}{\uparrow\!}
\newcommand{\downclose}{\downarrow\!}

A quasi-order $\le$ (\ie, a reflexive and transitive relation) over a set
$\aSet$ is said to be a \emph{well-quasi-order} (\wqo) if there is no
infinite decreasing sequence and no infinite incomparabe sequence of elements
of $S$. In orther words  if for every
infinite sequence $s_1\,s_2\,\dotsb \in \aSet^\omega$ there are two indices
$i<j$ such that $s_i \leq s_j$.
  Given a \wqo $(\aSet,\leq)$ and $\aSetb \subseteq \aSet$, we define the
  \emph{downward closure} of $\aSetb$ as $\downclose \aSetb \coloneqq \set{ s \in \aSet
    \mid \exists\, t \in \aSetb, s \leq t}$ and $\aSetb$ is
  \emph{downward closed} if $\downclose\aSetb=\aSetb$.

%Given a transition system $(\aSet,\rightarrow)$, and $\aSetb \subseteq \aSet$
%we define $Succ(\aSetb) \coloneqq \set{s\in\aSet \mid \exists\, t \in \aSetb
%  \text{ with } t \rightarrow s}$, and $Succ^*$ as its reflexive-transitive
%closure.   We say that $(\aSet,\rightarrow)$ is \emph{finitely
%    branching} if{f} $Succ(\set s )$ is finite for all $s \in \aSet$. If
%  $Succ(\set s)$ is also effectively computable for all $s$, we say that
%  $(\aSet,\rightarrow)$ is \emph{effective}.

% We adapt some results and definitions of \cite[\S 5]{FS01} of what is there
% called \emph{reflexive compatibility} (\ie, compatibility in zero or one
% steps) to extend them to $N$-compatibility (\ie, compatibility in at most
% $N$ steps).
% %
% Given a binary relation $R \subseteq \aSet \times \aSet$ and $\aSetk \subseteq
% \aSet$, let us write $R^{\leq n}$ for $\textit{Id} \cup R \cup
% R^2 \cup \dotsb \cup R^n$, where $\textit{Id}$ is the identity relation and $R^i$ is the $i$-fold composition of $R$.
%   Given $N\in\Np$, a transition system $(\aSet,\rightarrow)$ is
%   \emph{$N$-downward compatible} with respect to a \wqo $(\aSet,\leq)$ if{f}
%   for every $s_1,s_2,s'_1 \in \aSet$ such that $s'_1 \leq s_1$ and $s_1
%   \rightarrow s_2$, there exists $s'_2 \in \aSet$ such that $s'_2 \leq s_2$ and
%   $s'_1 \mathrel{(\rightarrow)^{\leq N}} s'_2$. In some sense any behavior from
%   the bigger element $s_1$ can be simulated by the smaller element $s'_1$. In
%   truth, $s'_1$ may need several transitions, but not more than $N$. 

\bigskip

Given a \BUDTA \anAut, a first goal could be to compute the set of reachable
configurations, i.e. the configurations $C$ such that there is a data tree
$\tT$ and a run $\rho$ of \anAut on $\tT$ such that $\rho(x)=C$ for the root
$x$ of $\tT$.

A na\"ive algorithm would be to start with the set of initial configurations
and then to enrich this set, step by step, by applying the transition function
of $\anAut$ to some configurations from the set, hence preserving reachability.

This idea immediately raises some issues. The first one is that the set of
reachable configurations is infinite and therefore we need a way to guarantee
that the run of the algorithm will eventually stop. Second, we need to make
sure we can compute the configurations reachable in one step from a finite set
of configurations. In particular, as our trees are unranked, we do not know in
advance how many reachable configurations we need to combine in order to derive
the next ones. Moreover, in order to apply the transition function of \anAut,
we need to know which of the tests of the form $\tup{\aMonoid}^=$ and
$\tup{\aMonoid}^{\neq}$ are true in the data trees that make a configuration
reachable.

To overcome this last issue, we enrich the notion of configuration, initially a
finite set of pairs $(q,d)\in Q\times\D$, by including all the information of
the current subtree that is necessary to maintain in order to continue the
simulation of the automaton from there. This information consists in a finite
set of pairs of the form $(\aMonoid,d) \in \Monoid\times\D$ whose presence
indicate that from the current node the data value $d$ can be reached following
a downward path evaluating to $\aMonoid$ via the homomorphism $\morphism$ of
\anAut. This enriched configuration will be henceforth called \emph{extended
  configuration}.

To overcome the problem of termination we introduce a well-quasi-order on
extended configurations and we show that, as far as emptiness is concerned, it
is not necessary to consider extended configurations that are bigger than those
already computed. We say that such a \wqo is \emph{compatible} with the transition system.\footnote{In our case this will be assured by a property of the transition sytem $(X,\rightarrow)$ and \wqo $(X,\leq)$ of the following form: for every $x,x',y \in X$ so that $x' \leq x$ and $x \rightarrow y$ there is $y'\in X$ so that $x' \rightarrow^* y'$ and $y' \leq y$.}
 In other words it is enough to solve the \emph{coverability} problem relative to that \wqo instead of the \emph{reachability problem}. Coverability can be decided by
inserting a new extended configuration into the running set of reachable
configurations only if the new extended configuration is incomparable to or
smaller than those already present. Termination is then guaranteed by
the well-quasi-order.

Finally, in order to overcome the problem coming from the tree unrankedness, we
decompose a transition of \anAut into two basic steps of a transition system (see
Figure~\ref{fig:root-merge-tran}). The first step simulates a transition of
\anAut assuming the current node is the only
child. The second step ``merges'' two configurations by simulating what \anAut
would have done if the roots of the corresponding trees were
identified. Repeating this last operation yields trees of arbitrary rank.
Altogether the unrankedness problem is also transferred to the coverability
problem.

To summarize, we associate to a given \BUDTA \anAut a transition system based on
the configurations of \anAut and the transition relation of \anAut, we
exhibit a compatible \wqo, and can then decide emptiness of \anAut using
the coverability algorithm sketched above.

%Luc: revoir la figure?
\begin{figure}
  \centering \includegraphics[scale=.55]{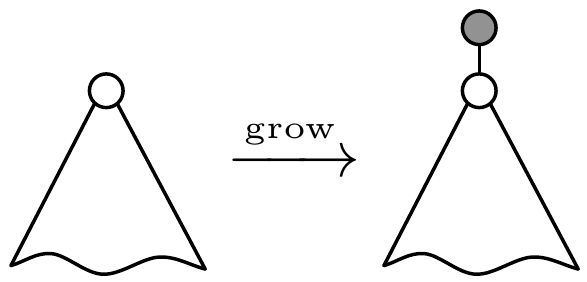}\hspace{2cm}
\includegraphics[scale=.6]{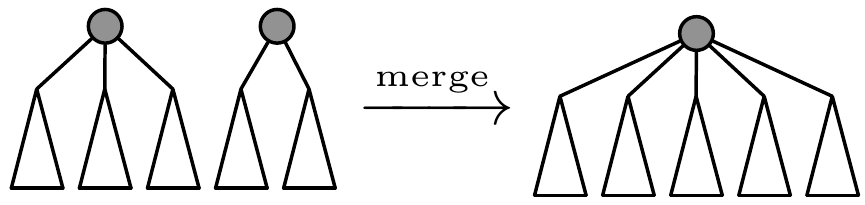}
  \caption{The \emph{grow} and \emph{merge} operations.}
  \label{fig:root-merge-tran}
\end{figure}
%

%In the sequel we implicitly assume that all our \BUDTA are in normal form.

\subsection{Extended configurations}\label{section-abstract-configuration}
\newcommand{\ac}{\ensuremath{\textup{EC}}\xspace}
\newcommand{\AC}{\ensuremath{\subsetsf{(\ac)}}\xspace}
\newcommand{\Datafun}{\ensuremath{\Delta}\xspace}
\newcommand{\prof}{\ensuremath{\chi}\xspace}
We define here the set of \emph{extended configurations} of a given \BUDTA
$\anAut=(\A,\B,Q,q_0,\delta,\Monoid,\morphism)$ in normal form.

\newcommand{\topg}{\top_{\!g}} \newcommand{\topng}{\top} In order to get a good
intuition about the definition of extended configurations it is necessary to
have a glimpse of what transitions will be performed on them. We
will simulate each transition $\tau\in\delta$ by several steps of the transition system. Recall that each
$\tau\in\delta$ consists of several tests and several actions. The first
step of the simulation of $\tau$ will generate as many threads as there are
tests and actions in $\tau$. Each of these threads will then have the task of
performing the corresponding test or action. We will also use symbols,
$\bot$, $\topng$ and $\topg$, to distinguish respectively: threads where no transition has
yet been executed, threads on which a transition (other than $\opguess$) has been successfully
applied, and threads on which a $\opguess$ transition has been successfully applied.

\newcommand{\testactionsbot}{\textrm{TA}}%
Let $\testactionsbot = \Tests\cup\Actions\cup\set{\topng,\topg,\bot}$.
An \emph{extended configuration} of \anAut is a tuple $(\Delta,\Gamma,r,m)$
where $r$ and $m$ are either true or false, $\Delta$ is a finite subset of
$Q\times\testactionsbot \times \D$
and $\Gamma$ is
a finite subset of $\Monoid\times\D$ such that
\begin{align}\label{eq:condition-Gamma-well-formed}
\tag{$\star$}
\text{$\Gamma$ contains exactly \emph{one} pair of the form $(h(c),d)$ with
  $c\in\A\times\B$.}
\end{align}
This unique element of $\A\times\B$ is denoted as the \emph{label} of the
extended configuration and the unique associated data value is denoted as the
\emph{data value} of the extended configuration.

Intuitively, $r$ says whether the current node should be treated as the root or
not, $m$ says whether the extended configuration is ready to be merged with
another one or not, and $\DataState$ represents the set of ongoing threads at the current node.  Its
elements have the form $(q,\alpha,d)$. If $\alpha \in \Tests \cup \Actions$
then this thread is expected to perform the corresponding test or action.  If
$\alpha \in \set{\topng,\topg}$ then this thread is ready to move up in the
tree.  If $\alpha= \bot$ then no transition has yet been applied on this
thread. On the other hand, a pair $(\mu,d) \in \Gamma$ simulates the existence of a downward path
evaluating to $\mu$ and whose last node carries the data value $d$. The condition
\eqref{eq:condition-Gamma-well-formed} is here for technical reasons. In
particular it permits to recover the data value and label of the current tree root.

\newcommand{\typeof}[1]{[#1]}
Consider $\theta =
(\DataState,\DataMonoid,r,m)$. %and $\prof\subseteq \Monoid$. 
Remember that $\DataState$ is the set of threads, $\DataMonoid$ is the set of
pairs ``(path evaluation, data value)'' present in the abstracted subtree, $r$
states whether the extended configuration is the root, and $m$ whether it is in
merge mode. In the sequel we will use the following notation for $\theta =
(\DataState,\DataMonoid,r,m)$:

\begin{align*}
\DataState(d) &= \set{(q,\alpha) \mid (q,\alpha,d)\in\DataState},\\
\DataMonoid(d) &= \set{\aMonoid \mid (\aMonoid,d)\in\DataMonoid},\\
\typeof \theta (d) &=(\DataState(d), \DataMonoid(d)),\\
\data(\theta) & = \set{d \mid (q,\alpha,d)\in\DataState \lor (\aMonoid,d)\in\DataMonoid}.
%\DataMonoid^{-1}(\prof) &= \set{d  \mid  \DataMonoid(d)=\prof}.
\end{align*}

We will also use the inverses of these, for $(R,\prof)\in \subsets(Q\times \testactionsbot) \times
\subsets(\Monoid)$:
\begin{align*}
\DataState^{-1}(R) &= \set{d \mid \DataState(d)=R},\\
\DataMonoid^{-1}(\prof) &= \set{d \mid \DataMonoid(d)=\prof},\\
\typeof \theta^{-1} (R,\prof) &= \set{d \mid \typeof \theta(d)=(R,\prof)},\\
\end{align*}

Note that $\DataState(d)$ gives the information about the current threads carrying $d$ in the register; $\DataMonoid(d)$ gives the information about downward paths that lead to the data value $d$; $\typeof \theta (d)$ is the aggregation of this information; $\data(\theta)$ is the set of all data values present in $\theta$; $\DataState^{-1}(R)$ is the set of data values whose thread information is precisely $R$; $\DataMonoid^{-1}(\prof)$ is the set of data values whose downward paths information is precisely $\prof$; and $\typeof \theta^{-1} (R,\prof)$ is the set of data values whose aggregated $\typeof \theta$-information is precisely the pair $(R,\prof)$.

We use the letter $\theta$ to denote an extended configuration and we write \ac to denote the set of all extended configurations. Similarly, we use $\Theta$ to denote a finite set of extended configurations.% and \AC for the set of finite sets of extended configurations.

%Given a data value $d$ and an extended configuration $\theta$, $\typeof \theta (d)$ is also denoted as the \emph{type of $d$ in $\theta$}.
An extended configuration $\theta=(\DataState,\DataMonoid,r,m)$ is said to be
\emph{initial} if it corresponds to a leaf node, i.e., is such that
$\DataState = \set{(q_0,\bot,d)}$ and $\DataMonoid = \set{(\morphism(a),d')}$ for
some $d,d' \in \D$ and $a \in \A\times\B$. Note that every initial extended configuration verifies condition~\eqref{eq:condition-Gamma-well-formed}.  An extended configuration is said to be \emph{accepting} if
$\DataState = \emptyset$.

We will very often use this notation for components of extended configuations:
\begin{align}\label{eq:notation}
  \begin{split}
    \theta &= (\DataState, \DataMonoid,r,m)\\
    \theta' &= (\DataState', \DataMonoid',r',m')
  \end{split}\tag{$\ddag$}
\end{align}

Two extended configurations $\theta_1$ and $\theta_2$ are said to be \emph{equivalent}
if they are identical modulo a bijection between data values, i.e. there is a
bijection $f: \D \to \D$ such that $f(\theta_1)=\theta_2$ (with some abuse of
notation). We denote this by $\theta_1 \eqconf \theta_2$.

Finally, we write $\Theta_I$ to denote the set of all initial extended
configurations modulo $\eqconf$ (i.e. a set containing at most one element for
each $\eqconf$ equivalence class). Note that $\Theta_I$ is \emph{finite} and
can be computed. A set of extended configurations is said to be
\emph{accepting} if it contains an accepting extended configuration.

\subsection{Well-quasi-orders}\label{sec-wqo}
\newcommand{\lqac}{\precsim}
\newcommand{\lqacstr}{\preceq}
\newcommand{\lqAC}{\leq_{\textit{min}}}
\newcommand{\profile}{\mathit{profile}}

We now equip \ac with a well-quasi-order $(\ac,\lqacstr)$. 

The \emph{profile of an extended configuration}
$\theta=(\DataState,\DataMonoid,r,m)$, denoted by $\profile(\theta)$, is 
the tuple $(A_0, A_1,r,m)$ with
\begin{align*}
A_0 &= \set{\prof \subseteq \Monoid :
  |\DataMonoid^{-1}(\prof)|=0},\\
A_1 &= \set{\prof \subseteq \Monoid :
  |\DataMonoid^{-1}(\prof)|=1}.
\end{align*}
% 
%We first define the quasi-order $\lqacstr$ over extended configurations, and
%then we define the order $(\ac,\lqac)$ as $(\ac,\lqacstr)$ modulo $\eqconf$.

The idea behind the definition of $A_0$ and $A_1$ is that the automata model can only test whether there exists either (i) none, (ii) exactly one, or (iii) more than one data values that are reachable by a downward path $\mu$. If it is zero, we store this information in $A_0$, if it is one we store it into $A_1$. However an automaton cannot count more than this. And this is a key
for decidability.

%\begin{definition}[$\lqacstr$]
  Given two extended configurations $\theta_1=(\DataState_1,\DataMonoid_1,r_1,m_1)$
  and $\theta_2=(\DataState_2,\DataMonoid_2,r_2,m_2)$, we denote by $\theta_1  \lqacstr \theta_2$ the fact that
  \begin{itemize}
    \item $\profile(\theta_1)=\profile(\theta_2)$, and
    \item there is a function $f : \D \to \D$ such that $\data(\theta_1)
      \subseteq f(\D)$ (henceforth called the \emph{surjectivity condition}) and for all $d \in \D$:
      \begin{itemize}
      \item $\DataState_1(f(d)) \subseteq \DataState_2(d)$,
      \item $\DataMonoid_1(f(d)) = \DataMonoid_2(d)$.
      \end{itemize}
%     and for all $d \in \data(\theta_1)$, $d \in f^{-1}(d)$.
  \end{itemize}
We write $\theta_1 \lqacstr_f \theta_2$ if we want to make explicit the
function $f$ witnessing the relation. 
%\end{definition}
%
\begin{remark}\label{remark-order}
  Assume that $\theta_1 \lqacstr_f \theta_2$ and consider a thread $(q,\alpha,d)\in \DataState_1$. Then by the surjectivity condition of $f$ there is a data value $e$ such that $f(e)=d$. From $\DataState_1(f(e)) \subseteq \DataState_2(e)$ it follows that $(q,\alpha,e)$ is a thread in $\DataState_2$. Similarly if $\theta_1$ contains $k$ threads $(q,\alpha,d_1), \dotsc, (q,\alpha,d_k)$ then so would $\theta_2$. The converse is however not true as $f$ need not be injective. In fact, the equality between the profiles is here to guarantee a minimal version of the converse: if for some $\prof \subseteq \Monoid$, $\theta_2$ contains no data value in $\DataMonoid_2(\prof)$, then neither does $\theta_1$, and if $\DataMonoid_2(\prof)$ contains only one data value, so does $\theta_1$. We will use these remarks implicitly in the sequel.
\end{remark}

\begin{remark}\label{remark-label}
  Notice that from the previous remark and
  condition~\eqref{eq:condition-Gamma-well-formed}, it follows that if
  $\theta_1 \lqacstr_f \theta_2$ then $\theta_1$ and $\theta_2$ have
  the same label. Moreover if $d$ is the data value of $\theta_2$ then $f(d)$
  is the data value of $\theta_1$.
\end{remark}

\begin{remark}
 Notice that if $\theta_1$ and $\theta_2$ are equivalent ($\theta_1 \eqconf \theta_2$), then $\theta_1 \lqacstr_f \theta_2$
 where $f$ is the function witnessing the equivalence.
\end{remark}
%We now define $\lqac$ as: $\theta_1 \lqac \theta_2$ if{f} $\theta'_1\lqacstr \theta_2$
%  for some $\theta'_1 \eqconf \theta_1$.

The following lemmas are key observations:
\begin{lemma}\label{lem:wqo-subsets-abstract-config}
 $(\ac,\lqacstr)$ is a \wqo.
\end{lemma}

This is a simple consequence of Dickson's Lemma, that we state next.
Let $\leq_k$ be the componentwise order of vectors of natural numbers of dimension $k$. That is, $(x_1, \dotsc, x_k) \leq_k (y_1, \dotsc, y_k)$ if $x_i \leq y_i$ for all $i \in [k]$.
\begin{lemma}[Dickson's Lemma \cite{dicksonslem}]\label{lem:dickson}
For every $k$,   $(\N^k,\leq_k)$ is a \wqo.
\end{lemma}
\begin{proof}[Proof of Lemma~\ref{lem:wqo-subsets-abstract-config}]
Note that $\theta \lqacstr_\iota \theta$ for all $\theta \in \ac$ through the identity function $\iota : \D \to \D$, and  that $\theta_1 \lqacstr_{f \circ g} \theta_3$ assuming $\theta_1 \lqacstr_f \theta_2$ and $\theta_2 \lqacstr_g \theta_3$. Thus, $\lqacstr$ is a quasi-ordering. We next show that in fact it is a \wqo.

\newcommand{\eqprofile}{\approx_\textit{prof}}%
\newcommand{\lqvector}{\leq_{\bar x}}%
Given an extended configuration $\theta = (\DataState,\DataMonoid,r,m)$, let us
define $\bar x ( \theta)$ as a function from $\subsets(Q\times \testactionsbot) \times
\subsets(\Monoid)$ to $\N$. We define $\bar x (\theta) (R,\prof) =
|\typeof\theta^{-1}(R,\prof)|$. Notice that $\bar x (\theta)$ can be seen as a
vector of natural numbers of dimension $k=|\subsets(Q\times\testactionsbot) \times
\subsets(\Monoid)|$. It then follows that for every $\theta_1,\theta_2$, if
\[\set{(R,\prof) \in \subsets(Q\times\testactionsbot)\times \subsets( \Monoid) :
  |\typeof{\theta_1}^{-1}(R,\prof)|=i}\]
is equal to \[\set{(R,\prof) \in \subsets(Q\times\testactionsbot)\times \subsets( \Monoid) :
  |\typeof{\theta_2}^{-1}(R,\prof)|=i}\] for every $i \in \set{0,1}$, and $\bar
x (\theta_1) \leq_k \bar x (\theta_2)$, then
$\profile(\theta_1)=\profile(\theta_2)$ and further $\theta_1 \lqacstr
\theta_2$. 

Consider an infinite sequence $(\theta_i)_{i \in \N}$. As there
are only finitely many possible values for $\profile$, there must be an infinite
subsequence $(\theta'_i)_{i \in \N}$ of $(\theta_i)_{i \in \N}$ so that
$\profile(\theta'_i) = \profile(\theta'_j)$ for all $i \neq j$.  By Dickson's
Lemma (Lemma~\ref{lem:dickson}), there are indices $i < j$ so that $\bar x
(\theta'_i) \leq_k \bar x (\theta'_j)$.  Hence, there are indices $i < j$ so that
$\profile(\theta_i) = \profile(\theta_j)$ and $\bar x (\theta_i) \leq_k \bar x
(\theta_j)$.
This means that $|\typeof{\theta_i}^{-1}(R,\prof)| \leq |\typeof{\theta_j}^{-1}(R,\prof)|$ for every $R,\prof$. Let $f : \data(\theta_i) \to \data(\theta_j)$ be so that $f(\typeof{\theta_j}^{-1}(R,\prof)) = \typeof{\theta_i}^{-1}(R,\prof)$ for every $R,\prof$. Note that $\DataState_i(f(d)) = \DataState_j(d)$ and $\DataMonoid_i(f(d)) = \DataMonoid_j(d)$ for every $d$. Thus, $\theta_i \lqacstr_f \theta_j$, and hence $(\ac,\lqacstr)$ is
a $\wqo$.
\end{proof}

Finally we will need the fact that the set of accepting extended configurations
is downward closed.

\begin{lemma}\label{lemma-downward-closed}
If $\theta \in\ac$ is accepting and $\theta'\in\ac$ is such that $\theta' \lqacstr \theta$ then $\theta'$ is
accepting.
\end{lemma}
\begin{proof}
  If $\theta' \lqacstr_f \theta$ then, by the surjectivity condition of $f$, we have that for all $d\in\data(\theta')$, $\DataState'(d) \subseteq \DataState(d')$ for some $d'$ such that $f(d')=d$. As $\theta$ is accepting this implies that $\DataState(d') = \emptyset$ for all $d' \in \D$. Hence $\theta'$ is also accepting.
\end{proof}

\subsection{Transition system}\label{sec-trans}
\newcommand\rightarrowAconf[1]{\ensuremath{\xrightarrow{\text{\tiny #1}}}\xspace}
\newcommand\RightarrowAconf[1]{\ensuremath{\xLongrightarrow{\text{\tiny #1}}}\xspace}
\newcommand\transup{\rightarrowAconf{\textup{grow}}}
\newcommand\transmerge{\rightarrowAconf{\textup{merge}}}
\newcommand\transdup{\rightarrowAconf{$\opdup$}}
\newcommand\transuniv{\rightarrowAconf{$\opuniv$}}
\newcommand\transdelta{\rightarrowAconf{$\delta$}}
\newcommand\transpump{\rightarrowAconf{\textup{inc}}}
\newcommand\transpumpSX{\rightarrowAconf{inc$(S,\prof)$}}
\newcommand\transpumpSpX{\rightarrowAconf{inc$(S',\prof)$}}
\newcommand\transpumpX{\rightarrowAconf{inc$(\prof)$}}
\newcommand\transempty{\rightarrowAconf{$\opdisjoin$}}
\newcommand\transunique{\rightarrowAconf{$\opunique$}}
\newcommand\transdown{\rightarrowAconf{$\tup{\aMonoid}$}}
\newcommand\transdownn{\rightarrowAconf{$\overline{\tup{\aMonoid}}$}}
\newcommand\transdowneq{\rightarrowAconf{$\tup{\aMonoid}^{=}$}}
\newcommand\transdowneqn{\rightarrowAconf{$\overline{\tup{\aMonoid}^{=}}$}}
\newcommand\transdownneq{\rightarrowAconf{$\tup{\aMonoid}^{\not=}$}}
\newcommand\transdownneqn{\rightarrowAconf{$\overline{\tup{\aMonoid}^{\not=}}$}}
\newcommand\transtestq{\rightarrowAconf{$q$}}
\newcommand\transtestnq{\rightarrowAconf{$\overline q$}}
\newcommand\transtesteq{\rightarrowAconf{$\optesteq$}}
\newcommand\transtestneq{\rightarrowAconf{$\optestneq$}}
\newcommand\transtestroot{\rightarrowAconf{$\optestroot$}}
\newcommand\transtestnroot{\rightarrowAconf{$\optestnroot$}}
\newcommand\transguess{\rightarrowAconf{$\opguess$}}
\newcommand\transstore{\rightarrowAconf{$\opstore$}}
\newcommand\transkeep{\rightarrowAconf{$\opnothing$}}
\newcommand\transaccept{\rightarrowAconf{$\opaccept$}}
\newcommand\transand{\rightarrowAconf{$\land$}}
\newcommand\transor{\rightarrowAconf{$\lor$}}
%
% \newcommand\Transup{\RightarrowAconf{root}}
% \newcommand\Transmerge{\RightarrowAconf{merge}}
% \newcommand\Transdup{\RightarrowAconf{dup}}
% \newcommand\Transpump{\RightarrowAconf{pump}}
% \newcommand\Transempty{\RightarrowAconf{$\opdisjoin$}}
% \newcommand\Transunique{\RightarrowAconf{$\opunique$}}
% \newcommand\Transdown{\RightarrowAconf{$\tup{\aMonoid}$}}
% \newcommand\Transdownn{\RightarrowAconf{$\overline{\tup{\aMonoid}}$}}
% \newcommand\Transdowneq{\RightarrowAconf{$\tup{\aMonoid}^{=}$}}
% \newcommand\Transdownneq{\RightarrowAconf{$\tup{\aMonoid}^{\not=}$}}
% \newcommand\Transtesteq{\RightarrowAconf{$\optesteq$}}
% \newcommand\Transtestneq{\RightarrowAconf{$\optestneq$}}
% \newcommand\Transguess{\RightarrowAconf{$\opguess$}}
% \newcommand\Transand{\RightarrowAconf{$\land$}}
% \newcommand\Transor{\RightarrowAconf{$\lor$}}
% %
% \newcommand\Transeps{\RightarrowAconf{$\varepsilon$}}
%
\newcommand{\tranaconf}{\rightarrow}
\newcommand{\tranaconfe}{\rightarrow_\epsilon}
\newcommand{\tranAC}{\mathrel{\Rightarrow}}
\newcommand{\NonIncData}{I}

We now equip \ac with a transition relation $\tranaconf$ that reflects the
transition function $\delta$ of \anAut. As mentioned earlier, we decompose each
transition of \anAut into several basic steps. The first kind can be viewed as
$\epsilon$-transitions, each of them concerning a single thread, which perform
a basic test or action but without moving up in the tree. When all threads of
an extended configuration have performed their $\epsilon$-transition, the
extended configuration switches to a merging state and is ready for the next
step. The second kind merges several extended configurations with the same label
and data value that are in a merge state (\ie, $m = \textit{true}$) in order to
combine their immediate subtrees, this operation is merely a union. A third
kind concludes the simulation of the transition by ``adding a root'' to an
extended configuration over a data tree.

We will denote by \transup the third kind of transition, by \transmerge, the
second kind, and by $\tranaconfe$ the union of all the
$\epsilon$-transitions. In order to obtain our compatibility result with
respect to the partial order, we will need one
extra transition $\transpump$ that makes our trees in some sense ``fatter'', by
duplicating immediate subtrees of the root, and whose purpose will be essential for showing compatibility for the $\transmerge$ transition (Lemma~\ref{lem:monotonic-merge}). The union of all these transitions will form the transition relation
$\tranaconf$ of \ac.

We start with the description for $\epsilon$-transitions. Those are defined in
a straightforward way in order to simulate the tests and actions of the
initial automaton.

Given two extended configurations $\theta_1=(\Delta_1,\Gamma_1,r_1, m_1)$ and
$\theta_2=(\Delta_2,\Gamma_2,r_2,m_2)$, we say that $\theta_1 \tranaconfe
\theta_2$ if $m_1=m_2=\textit{false}$ (the merge information is used for
simulating an \up-transition as will be explained later), $r_2=r_1$ (whether
the current node is the root or not should not change), $\theta_1$ and
$\theta_2$ have the same label and data value, $\DataMonoid_2=\DataMonoid_1$
and, furthermore, one of the following  holds:

\newcounter{transcounter}
\begin{enumerate}[leftmargin=8mm]
\item $\theta_1 \transdelta \theta_2$. This transition can happen if there is
  $(q,\bot,d)\in\DataState_1$ for some state $q\in Q$ and data $d\in\D$. In this case, for
  some transition $\tau \in \delta$, $\theta_2$ is such that: $\DataState_2 =
  (\DataState_1 \setminus \set{(q,\bot,d)}) \cup \set{(q,\alpha_i,d) : i\leq
    |\tau|}$, where $(\alpha_i)_{i \leq |\tau|}$ are all the tests and actions
  occurring in $\tau$.

\item $\theta_1 \transuniv \theta_2$. This transition can happen if there is
  $(q,\opuniv(p),d)\in\DataState_1$ for some states $p,q \in Q$ and data $d\in\D$. In this
  case $\theta_2$ is such that $\DataState_2 = (\DataState_1 \setminus
  \set{(q,\opuniv(p),d)}) \cup \set{(p,\topng,e) : \exists \aMonoid ~.~
    (\aMonoid,e) \in \DataMonoid_1} $.

\item $\theta_1 \transguess \theta_2$.  This transition can happen if there is
  $(q,\opguess(p),d) \in\DataState_1$ for some states $p,q \in Q$ and data $d\in\D$.  In
  this case $\theta_2$ is such that $\DataState_2=(\DataState_1 \setminus
  \set{(q,\opguess(p),d)}) \cup \set{(p,\topg, d')}$ for some $d' \in \D$.
%, where $a = \topng$ if $|\DataMonoid_1^{-1}(\DataMonoid_1(d))| \neq 1$ or $a = \topg$ otherwise.

\item $\theta_1 \transstore \theta_2$. This transition can happen if there is
  $(q,\opstore(p),d) \in\DataState_1$ for some states $p,q \in Q$ and data $d\in\D$.  In
  this case $\theta_2$ is such that $\DataState_2=(\DataState_1 \setminus
  \set{(q,\opstore(p),d)}) \cup \set{(p,\topng, d')}$ where $d'$ is the data
  value of both $\theta_1$ and $\theta_2$ given by~\eqref{eq:condition-Gamma-well-formed}.

\item $\theta_1 \transkeep \theta_2$. This transition can happen if there is
  $(q,\opnothing(p),d) \in\DataState_1$ for some states $p,q \in Q$ and data $d\in\D$.  In
  this case $\theta_2$ is such that $\DataState_2=(\DataState_1 \setminus
  \set{(q,\opnothing(p),d)}) \cup \set{(p,\topng, d)}$.

\item $\theta_1 \transaccept \theta_2$. This transition can happen if there is
  $(q,\opaccept,d) \in\DataState_1$ for some state $q \in Q$ and data $d\in\D$.  In this
  case $\theta_2$ is such that $\DataState_2=\DataState_1 \setminus
  \set{(q,\opaccept,d)}$.

\item $\theta_1 \transdowneq \theta_2$. This transition can happen if there is
  $(q,\tup{\aMonoid}^{=},d) \in \DataState_1$ for some state $q \in Q$ and data
  $d\in\D$ with $(\aMonoid,d) \in \DataMonoid_1$. In this case $\theta_2$ is
  such that $\DataState_2 = \DataState_1 \setminus \set{(q,\tup{\aMonoid}^{=},
    d)}$.

\item $\theta_1 \transdownneq \theta_2$. This transition can happen if there is
  $(q,\tup{\aMonoid}^{\neq},d) \in \DataState_1$ for some state $q \in Q$ and
  data $d\in\D$ and there exists a data $e \in \D$, $e\neq d$ such that $(\aMonoid,e) \in \DataMonoid_1$. In this case $\theta_2$ is such that $\DataState_2 = \DataState_1 \setminus \set{(q,\tup{\aMonoid}^{\neq},d)}$.  

\item $\theta_1 \transtestroot \theta_2$. This transition can happen if
  $r_1=\textit{true} $ and there is $(q,\optestroot,d) \in \DataState_1$ for
  some state $q\in Q$ and data $d\in \D$. In this case $\theta_2$ is such that
  $\DataState_2 = \DataState_1 \setminus \set{(q,\optestroot, d)}$.

\item $\theta_1 \transtestq \theta_2$. This transition can happen there is
  $(q,q,d) \in \DataState_1$ for some data $d\in \D$. In this case $\theta_2$
  is such that $\DataState_2 = \DataState_1 \setminus \set{(q,q, d)}$.

\item The negation of these tests: $\transtestnq$, $\transdowneqn$,
  $\transdownneqn$ and $\transtestnroot$, are defined in a similar way. For
  instance $\theta_1 \transdownneqn \theta_2$ can happen if there is
  $(q,\overline{\tup{\aMonoid}^{\neq}},d) \in \DataState_1$ for some state $q \in Q$ and
  data $d\in\D$ but no data $e \in \D$, $e\neq d$ such that $(\aMonoid,e) \in \DataMonoid_1$. In this case $\theta_2$ is such that $\DataState_2 = \DataState_1 \setminus \set{(q,\overline{\tup{\aMonoid}^{\neq}},d)}$.
\setcounter{transcounter}{\value{enumi}}
\end{enumerate}
\bigskip

Notice that we do not include transitions for the tests $\optesteq$,
$\optestleaf$, $a$ and $b$ as we work with abstractions of $\budta$ in normal form (\cf~Proposition~\ref{prop-normal-form}).

\begin{remark}\label{remark-profile}
  Notice that all the transitions above do not modify the $\profile$ of the
  extended configuration as their $\DataMonoid$ part remains untouched.
\end{remark}

\bigskip

As mentioned earlier, for technical reasons we will need a transition that
makes our trees fatter. The idea is that this transition corresponds to
duplicating an immediate subtree of the root, using a fresh new name for one of
its data value of a given type, where the type of a data value $d$ with an
extended configuration $\theta$ is $\typeof{\theta}(d)$. This transition
assumes the same constraints as for $\tranaconfe$ except that we no longer have
$\DataMonoid_2=\DataMonoid_1$.
%
%For example, if the root has $\tT_1$ and $\tT_2$ as subtrees, consider
%the operation of now having $\tT_1$ $\tT_2$ $\tT'_1$ $\tT'_2$ as subtrees,
%where $\tT'_1$ and $\tT'_2$ are identical to $\tT_1$ and $\tT_2$ except for one
%data value with profile $(S,\chi)$ that in $\tT'_1$ and $\tT'_2$ is replaced by
%a fresh data value. 

\medskip

\begin{enumerate}[leftmargin=8mm]
\setcounter{enumi}{\value{transcounter}}
\item $\theta_1 \transpumpSX \theta_2$. This transition can happen if
  $|\typeof{\theta_1}^{-1}(S,\prof)| \geq 1$ and $|(\DataMonoid_1)^{-1}(\prof)|
  \geq 2$.  Then $\theta_2$ is such that $\data(\theta_2)=\data(\theta_1) \cup
  \set{e}$ for some data $e \not\in\data(\theta_1)$, $\typeof{\theta_2}(e)=(\hat
  S,\prof)$ where $\hat S=S \setminus \set{(q, \topg) : (q, \topg) \in S}$, and
  for all $d \neq e$, $\typeof{\theta_2}(d)=\typeof{\theta_1}(d)$.
  \setcounter{transcounter}{\value{enumi}}
\end{enumerate}

Observe that the conditions required for using $\transpumpSX$ enforce that the
truth value of any test is not changed. Indeed, condition
$|\typeof{\theta_1}^{-1}(S,\prof)| \geq 1$ says that there are at least one
data value of type $(S,\prof)$. On the
other hand, condition $|(\DataMonoid_1)^{-1}(\prof)| \geq 2$ ensures in fact
this is not a `special' data value in the sense that it is the only data value
accessible through $\prof$. This is essential to make sure that
there is at least one data value that can be `copied' multiple times since the automaton can
express properties like ``there is exactly one data value accessible through
$\aMonoid \in \prof$''.  This operation is, intuitively, a way of duplicating
the derivation of $\theta_1$ in the transition system. The idea is that it
would correspond to a run on an expansion of the tree that would yield the
extended configuration $\theta_1$, by duplicating the immediate subtrees of the
root, renaming the data value $d$ with $e$ (see Figure~\ref{fig:inc-tran}).
\begin{figure}
  \centering \includegraphics[scale=1]{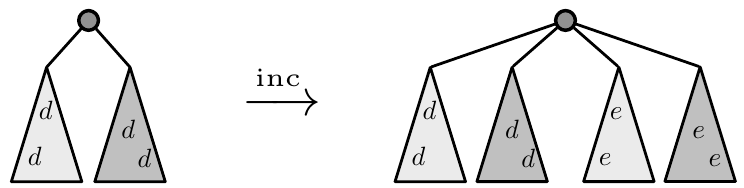}
  \caption{The \emph{inc} operation.}
  \label{fig:inc-tran}
\end{figure}
However, note that the $\topg$-flagged threads are not duplicated in this
operation. This corresponds to assuming that when applying this operation, the new threads do not guess a `new' value $e$ but they stick to guessing the `old' one $d$. This ensures that $\theta_2$  reaches all accepting configurations that $\theta_1$ reaches.

We use $\transpump$ as the union of all $\transpumpSX$ for all $(S,\prof)$.

\begin{remark}\label{remark-inc}
  Notice that the preconditions of \transpumpSX implies that the $\profile$ of the extended
  configuration is not changed.
\end{remark}

\medskip

We now turn to the simulation of an extended configuration moving up in the tree. We split this part into two phases: adding a new root symbol and merging the roots (see Figure~\ref{fig:root-merge-tran}).
\begin{enumerate}[leftmargin=8mm]
\setcounter{enumi}{\value{transcounter}}
\item $\theta_1 \transup \theta_2$. Given two extended configurations
  $\theta_1$ and $\theta_2$ as above, we define $\theta_1 \transup \theta_2$ if
  $r_1=m_1=\textit{false}$, and for all $(q,\alpha, d) \in \DataState_1$, we
  have $\alpha\in\set{\topng,\topg}$. In this case $\theta_2$ is such that
  \begin{align*}
    \DataState_2 &= \set{(q,\bot,d) : (q, \alpha, d) \in \DataState_1
      } \text{,~~~and}\\
   \DataMonoid_2 &=
    \set{(\aMonoid',e) : (\aMonoid,e)\in\DataMonoid_1, \aMonoid' =
      \morphism(c) \conc \aMonoid} \cup \set{(\morphism(c),d')},
  \end{align*}
for some $c \in \A \times \B$ and $d' \in \D$. Notice that $c$ and $d'$ are then the label and data value of $\theta_2$ and that $r_2$ and $m_2$ could be \textit{true} or \textit{false}. As a consequence of the normal form \ref{eq:semigroup-normal-form} of the semigroup, this operation preserves property \eqref{eq:condition-Gamma-well-formed}.

\item $\theta_1, \theta_2 \transmerge \theta_0$. Given 3 extended
  configurations $\theta_1=(\Delta_1,\Gamma_1,r_1,m_1)$,
  $\theta_2=(\Delta_2,\Gamma_2,r_2,m_2)$,
  $\theta_0=(\Delta_0,\Gamma_0,r_0,m_0)$ we say that $\theta_1, \theta_2
  \transmerge \theta_0$ if they all have the same label and data value, $m_1 =
  m_2 = \textit{true}$, $r_1=r_2=r_0$, $\DataState_0=\DataState_1 \cup
  \DataState_2$, and $\DataMonoid_0 = \DataMonoid_1 \cup \DataMonoid_2$. Notice
  that this operation preserves property
  \eqref{eq:condition-Gamma-well-formed} and that $m_0$ could be either \textit{true} or
  \textit{false}.
\end{enumerate}

Given a finite set $\Theta$ of extended configurations and $\theta\in \ac$ we
write $\Theta \tranAC \theta$ if there are extended configurations in $\Theta$
generating $\theta$ according to the transition rules 1 to 14. Based on this,
we define by induction $\Theta \tranAC^+ \theta$ if:
  \begin{enumerate}
  \item $\Theta \tranAC \theta$, or
  \item there is an extended configuration $\theta'\in \ac$ such that $\Theta \tranAC \theta'$ and $\Theta \cup
    \set{\theta'} \tranAC^+ \theta$.
  \end{enumerate}

In the definition of the transition system, the $m$ flag is simply used to
constrain the transition system to have all its $\transmerge$ operations right
after $\transup$ and before any $\tranaconfe$. Thus any transition $\tranAC^+$
complies to the following regular expression:
\begin{align}\label{eq:form-of-derivation}
  \big((\tranaconfe \midd \transpump)^* 
  \transup (\transmerge)^* \big)^*(\tranaconfe \midd \transpump)^* \ . \tag{$\dag$}
\end{align}

% We will show the following:
% \begin{proposition}\label{prop:complextransition-downwards-comp}
%   Let 
% \[{\tranAC} \quad\coloneqq\quad \big({\Transup ; (\Transmerge)^* ; (\Transeps \cup \Transpump)^*}  \big).\]
% $(\AC,\Rightarrow)$ is downwards compatible with respect to $\lqAC$.
% \end{proposition}

\subsection{Compatibility}\label{section-compatibility}
\newcommand{\kUpto}{\textbf{\textit l}}
\newcommand{\lqacstrX}{\lqacstr_\kUpto}
\newcommand{\lqacX}{\lqac_\kUpto}
\newcommand{\konstN}{\mathbf{N}}
\newcommand{\eqAC}{\equiv} 

Given two finite sets of extended configurations $\Theta$ and $\Theta'$ we write $\Theta \lqAC \Theta'$ if for all $\theta' \in \Theta'$ there is $\theta \in \Theta$ such that $\theta \lqacstr \theta'$. That is, every element from $\Theta'$ is \emph{minorized} by an element of $\Theta$.

Given a finite set $\Theta$ of extended configurations, we denote by
$\Theta^\eqconf$ the set of all extended configurations equivalent to some configuration in
$\Theta$. In other words, $\Theta^\eqconf$ is the closure of $\Theta$ under
bijections on data values.

The following proposition essentially shows that bigger extended configurations
can be safely ignored in order to test for emptiness. This will be the main
technical contribution of this section.

\begin{proposition}\label{prop-compatible}
  If $\Theta,\Theta'$ are finite sets of extended configurations such that
  $\Theta \lqAC \Theta'$ and $\theta'\in\ac$ is such that $\Theta' \tranAC
  \theta'$ then there exists $\theta\lqacstr \theta'$ such that
  $\Theta^\eqconf\tranAC^n \theta$, for some $n$ computable from $\Theta$.
\end{proposition}

The key ingredient to prove this, is to show that transitions on \ac are compatible with the order of \ac. We treat the case of \transmerge in a separate lemma later.

We first show useful lemmas illustrating the power of \transpump.

\begin{figure}
  \centering
  \includegraphics[scale=1]{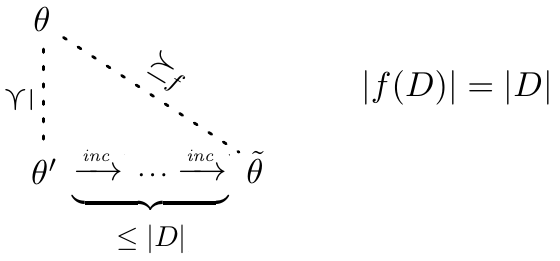}
  \caption{Lemma~\ref{lemma-injectivity}.}
  \label{fig:lemma-injectivity}
\end{figure}
\begin{lemma}[Figure~\ref{fig:lemma-injectivity}]\label{lemma-injectivity}
  Let $\theta,\theta' \in \ac$ such that $\theta' \lqacstr \theta$.  Let $D\subseteq \D$ be a finite set of data values.  Then there exists $\tilde\theta$ such that:
  \begin{enumerate}
  \item $\theta'~ (\transpump)^n ~\tilde\theta$, for some $n \leq |D|$,
  \item $\tilde\theta \lqacstr_f \theta$,
  \item $f$ is injective on $D$.
  \end{enumerate}
\end{lemma}
\begin{proof}
  Fix $\theta=(\DataState,\DataMonoid,r,m)$ and
  $\theta'=(\DataState',\DataMonoid',r',m')$.  Assume that $\theta' \lqacstr_f
  \theta$. We will modify $\theta'$ and $f$ until $f$ becomes injective.

  If $f$ is not injective then we have $d,d'\in D$, $d\neq d'$ but $f(d)=f(d')$. 

  Let
  $(S,\prof)=\typeof{\theta'}(f(d))$. Then we have
  $\prof=\DataMonoid(d)=\DataMonoid'(f(d))=\DataMonoid'(f(d'))=\DataMonoid(d')$.
  From $\theta'\lqacstr\theta$ it follows that
  $\profile(\theta)=\profile(\theta')$. Therefore $\theta'$ must contain
  another data value $e\neq f(d)$ such that $\DataMonoid'(e)=\prof$. Hence we
  can apply $\theta'\transpumpSX \theta''$ adding a copy $\tilde d$ of $f(d)$, so
  that $\typeof{\theta''}(\tilde d)=(\hat S, \prof)$ where $\hat S = S
  \setminus \set{(q, \topg) : (q, \topg) \in S}$. The reader can verify that
  the resulting extended configuration $\theta''$ is such that $\theta''
  \lqacstr_{g} \theta$ where $g$ is the mapping identical to $f$ except that
  $g(d)=\tilde d\neq f(d')=g(d')$. Repeating this operation at most $|D|$ times we
  obtain the desired $\tilde \theta$.
\end{proof}

\begin{lemma}\label{lem:augment-elements}
  Let $\theta,\theta' \in \ac$ as in \eqref{eq:notation} such that $\theta' \lqacstr \theta$.  Then there
  exists $\tilde\theta$ such that:
  \begin{enumerate}
  \item $\theta'~ (\transpump)^n ~\tilde\theta$, for some $n \leq 3 \cdot 2^{|\Monoid|}$,
  \item $\tilde\theta \lqacstr \theta$,
  \item for all $\prof\subseteq\Monoid$, if $\theta$ contains strictly more
    than $2$ data values $d$ with $\DataMonoid(d)=\prof$, then so does
    $\tilde\theta$.
  \end{enumerate}
\end{lemma}
\begin{proof}
  We apply Lemma~\ref{lemma-injectivity} to the set $D$ defined by selecting
  for each $\prof\subseteq\Monoid$ three data values $e$ such that
  $\DataMonoid(e)=\prof$, if this is possible. This gives an upper bound of $3 \cdot 2^{|\Monoid|}$ applications of $\transpump$.
\end{proof}

\begin{lemma}\label{lem:injective-d}
  Let $\theta,\theta' \in \ac$ as in \eqref{eq:notation} such that $\theta' \lqacstr \theta$ and let $d \in \D$.  Then there
  exists $\tilde\theta$ such that:
  \begin{enumerate}
  \item $\theta'~ (\transpump)^n ~\tilde\theta$, for some $n \leq 2^{|\Monoid|}+1$,
  \item $\tilde\theta \lqacstr_f \theta$,
  \item $f^{-1}(f(d)) = \set d$.
  \end{enumerate}
\end{lemma}
\begin{proof}
  For every $S \subseteq Q \times \testactionsbot$, let $E_S \subseteq \D$ be a
  subset of $f^{-1}(f(d))$ containing one data value $e \neq d$ so that
  $\DataState(e) = S$, or no data values otherwise.

Then, we can apply Lemma~\ref{lemma-injectivity} with $D = \set d \cup
\bigcup_{S \subseteq Q \times \testactionsbot} E_S$ and we obtain an extended
configuration $\tilde \theta$ so that $\theta' ~ (\transpump)^{\leq |D|}
~\tilde\theta$ and $\tilde \theta \lqacstr_{f'} \theta$ for some $f'$ that is
injective on $D$. From $f'$ we can create the desired $f$, by sending each $e
\neq d$, $e \in f^{-1}(f(d))$ to $f'(e')$ where $\set {e'} =
E_{\DataState(e)}$.
\end{proof}

We now show the first monotonicity lemma regarding all transitions except
$\transmerge$. In this case $n$  does not depend on the
extended configurations involved, only on size of the monoid $\Monoid$.

\begin{lemma}\label{lem:monotonic}
  Let $\theta_1,\theta_2$ and $\theta'_1$ be extended configurations such that
  $\theta_1 \tranaconf \theta_2$ and $\theta'_1 \lqacstr \theta_1$, where
  $\tranaconf$ is either $\transup$ or in $\tranaconfe$.  Then there exists an
  extended configuration $\theta'_2$ such that $\theta'_2 \lqacstr \theta_2$ and
  $\theta'_1 \tranaconf^n \theta'_2$ for some $n \leq 4 \cdot 2^{ |\Monoid|}+1 $.
\end{lemma}

\begin{proof}%[Proof of Lemma~\ref{lem:monotonic}]
  This is done by a case analysis depending on where $\tranaconf$ comes
  from. Throughout the proof we use the following notation for $i\in\set{1,2}$:
  $\theta_i=(\DataState_i,\DataMonoid_i,r_i,m_i)$, similarly for
  $\theta'_i$. Moreover $a_i$ and $d_i$ denote the label and data value
  associated to $\theta_i$ (similarly for $\theta'_i$). We also denote by $f$
  a function witnessing $\theta'_1 \lqacstr \theta_1$.
  In particular (recall Remark~\ref{remark-label}) we have $a_1=a'_1$ and $f(d_1)=d'_1$.

  By Lemma~\ref{lem:augment-elements}, applying at most $3\cdot 2^{ |\Monoid|}$
  transition steps on $\theta'_1$ we can now further assume that
  the conclusion of Lemma~\ref{lem:augment-elements} holds for
  $\theta_1$ and $\theta'_1$. We now turn to the case analysis. Each case will
  transform further $\theta'_1$ in order to get the desired property. This last
  transformation will require at most $\max(3,2^{ |\Monoid|}+1)$ steps. 

\begin{enumerate}[label=(\alph*),ref=\alph*]
\item
Suppose $\theta'_1 \lqacstr_f \theta_1 \transup \theta_2$.

  We first show that $\transup$ can also be applied to $\theta'_1$.  Assume
  towards a contradiction that this is not the case. Then $\theta'_1$ contains
  a thread $(q,\alpha,e')$ such that $\alpha\not\in \set{\topng, \topg}$. Then
  by Remark~\ref{remark-order} a thread of the form $(q,\alpha,e)$, with $f(e)=e'$
  is in $\theta_1$, a contradiction from the fact that \transup was
  applied to $\theta_1$.

  Let $\theta'_2$ be the extended configuration such that $\theta'_1 \transup
  \theta'_2$, $r'_2=r_2$, $a'_2=a_2$ and $d'_2=f(d_2)$. We show that $\theta'_2
  \lqacstr_f \theta_2$ concluding this case. Note that $f$ still satisfies the
  surjectivity condition. It remains to show that
  $\profile(\theta'_2)=\profile(\theta_2)$, and for every $d$, $\DataState'_2(f(d)) \subseteq
  \DataState_2(d)$ and $\DataMonoid'_2(f(d)) = \DataMonoid_2(d)$.

\begin{enumerate}[label=(\roman*),ref=\roman*]
\item $\DataState'_2(f(d)) \subseteq \DataState_2(d)$ and
  $\DataMonoid'_2(f(d)) = \DataMonoid_2(d)$.  By definition of $\transup$, we
  have that $\DataMonoid_2(d)$ is completely determined from
  $\DataMonoid_1(d)$, $a_2$ and the fact that $d$ is equal to $d_2$ or
  not. In a similar way, whether a thread $(q, \bot)$ is in
  $\DataState_2(d)$ is determined by whether $(q, \topng)$ or $(q, \topg)$
  is in $\DataState_1(d)$.  But by construction, modulo $f$, all these facts are identical for
  $\theta'_2$ and $\theta_2$. Hence, since $\DataState'_1(f(d)) \subseteq
  \DataState_1(d)$ and $\DataMonoid'_1(f(d)) = \DataMonoid_1(d)$, we have
  that $\DataState'_2(f(d)) \subseteq \DataState_2(d)$ and
  $\DataMonoid'_2(f(d)) = \DataMonoid_2(d)$.

\item $\profile(\theta'_2)=\profile(\theta_2)$. We already have by construction
  $r'_2=r_2$ and $m'_2=m_2$.  From the previous item we have for all $\prof$,
  $|{\DataMonoid'_2}^{-1}(\prof)|>0$ implies $|\DataMonoid_2^{-1}(\prof)|>0$
  and that $|{\DataMonoid'_2}^{-1}(\prof)|>1$ implies $|\DataMonoid_2
  ^{-1}(\prof)|>1$. It remains to show the converse of these implications. We
  only show the second one as the first one is similar and simpler. We thus
  only show that if $|\DataMonoid_2 ^{-1}(\prof)|>1$ then $|{\DataMonoid'_2}
  ^{-1}(\prof)|>1$.

  Assume we have two distinct data values $e_1\neq e_2$ such that
  $\DataMonoid_2(e_1) = \DataMonoid_2(e_2) = \prof$. Let
  $\prof_1=\DataMonoid_1(e_1)$ and $\prof_2=\DataMonoid_1(e_2)$.  From $\profile(\theta'_1)=\profile(\theta_1)$ we know that there
  exist data values $e'_1$ and $e'_2$ such that $\prof_1={\DataMonoid'_1}(e'_1)$ and $\prof_2={\DataMonoid'_1}(e_2)$. Even in the case where
  $\prof_1=\prof_2$ we get from $\profile(\theta'_1)=\profile(\theta_1)$ that
  $e'_1$ and $e'_2$ can be chosen distinct. Moreover, as
  $\DataMonoid'_1(f(d_2))=\DataMonoid_1(d_2)$, if $e_1=d_2$ (resp.\  $e_2=d_2$) then we
  can pick $e'_1=f(d_2)$ (resp.\  $e'_2=f(d_2)$). In the case where both $e_1$ and $e_2$ are different from
  $d_2$ we pick $e'_1$ and $e'_2$ different from $f(d_2)$. The latter is always possible
  because of Item~(3) of Lemma~\ref{lem:augment-elements}.
  Altogether we have $e_i=d_2$ iff $e'_i=f(d_2)$ for $i=1,2$.

  Hence, since $a_2=a'_2$ and $f(d_2) = d'_2$, by definition of \transup we
  have ${\DataMonoid'_2}(e'_1)={\DataMonoid'_2}(e'_2)=\prof$.
\end{enumerate}

\item 
Suppose $\theta'_1 \lqacstr_f \theta_1 \transuniv \theta_2$. 

By definition of $\transuniv$ there is $(q,\opuniv(p),d) \in \DataState_1$ and
$a_2=a_1$, $d_2=d_1$, $\DataMonoid_2=\DataMonoid_1$, and $\DataState_2 =
(\DataState_1 \setminus \set{(q,\opuniv(p),d)}) \cup \set{(p,\topng,e) :
  \exists \aMonoid ~.~ (\aMonoid,e) \in \DataMonoid_1}$.
Applying Lemma~\ref{lem:injective-d} to $d$, we can assume that $f^{-1}(f(d)) = \set d$.

\begin{itemize}
\item If $(q, \opuniv(p),f(d)) \not\in \DataState'_1$, then we show that
  $\theta'_1 \lqacstr_f \theta_2$. First, it is immediate that
  $\profile(\theta'_1) = \profile(\theta_2)$ since $\profile(\theta_2) =
  \profile(\theta_1)$ by definition of $\transuniv$ (recall
  Remark~\ref{remark-profile}) and $\profile(\theta_1) =
  \profile(\theta'_1)$ since $\theta'_1 \lqacstr \theta_1$. Furthermore:
\begin{itemize}
\item for all $\hat e \neq d$, we have that $\DataMonoid'_1(f(\hat e)) =
  \DataMonoid_1(\hat e) = \DataMonoid_2(\hat e)$ and $\DataState'_1(f(\hat
  e)) \subseteq \DataState_1(\hat e) \subseteq \DataState_2(\hat e)$,
\item for $d$ we have $\DataMonoid'_1(f(d)) = \DataMonoid_1(d) =
  \DataMonoid_2(d)$ and $\DataState'_1(f(d)) \subseteq \DataState_1(d)
  \setminus \set{(q, \opuniv(p))} \subseteq \DataState_2(d)$, since $(q,
  \opuniv(p)) \not\in \DataState'_1(f(d))$.
\end{itemize}

\item If $(q, \opuniv(p),f(d)) \in \DataState'_1$, we perform a
  \transuniv transition on $\theta'_1$ and obtain $\theta'_2$ so that $\theta'_1 \transuniv
  \theta'_2$ and $\theta'_2 \lqacstr \theta_2$. Let $\theta'_2$ be so that
  $r'_2 = r'_1$, $m'_2 = m'_1$, $\DataMonoid'_2 = \DataMonoid'_1$, and
  $\DataState'_2 = (\DataState'_1 \setminus \set{(q,\opuniv(p),f(d))}) \cup
  \set{(p,\topng,e) : \exists \aMonoid ~.~ (\aMonoid,e) \in \DataMonoid'_1}$.
%Note that $a = \top$ iff $|{\DataMonoid_1}^{-1}( \DataMonoid_1(d) )|= |{\DataMonoid_1}^{-1}( \DataMonoid'_1(f_1(d)) )| \neq 1$ (by $\theta'_1 \lqacstr_{f_1} \theta_1$) iff $|{\DataMonoid'_1}^{-1}( \DataMonoid'_1(f_1(d)) )| \neq 1$ (by $\profile(\theta'_1) = \profile(\theta'_2)$). 
Hence, $\theta'_1 \transuniv \theta'_2$. 

We now show that $\theta'_2 \lqacstr_{f} \theta_2$.
 We have that $\profile(\theta'_2) = \profile(\theta'_1)$ and $\profile(\theta_2) = \profile(\theta_1)$ by definition of $\transuniv$, and that $\profile(\theta_1) = \profile(\theta'_1)$ by $\theta'_1 \lqacstr \theta_1$. Hence, $\profile(\theta'_2) = \profile(\theta_2)$. 

       \begin{itemize}
       \item For all $\hat e$ so that $f(\hat e) \neq f(d)$ (and hence  $\hat e \neq d$), we have that $\DataMonoid'_2(f(\hat e)) =
           \DataMonoid'_1(f(\hat e)) = \DataMonoid_1(\hat e) =
           \DataMonoid_2(\hat e)$ and $\DataState'_2(f(\hat e)) =
           \DataState'_1(f(\hat e)) \cup \set{(p,\top) :
             \DataMonoid'_1(f(\hat e))\neq\emptyset)} \subseteq
           \DataState_1(\hat e) \cup \set{(p,\top) :
             \DataMonoid_1(\hat e)\neq\emptyset)} = \DataState_2(\hat
           e)$.

       \item For $d$, we have $\DataMonoid'_2(f(d)) =\DataMonoid'_1(f(d)) =
         \DataMonoid_1(d) = \DataMonoid_2(d)$ and, on the other hand,
         $\DataState'_2(f(d)) = \DataState'_1(f(d))\setminus
         \set{(q,\opuniv(p))} \cup \set{(p,\top) :
           \DataMonoid'_1(f(d))\neq\emptyset)} \subseteq \DataState_1(d)
         \setminus \set{(q, \opuniv(p))} \cup \set{(p,\top) :
           \DataMonoid_1(d)\neq\emptyset)} = \DataState_2(d)$.
       \item Finally, for any $\hat e$ such that $f(\hat e) = f(d)$ we have that $\hat e = d$ since $f^{-1}(f(d)) = \set d$, and we apply the previous item.
        \end{itemize}
Hence, $\theta'_2  \lqacstr_{f} \theta_2$.
\end{itemize}

\item 
Suppose  $\theta'_1 \lqacstr_f \theta_1 \transguess \theta_2$. 

 By definition of $\transguess$ there is $(q,\opguess(p), d) \in \DataState_1$ and $a_2=a_1$, $d_2=d_1$, $r_2=r_1$, $m_2=m_1$, $\DataMonoid_2=\DataMonoid_1$ and $\DataState_2 = (\DataState_1 \setminus \set{(q,\opguess(p), d)})\cup \set{(p,\topg, e)}$ for some data value $e$.
% , where
% where $a = \topng$ if $|\DataMonoid_1^{-1}(\DataMonoid_1(d))| \neq 1$ or $a = \topg$ otherwise. 
We show that either $\theta'_1 \lqacstr \theta_2$ or there is some $\theta'_2$ so that $\theta'_1 \transguess \theta'_2 \lqacstr \theta_2$.
By applying Lemma~\ref{lemma-injectivity} for $D = \set{d,e}$, we can assume without any loss of generality that if $d \neq e$ then $f(d) \neq f(e)$.

\begin{itemize}
\item If $(q, \opguess(p),f(d)) \not\in \DataState'_1$, then we show that
  $\theta'_1 \lqacstr_f \theta_2$. First, it is immediate that
  $\profile(\theta'_1) = \profile(\theta_2)$ since $\profile(\theta_2) =
  \profile(\theta_1)$ by definition of $\transguess$ and $\profile(\theta_1) =
  \profile(\theta'_1)$ since $\theta'_1 \lqacstr \theta_1$. For all data values
  $\hat e$ we have  $\DataMonoid'_1(f(\hat e)) = \DataMonoid_1(\hat e) = \DataMonoid_2(\hat e)$. Furthermore:
       \begin{itemize}
       \item for all $\hat e \not\in\set{d,e}$, we have $\DataState'_1(f(\hat e)) \subseteq \DataState_1(\hat e) =
         \DataState_2(\hat e)$,
       \item for $d$ we have $\DataState'_1(f(d)) \subseteq \DataState_1(d)
         \setminus \set{(q, \opguess(p))} \subseteq \DataState_2(d)$ (since
         $(q, \opguess(p)) \not\in \DataState'_1(f(d))$),
       \item for $e$, assuming that $e \neq d$ (otherwise the previous item
         applies), we have $\DataState'_1(f(e)) \subseteq \DataState_1(e)
         \subseteq \DataState_2(e)$.
        \end{itemize}
        Hence, $\theta'_1 \lqacstr_{f} \theta_2$.
      \item If $(q, \opguess(p),f(d)) \in \DataState'_1$, we show that can
        perform the same action and we obtain $\theta'_2$ so that $\theta'_1
        \transguess \theta'_2$ and $\theta'_2 \lqacstr \theta_2$. Let
        $\theta'_2$ be so that $r'_2 = r'_1$, $m'_2 = m'_1$, $\DataMonoid'_2 =
        \DataMonoid'_1$, and $\DataState'_2 = (\DataState'_1 \setminus \set{(q,
          \opguess(p), f(d)}) \cup \set{(p,\topg,f(e))}$.  Hence, $\theta'_1
        \transguess \theta'_2$.

We now show that $\theta'_2 \lqacstr_f \theta_2$.  We have that
$\profile(\theta'_2) = \profile(\theta'_1)$ and $\profile(\theta_2) =
\profile(\theta_1)$ by definition of $\transguess$, and that
$\profile(\theta_1) = \profile(\theta'_1)$ by $\theta'_1 \lqacstr
\theta_1$. Hence, $\profile(\theta'_2) = \profile(\theta_2)$. By construction
we also have for all $\hat e$,  $\DataMonoid'_2(f(\hat e)) =
\DataMonoid'_1(f(\hat e)) = \DataMonoid_1(\hat e) = \DataMonoid_2(\hat e)$. Furthermore:
       \begin{itemize}
       \item for all $\hat e \not\in\set{d,e}$, we have that
         $\DataState'_2(f(\hat e)) = \DataState'_1(f(\hat e)) \subseteq
         \DataState_1(\hat e) = \DataState_2(\hat e)$,
          \item for $d$, we have 
            \begin{itemize}
            \item if $d \neq e$, then (since we assumed we have applied Lemma~\ref{lemma-injectivity} with $D = \set{d,e}$) we have $f(d) \neq f(e)$, and hence $\DataState'_2(f(d)) =
                \DataState'_1(f(d))\setminus \set{(q,\opguess(p))}
                \subseteq \DataState_1(d) \setminus \set{(q,
                  \opguess(p))} = \DataState_2(d)$,
            \item if $d = e$, 
$\DataState'_2(f(d)) =
              (\DataState'_1(f(d))\setminus \set{(q,\opguess(p))}) \cup
              \set{(p,\topg)} \subseteq (\DataState_1(d) \setminus \set{(q,
                \opguess(p))}) \cup \set{(p,\topg)} = \DataState_2(d)$,
            \end{itemize}
          \item for $e$, assuming that $e \neq d$ (otherwise the previous item applies), we have  
$\DataState'_2(f(e)) = \DataState'_1(f(e)) \cup \set{(p,\topg)} \subseteq \DataState_1(e) \cup \set{(p,a)} = \DataState_2(e)$.
        \end{itemize}
Hence, $\theta'_2  \lqacstr_{f} \theta_2$.
\end{itemize}

\item Suppose we have $\theta'_1 \lqacstr_f \theta_1
  \transpumpSX \theta_2$.

  By definition of \transpumpSX, this implies that
  $|\typeof{\theta_1}^{-1}(S,\prof)| \geq 1$ and
  $|\DataMonoid_1^{-1}(\prof)|\geq 2$.

  Let $d$ be the only data value in $\data(\theta_2)\setminus\data(\theta_1)$
  given by the definition of \transpumpSX.  Let $\hat d \in
  \typeof{\theta_1}^{-1}(S,\prof)$.  We further have that $\typeof{\theta_2}(e)
  = \typeof{\theta_1}(e)$ for all $e \neq d$, and $\typeof{\theta_2}(d) = (\hat
  S,\prof)$ where $\hat S = S \setminus \set{(q, \topg) : (q,\topg) \in S}$.

  From $\theta'_1 \lqacstr_{f} \theta_1$ it follows that
  $\typeof{\theta'_1}(f(\hat d)) = (S',\prof)$ for some $S' \subseteq
  S$. Further, by $\profile(\theta'_1)=\profile(\theta_1)$ we have that
  $|{\DataMonoid'_1}^{-1}(\prof)| \geq 2$.  We can then
  apply $\transpumpSpX$ to $\theta'_1$. Let $\theta'_2$ be so that $\theta'_1
  \transpumpSpX \theta'_2$ and $f(d)$ is chosen as the new data value duplicating $f(\hat  d)$. Note that $\Delta'_1(f(d)) \subseteq \Delta_1(d)$ and $\Gamma'_1(f(d))=\Gamma_1(d)$ and $\Delta_1(d)=\emptyset=\Gamma_1(d)$, and hence $f(d)$ has all the desired properties. We show that $\theta'_2 \lqacstr_{f} \theta_2$ concluding this
  case. Note that $f$ satisfies the surjectivity condition. It remains to show
  that $\profile(\theta'_2) = \profile(\theta_2)$ and for all $e\in\D$ we have
  $\DataState'_2(f(e)) \subseteq \DataState_2(e)$ and $\DataMonoid'_2(f(e)) =
  \DataMonoid_2(e)$.

\begin{enumerate}[label=(\roman*),ref=\roman*]
\item $\profile(\theta'_2) = \profile(\theta_2)$. This is immediate from the
  fact that $\profile(\theta'_2) = \profile(\theta'_1)$, $\profile(\theta_2) =
  \profile(\theta_1)$ by definition of $\transpump$ (recall Remark~\ref{remark-inc}), and $\profile(\theta'_1) =
  \profile(\theta_1)$ by definition of $\theta'_1 \lqacstr \theta_1$.
\item $\DataState'_2(f(e)) \subseteq \DataState_2(e)$ and
  $\DataMonoid'_2(f(e)) = \DataMonoid_2(e)$. For every $e \neq d$, it is
  immediate as $\typeof{\theta}(e)$ is not affected by $\transpump$.

  For the data value $d$, note that from $\theta'_1 \lqacstr_{f} \theta_1$ and
  the definition of \transpump it follows that $\DataMonoid'_2(f(\hat d)) =
  \DataMonoid'_1(f(\hat d)) = \DataMonoid_1(\hat d) = \DataMonoid_1(d) =
  \DataMonoid_2(d)$, as desired.  Similarly we obtain that $\DataState'_2(f(d))
  = \hat S'$ and $\DataState_2(d) = \hat S$, where $\hat S' = S' \setminus
  \set{(q, \topg) : (q, \topg) \in S'}$. Since $S' \subseteq S$, we have $\hat
  S' \subseteq \hat S$, and hence $\DataState'_2(f(d)) \subseteq
  \DataState_2(d)$.
\end{enumerate}

\item Suppose we have $\theta'_1 \lqacstr \theta_1
  \transstore \theta_2$.  This is treated like for \transguess.

\item Suppose we have $\theta'_1 \lqacstr \theta_1
  \transkeep \theta_2$. This is treated like for \transguess.

\item Suppose we have $\theta'_1 \lqacstr \theta_1
  \transaccept \theta_2$. This is treated like for \transguess.

\item Suppose we have $\theta'_1 \lqacstr_f \theta_1
  \transdelta \theta_2$.  Then, there is $(q, \bot, d) \in \DataState_1$ and some $\tau \in \delta$ so that $\DataMonoid_2 = \DataMonoid_1$, $r_2 = r_1$, $m_2 = m_1$, and $\DataState_2 = (\DataState_1 \setminus \set{(q,\bot,d)}) \cup \set{(q,\alpha_i,d) : i \leq |\tau|}$, where $(\alpha_i)_{i \leq |\tau|}$ are all the tests and actions occurring in $\tau$. Applying Lemma~\ref{lem:injective-d} to $d$, we can assume that $f^{-1}(f(d)) = \set d$.
  \begin{itemize}
  \item If $(q, \bot, f(d)) \not\in \DataState'_1$, it is immediate that $\theta'_1 \lqacstr_f \theta_2$.
  \item If $(q,\bot,f(d)) \in \DataState'_1$, let $\theta'_2$ the result of a similar transition $\transdelta$ triggered by $\tau$ on $(q,\bot,f(d))$. That is, $\DataMonoid'_2 = \DataMonoid'_1$, $r'_2 = r'_1$, $m'_2 = m'_1$, and $\DataState'_2 = (\DataState'_1 \setminus \set{(q,\bot,f(d))}) \cup \set{(q,\alpha_i,f(d)) : i \leq |\tau|}$. Hence, $\theta'_1 \transdelta \theta'_2$. We show that $\theta'_2 \lqacstr_f \theta_2$. First note that  $\profile(\theta'_2) = \profile(\theta'_1)$ and $\profile(\theta_2) = \profile(\theta_1)$ by Remark~\ref{remark-profile}, and that $\profile(\theta'_1) = \profile(\theta_1)$ by $\theta'_1 \lqacstr \theta_1$; thus, $\profile(\theta'_2) = \profile(\theta_2)$.
On the other hand, since $\DataMonoid'_2 = \DataMonoid'_1$ and $\DataMonoid_2 = \DataMonoid_1$ and $\DataMonoid_1 = \DataMonoid'_1$, we are only left with checking that for every $e$, $\DataState'_2(f(e)) \subseteq \DataState'_1(e)$.
\begin{itemize}
\item For every $e$ so that $f(e) \neq f(d)$ this is true since
  $\DataState'_2(f(e)) = \DataState'_1(f(e)) \subseteq \DataState_1(e)
  = \DataState_2(e)$.
\item For $d$ we have $\DataState'_2(f(d)) = (\DataState'_1(f(d)) \setminus \set{(q,\bot)}) \cup \set{(q,\alpha_i) : i \leq |\tau|} \subseteq (\DataState_1(d) \setminus \set{(q,\bot)}) \cup \set{(q,\alpha_i) : i \leq |\tau|} = \DataState_2(d)$.
\item Finally, for any $e$ so that $f(e) = f(d)$ we have that $e = d$ since $f^{-1}(f(d)) = \set d$, and the previous item applies.
\end{itemize}
  \end{itemize}

\item Suppose we have $\theta'_1 \lqacstr_f \theta_1 \transdowneq \theta_2$.

 By definition of $\transdowneq$ there is $(q,\tup{\aMonoid}^{=},d) \in
 \DataState_1$ and $\theta_2$ is defined as $a_2=a_1$,
  $d_2=d_1$, $\DataMonoid_2=\DataMonoid_1$ and $\DataState_2 = \DataState_1
  \setminus \set{(q,\tup{\aMonoid}^{=},d)}$.

\begin{itemize}
\item   Assume first that $(q,\tup{\aMonoid}^{=},f(d))\not\in\DataState'_1$. We show
  that $\theta'_1 \lqacstr_f \theta_2$. By construction $f$ satisfies the
  surjectivity condition. Moreover the profiles remain untouched during the
  transition hence $\profile(\theta'_1)=\profile(\theta_2)$ and for all
  $e\in\D$, $\DataMonoid'_1(f(e)) = \DataMonoid_2(e)$. 
  It remains to show that for all $e\in\D$, $\DataState'_1(f(e)) \subseteq
  \DataState_2(e)$. For $e\neq d$ this is immediate as
  $\DataState_1(f(e))=\DataState_2(f(e))$. For $e=d$ we know by hypothesis that
  $\DataState'_1(f(d))$ does not contain  $(q,\tup{\aMonoid}^{=})$ which is the
  only one affected by the transition. The result follows.

\item  Assume now that $(q,\tup{\aMonoid}^{=},f(d))\in\DataState'_1$.  As
  $\DataMonoid'_1(f(d))=\DataMonoid_1(d)$, $\mu\in\DataMonoid'_1(f(d))$ and we can apply a transition $\transdowneq$ to $\theta'_1$ using this
  thread. Let $\theta'_2$ be the resulting extended configuration. Applying Lemma~\ref{lem:injective-d} to $d$, we can assume that $f^{-1}(f(d)) = \set d$.

  We show that $\theta'_2 \lqacstr_f \theta_2$. Notice that $f$ does satisfies
  the surjectivity condition. It remains to show
  that $\profile(\theta'_2) = \profile(\theta_2)$ and for all $e\in\D$ we have
  $\DataState'_2(f(e)) \subseteq \DataState_2(e)$ and $\DataMonoid'_2(f(e)) =
  \DataMonoid_2(e)$.

 \begin{enumerate}[label=(\roman*),ref=\roman*]
 \item $\DataState'_2(f(e)) \subseteq \DataState_2(e)$ and
   $\DataMonoid'_2(f(e)) = \DataMonoid_2(e)$. Let $(S,\prof)=\typeof{\theta_2}(e)$.
  
   If $e \neq d$ then it follows from the definition of $\transdowneq$ that
   $\typeof{\theta_2}(e)=\typeof{\theta_1}(e)$ and
   $\typeof{\theta'_2}(f(e))=\typeof{\theta'_1}(f(e))$ (since $f(e) \neq f(d)$ by $f^{-1}(f(d)) = \set d$). From $\theta'_1
   \lqacstr_f \theta_1$ it follows that $\typeof{\theta'_1}(f(e))=(S',\prof)$ with $S'
   \subseteq S$. The result follows.
 
   In the case where $e=d$,
   $\typeof{\theta_2}(e)=\typeof{\theta'_2}(f(e))=(S,\prof)$.

 \item $\profile(\theta_2)=\profile(\theta'_2)$. This is obvious as by
   Remark~\ref{remark-profile}, the transition does not affect profiles.
 \end{enumerate}
\end{itemize}

\medskip
\item The cases $\transdownneq, \transtestroot, \transtestq$ and their
  negative counterparts are treated similarly. 

  We only need to verify that whenever $\theta'_1 \lqacstr \theta_1$, if
  $\theta_1$ satisfies the condition for any transition in
  $\set{\transdownneq,\transtestroot,\transtestq}$ or their negative counterparts, then
  $\theta'_1$ also satisfies this condition.  For all these cases this is a
  simple consequence of $\theta_1$ and $\theta'_1$ having the same
  profile.

  In the case $\overline{\tup{\aMonoid}^{=}}$, if $\aMonoid \not\in
  \DataMonoid_1(d)$, then the same hold in $\theta'_1$ for $f(d)$.

  For the case $\tup{\aMonoid}^{\not=}$, if $\aMonoid \in \DataMonoid_1(d')$
  for some $d'\neq d$, then by equality of profiles there is a data value $e$
  such that $\typeof{\theta'_1}(e)=\typeof{\theta_1}(d')$.  Notice that even in the case when
  $\typeof{\theta_1}(e)=\typeof{\theta_1}(d)$, the definition of $\profile$ and the fact that
  $d\neq d'$ guarantees that we can always choose $e\neq d$. Hence $\aMonoid
  \in \DataMonoid_1(e)$ and $e\neq d$ as required.

For the case $\overline{\tup{\aMonoid}^{\not=}}$, if $\aMonoid \not\in
  \DataMonoid_1(d')$ for all $d'\neq d$, $\DataMonoid_1(d)$ is the unique
  one that may contain $\mu$. By equality of the profile, this must also be the
  case for $\DataMonoid'_1(d)$ as desired.

The cases \transtestroot and \transtestnroot are straightforward as the value
of $r$ is preserved by $\lqacstr$.

The cases \transtestq and \transtestnq are also immediate.\qedhere
\end{enumerate}
\end{proof}

We now do the same for \transmerge. In this case $n$ depends on the size of the extended configurations and we may need to perform a permutation on the data values. In this case, the use of $\transpump$ becomes essential. Indeed, while it is not true that a transition $\theta_1, \theta_2  \transmerge \theta_0$ can be downward-simulated with a transition $\transmerge$, it can be simulated by a sequence of transitions $(\transpump)^*\transmerge$ as we show next, which is the reason why we needed to introduce $\transpump$ in our transition system. More precisely, while for $\theta'_1 \lqacstr \theta_1$, $\theta'_2 \lqacstr \theta_2$ we have that $\theta_1, \theta_2  \transmerge \theta_0$ may not be simulated from applying a merge transition to $\theta'_1, \theta'_2$, it still \emph{could} be simulated from some ``fatter'' versions of $\theta'_1, \theta'_2$ (\ie, from the result of applying a number of $\transpump$ transitions).
\begin{lemma}~\label{lem:monotonic-merge}
  Let $\theta_1,\theta_2,\theta_0$ be extended configurations such that
  $\theta_1, \theta_2  \transmerge \theta_0$.
  If $\theta'_1 \lqacstr \theta_1$ and $\theta'_2 \lqacstr \theta_2$,
  then there are $\hat \theta_1 \eqconf \theta'_1$, $\hat\theta'_2 \eqconf \theta'_2$ and  $\theta'_0$ such that $\theta'_0 \lqacstr \theta_0$ and
  $\set{\hat\theta'_1 , \hat\theta'_2} \tranAC^n \theta'_0$ for some $n$ computable from
  $\theta'_1$ and $\theta'_2$.
\end{lemma}

% old stuff
%\begin{lemma}\label{lem:monotonic-merge}
%  Let $\theta_1,\theta_2,\theta_0$ be extended configurations such that
%  $\theta_1, \theta_2  \transmerge \theta_0$.
%  Given $\theta'_1 \lqacstr \theta_1$ and $\theta'_2 \lqacstr \theta_2$,
%  then there is $\theta'_0$ such that $\theta'_0 \lqacstr \theta_0$ and
%  $\set{\theta'_1 , \theta'_2} \tranAC^n \theta'_0$ for some $n$ computable from
%  $\theta'_1$ and $\theta'_2$.
%\end{lemma}

\begin{proof}
  Throughout the proof we use the following notation for $i\in\set{0,1,2}$:
  $\theta_i=(\DataState_i,\DataMonoid_i,r_i,m_i)$, similarly for
  $\theta'_i$. Moreover $a_i$ and $d_i$ denote the label and data value
  associated to $\theta_i$ (similarly for $\theta'_i$). By definition of
  $\transmerge$ we have $a_1=a_2$ and $d_1=d_2$. Suppose $\theta'_1 \lqacstr_{f_1} \theta_1$ and $\theta'_2 \lqacstr_{f_2} \theta_2$.  In particular (recall
  Remark~\ref{remark-label}) we have $a_1=a'_1=a_2=a'_2$, $f_1(d_1)=d'_1$
  and $f_2(d_2)=d'_2$.

  By definition of $\transmerge$, $\theta_0$ is constructed by taking the union
  of $\theta_1$ and $\theta_2$. Hence, for all $d$ we have
  $\typeof{\theta_0}(d)=\typeof{\theta_1}(d) \cup \typeof{\theta_2}(d)$.

  We need to merge $\theta'_1$ and $\theta'_2$ in order to reflect the way the
  union of $\theta_1$ and $\theta_2$ is done. For instance if $d$ is a data
  value that occurs both in $\data(\theta_1)$ and $\data(\theta_2)$ then we
  would like to identify $f_1(d)$ with $f_2(d)$. We also need to keep track of
  the number of data values having the same type, up to 2. This is
  essentially what we do below.

  Let $I=\data(\theta_1) \cap \data(\theta_2)$, $J_1=\data(\theta_1)\setminus
  I$, $J_2=\data(\theta_2)\setminus I$. Figure~\ref{fig:merge-compatibility} contains a depiction of the main argument that will be described in the next paragraphs.

  \begin{figure}
    \centering
    \includegraphics[width=.75\textwidth]{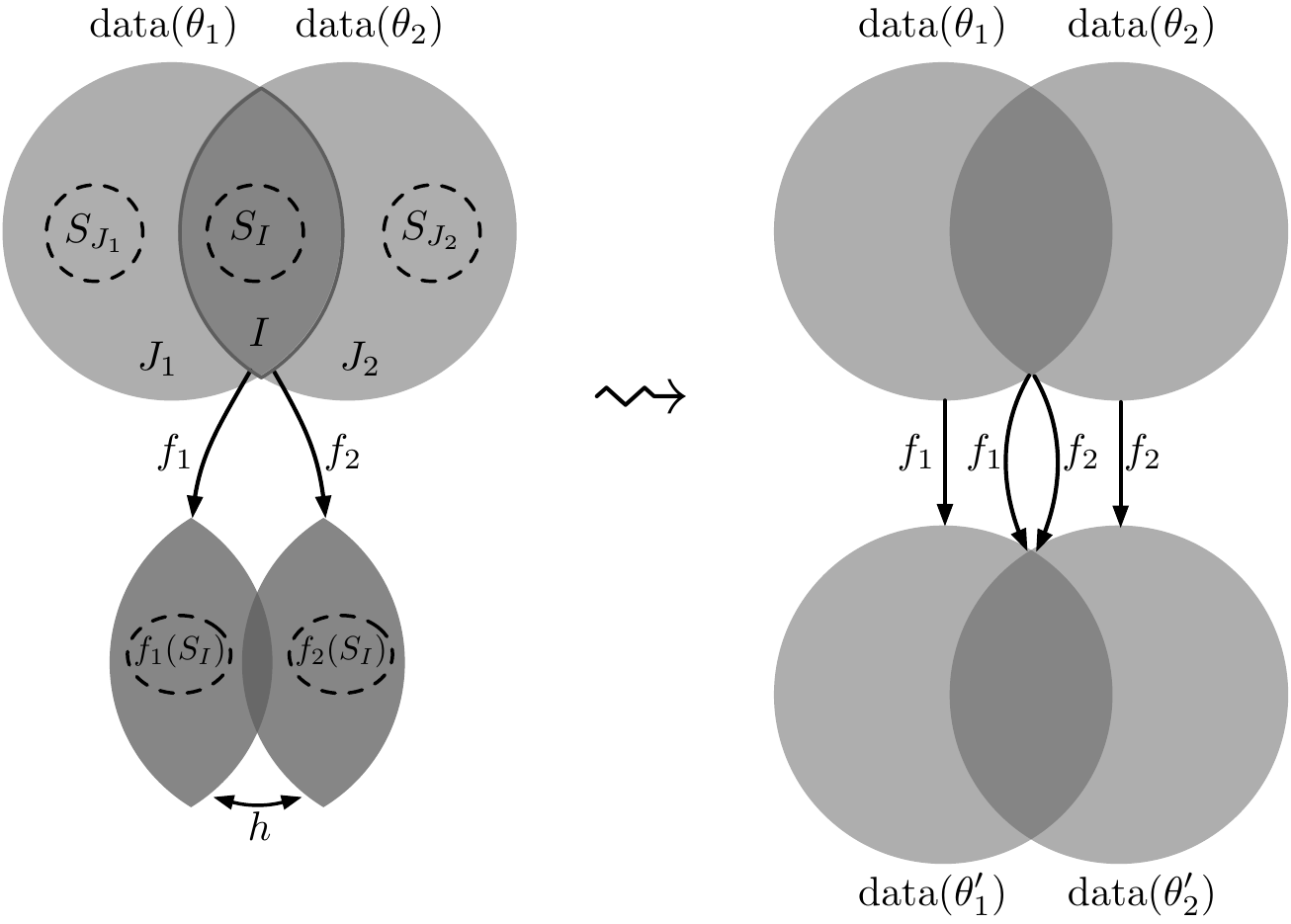}
    \caption{Argument of Lemma~\ref{lem:monotonic-merge}.}
    \label{fig:merge-compatibility}
  \end{figure}
  We first define sets $S_1=S_I\cup S_{J_1}$ and $S_2=S_I\cup S_{J_2}$ of
  \emph{special} data values. $S_I$ is constructed as follows: For each pair
  $\prof_1$, $\prof_2$ $\subseteq\Monoid$ we put in $S_I$ two (if possible)
  elements $d\in I$ such that $\DataMonoid_1(d)=\prof_1$ and
  $\DataMonoid_2(d)=\prof_2$. If only one such data value exists we put in
  $S_I$ only this one. Note that $S_I$ contains $d_1$ and hence
  $d_2$. Similarly, $S_{J_1}$ is constructed as follows: For each
  $\prof_1\subseteq\Monoid$ we add to $S_{J_1}$ two (if possible) elements
  $d\in J_1$ such that $\DataMonoid_1(d)=\prof_1$. If only one such data value
  exists we add only this one. $S_{J_2}$ is defined analogously. Note that
  the size of $S_1$ and $S_2$ are bounded by $2\cdot 2^{2 \cdot |\Monoid|}$.

  We then apply Lemma~\ref{lemma-injectivity} on $\theta'_1 \lqacstr\theta_1$
  with $S_1$ and on $\theta'_2 \lqacstr \theta_2$ with $S_2$ and further assume
  that $f_1$ and $f_2$ are injective on $S_1$ and $S_2$.

  \smallskip

  Next we establish a bijection between $f_1(I)$ and $f_2(I)$. 
Let $h$ be the relation $\set{(f_1(d),f_2(d)) \mid d \in I}$. 
Note that if we have the property $f_1(d)=f_1(d')$ iff $f_2(d)=f_2(d')$ for all $d \in I$, then $h$ is a bijection: it is an injective function due to the property, and it is surjective since we define it from $f_1(I)$ to $f_2(I)$. Although this property may not necessarily hold, we show next that in this case we can apply a number of $\transpump$ transitions to $\theta'_1$ and $\theta'_2$ in order to obtain extended configurations in which the property holds, and thus that we can safely assume that the property holds for the remaining of the proof. Indeed, suppose the property does not hold: for some $d,d'$ we have $f_1(d)=f_1(d')$ but $f_2(d)\neq f_2(d')$ (the other case of the failing of the property is, of course, symmetrical), and let $(S,\prof)=\typeof{\theta'_1}(f_1(d))$. Note that we  have
\[\prof=\DataMonoid_1(d)=\DataMonoid'_1(f_1(d))=\DataMonoid'_1(f_1(d'))=\DataMonoid_1(d').\]
Since $\profile(\theta_1)=\profile(\theta'_1)$, $\theta'_1$ must contain
another data value $e\neq f_1(d)$ such that $\DataMonoid'_1(e)=\prof$. Hence we
are enabled to apply $\transpumpSX$ to $\theta'_1$ and add a fresh copy $\hat
d$ of $f_1(d)$. The resulting extended configuration $\hat\theta'_1$ is still
smaller than $\theta_1$ by means of a function $\hat f_1$ identical to $f_1$
except that $\hat f_1(d)=\hat d\neq f_1(d')=\hat f_1(d')$. Note that now there is one
less pair $d,d' \in I$ contradicting $\hat f_1(d)=\hat f_1(d')$ iff
$f_2(d)=f_2(d')$, for $\hat \theta'_1 \lqacstr_{\hat f_1} \theta_1$ and
$\theta'_2 \lqacstr_{f_2} \theta_2$.  Thus, we can apply $\transpump$ to
$\theta'_1$ and $\theta'_2$ a number of times ---bounded by
$|\data(\theta'_1)|+|\data(\theta'_2)|$--- in order to obain extended
configurations $\hat\theta'_1 \lqacstr_{\hat f_1} \theta_1$ and $\hat\theta'_2
\lqacstr_{\hat f_2} \theta_2$ so that $\set{(\hat f_1(d),\hat f_2(d)) \mid d
  \in I}$ is a bijection. For the rest of the proof, and in order to avoid
introducing more symbols, we are simply going to assume that $h$ is a bijection
from $f_1(I)$ to $f_2(I)$ without loss of generality.

We extend $h$ to a bijection on $\D$ making sure that $h(d)\not\in S_{J_2}$ for
all $d\in S_{J_1}$, and $h(d) \not\in S_{J_1}$ for all $d \in S_{J_2}$. The
fact that we can extend $h$ follows easily from $\D$ being infinite, as we show
next. For $i =1,2$, let $D_i \subseteq \D \setminus (S_{J_1} \cup S_{J_2} \cup
f_1(I) \cup f_2(I))$ be so that $|D_i| = |S_{J_i}|$ so that $D_1 \cap D_2
=\emptyset$, and let $h_i$ be a bijection between $S_{J_i}$ and $D_i$. Now take
any bijection $h'$ between $\D \setminus (f_1(I) \cup S_{J_1} \cup S_{J_2})$
and $\D \setminus (f_2(I) \cup D_1 \cup D_2)$ ---it exists since these sets
have the same cardinality $\aleph_0$--- and define the extension as $h \cup h_1
\cup h_2 \cup h'$.

Finally, consider the extended configuration $\theta''_1$ resulting from replacing every data value $d \in \data(\theta'_1)$ with $h(d)$ in $\theta'_1$. The intuition is that in this way we make equal the data values of $f_1(I)$ with those of $f_2(I)$  without making equal any other values between $\theta'_1$ and $\theta'_2$. Let $f'_1 = f_1 \circ h$; we now have that $\theta''_1 \lqacstr_{f'_1} \theta_1$.

$\theta''_1$ and $f'_1$ satisfy the same hypothesis as the previous $\theta'_1$ and $f_1$ but now $f'_1(d)=f_2(d)$ for all $d\in I$ (in particular $d''_1=d'_2$) and furthermore $f'_1(I)=\data(\theta''_1)\cap\data(\theta'_2)=f_2(I)$.
Note that we are implicitly using that transitions (in particular $\transpump$) are closed under data bijections; that is, $\theta \transpump \theta'$ iff $f(\theta) \transpump f(\theta')$ for any data bijection $f$.
 
  Let $\theta'_0$ be the extended configuration obtained by applying
  $\transmerge$ to $\theta''_1$ and $\theta'_2$. We claim that $\theta'_0
  \lqacstr \theta_0$. This is witnessed by the function $f(d)=f'_1(d)$ if $d\in
  \data(\theta_1)$ and $f(d)=f_2(d)$ otherwise. Note that $f$ satisfies the
  surjectivity condition. It remains to show that $\profile(\theta'_0) =
  \profile(\theta_0)$ and for all $e\in\D$ we have $\DataState'_0(f(e))
  \subseteq \DataState_0(e)$ and $\DataMonoid'_0(f(e)) = \DataMonoid_0(e)$.

  \begin{enumerate}
  \item Let us first show $\DataState'_0(f(e)) \subseteq \DataState_0(e)$ and
    $\DataMonoid'_0(f(e)) = \DataMonoid_0(e)$. 
    \begin{enumerate}
    \item Assume $e\in I$. Then
    $\DataState'_0(f(e))=\DataState''_1(f(e)) \cup
    \DataState'_2(f(e))=\DataState''_1(f'_1(e)) \cup \DataState'_2(f_2(e))
    \subseteq \DataState_1(e) \cup \DataState_2(e) \subseteq
    \DataState_0(e)$. Similarly, $\DataMonoid'_0(f(e))=\DataMonoid''_1(f(e)) \cup
    \DataMonoid'_2(f(e))=\DataMonoid''_1(f'_1(e)) \cup \DataMonoid'_2(f_2(e))
    = \DataMonoid_1(e) \cup \DataMonoid_2(e) =
    \DataMonoid_0(e)$. 
    \item     If $e\in\data(\theta_1)\setminus I$ we have:
    $\DataState'_0(f(e))=\DataState''_1(f(e)) \cup
    \DataState'_2(f(e))=\DataState''_1(f'_1(e))
    \subseteq \DataState_1(e) \subseteq \DataState_0(e)$. Similarly, $\DataMonoid'_0(f(e))=\DataMonoid''_1(f(e)) \cup
    \DataMonoid'_2(f(e))=\DataMonoid''_1(f'_1(e))
    = \DataMonoid_1(e) = \DataMonoid_0(e)$.
    \end{enumerate}
  \item Let us now show $\profile(\theta'_0) = \profile(\theta_0)$. This is an immediate consequence of
    the injectivity of $f'_1$ and $f_2$ on the special values. \qedhere
\end{enumerate}
\end{proof}

\begin{proof}[Proof of Proposition~\ref{prop-compatible}]
It follows directly from Lemma~\ref{lem:monotonic} and Lemma~\ref{lem:monotonic-merge}.
\end{proof}

\medskip

\begin{corollary}\label{decid-wsts}
Given a \BUDTA \anAut, it is decidable whether an accepting extended
configuration is derivable from $\Theta_I^\eqconf$ in the transition system associated with \anAut. 
\end{corollary}
\begin{proof}
Let $f$ be the function computing the maximum value of $n$ as specified in Proposition~\ref{prop-compatible}.
Consider the following procedure.  
\begin{enumerate}
\item Initialize $\Theta$ to $\Theta_I$. 
\item While there is some $\theta$ so that $\Theta \not \lqAC \set\theta$ and
  $\Theta^\eqconf \tranAC^n \theta$ for $n \leq f(\Theta)$: add $\theta$ to
  $\Theta$.
\item Accept if there is an accepting extended configuration in $\Theta$, otherwise reject.
\end{enumerate}
Because $\lqacstr$ is a \wqo, the procedure always terminates.
It should also be noted that the second step is computable. Indeed there are
only finitely many elements within $\Theta^\eqconf$ to consider, those
derived from $\Theta$ via a permutation of the data
values $\data(\Theta)$, as the other ones would derive an extended
configuration equivalent to one already derived. This means that for example if
$n=1$, it suffices to consider all distinct $\theta$ modulo bijection of data values so that $\Theta \not \lqAC
\set\theta$ and 
\begin{itemize}
\item $\theta_1 \rightarrow \theta$ for some $\theta_1 \in \Theta$ for
  any transition except $\transmerge$, or
\item $\theta'_1, \theta'_2 \transmerge \theta$ for some $\theta'_1 \eqconf
  \theta_1$, $\theta'_2 \eqconf \theta_2$, and $\theta_1,\theta_2 \in
  \Theta$. However, in this case we only need to consider all pairs
  $(\theta'_1, \theta'_2)$ that are images of $(\theta_1,\theta_2)$ via a
  bijection of $\data(\theta_1) \cup \data(\theta_2)$, and there are only
  finitely many such bijections.
\end{itemize}
Note that due to transitions $\transguess$ and $\transup$ there may be
infinitely many such $\theta$, but only finitely many modulo $\eqconf$.

Let $\Theta$ be the final set after the evaluation of the algorithm. First,
note that all the extended configurations $\theta$ from $\Theta$ are so that
$\Theta_I^\eqconf \tranAC^+ \theta$. Therefore, if the algorithm accepts, there is
an accepting extended configuration derivable from $\Theta_I^\eqconf$. We
therefore show the converse.

Suppose that $\Theta_I^\eqconf \tranAC^t \theta_t$ and let $\theta_1, \dotsc,
\theta_t$ be the extended configurations derived at each step.
We show by induction that for every $i$ we have $\Theta \lqAC
\Theta_I \cup \set{\theta_1, \dotsc, \theta_i}$.  The base case when $i=0$ is
trivial since $\Theta_I \subseteq \Theta$. For the inductive case, suppose that
$\Theta \lqAC \Theta_I \cup \set{\theta_1, \dotsc, \theta_i}$. By
Proposition~\ref{prop-compatible}, since $\Theta_I \cup \set{\theta_1, \dotsc,
  \theta_i} \tranAC \theta_{i+1}$ there is some $n \leq f(\Theta)$ so that $\Theta^\eqconf
\tranAC^n \theta$ with $\theta \lqacstr \theta_{i+1}$. By the condition of the
algorithm, it must be so that $\Theta \lqAC \set \theta$ and hence $\Theta
\lqAC \set{\theta_{i+1}}$. Therefore, $\Theta \lqAC \Theta_I \cup
\set{\theta_1, \dotsc, \theta_{i+1}}$.

As a consequence of this property, if there is an accepting configuration
$\theta_F$ so that $\Theta_I^\eqconf \tranAC^+ \theta_F$, then in particular we have
$\Theta \lqAC \set {\theta_F}$ and hence there must be some $\theta \in \Theta$
so that $\theta \lqacstr \theta_F$. Since $\theta_F$ is accepting, and the set of
accepting extended configurations is downward closed
(Lemma~\ref{lemma-downward-closed}), it follows that $\theta$ is accepting.
\end{proof}

As shown next,  this implies the decidability for the emptiness problem for $\BUDTA$.

\subsection{From \BUDTA to its extended configurations}\label{section-BUDA-AC}
\newcommand{\transeps}{\ensuremath{T_\epsilon}\xspace}
\newcommand{\anWSTSM}{\ensuremath{\anWSTS_\anAut}\xspace}

The transition system $\anWSTS_\anAut$ associated to a \BUDTA \anAut is then
defined as follows.  Its elements are the extended configurations of \anAut as
defined in Section~\ref{section-abstract-configuration}. Its transition
relation is as defined in Section~\ref{sec-trans}. As shown in
Section~\ref{section-compatibility} the \wqo defined in Section~\ref{sec-wqo}
is compatible with the transition system. Hence coverability of
$\anWSTS_\anAut$ is decidable. It remains to show that  $\anWSTS_\anAut$ has the
desired behavior, i.e. that its coverability problem is equivalent to the
emptiness problem of \anAut. This is what we do in this section.

One direction is easy as the transition system can easily simulate \anAut. The
other direction requires more care. As evidenced in
\eqref{eq:form-of-derivation}, $\anWSTS_\anAut$ may perform a $\transpumpSX$
transition anytime. Remember that the effect of $\transpumpSX$ can be seen as a
result of the tree growing in width. We will see that this can be simulated by
a \BUDTA only when it moves up in the tree. This issue is solved by showing
that all transitions except $\transup$ commute with $\transpumpSX$ hence all
transitions \transpumpSX can be grouped just before a \transup and combined
with it in order to form an up-transition of a \BUDTA. As a consequence, we
obtain the following.
% (details in Appendix~\ref{section-lemma-aut-wsta}).
\begin{proposition}\label{prop-wsts-aut}
  Let \anAut be a \BUDTA. Let $\anWSTS_\anAut$ be the transition system
  associated with \anAut.  Then \anAut has an accepting run if, and only if,
  $\anWSTS_\anAut$ can reach an accepting extended configuration from the set
  of initial extended configurations.
\end{proposition}
In the sequel, we say that $\DataMonoid\subseteq (\Monoid\times\D)$ is
\emph{consistent}  with a data tree
$\tT=\aA\prd\bB\prd\dD\in\Trees(\A\times\B\times\D)$ when for every possible
$(\aMonoid,d)$, $\DataMonoid$ contains $(\aMonoid,d)$ if{f} there is a downward
path in $\tT$ that starts at the root and ends at some position $x$ such that
$\dD(x) = d$ and evaluates to $\aMonoid$ via $\morphism$. In particular this
implies that the label and data value of $\DataMonoid$ are the label and data
value of the root of $\tT$.

We first show that the transition system associated with a \BUDTA at least
simulates its behavior.

\begin{lemma}\label{lemma-aut-wsta}
  Consider an $\anAut \in \BUDTA$ and its associated transition system \anWSTSM. If
  \anAut has an accepting run then \anWSTSM can reach an accepting extended
  configuration from its initial extended configuration.
\end{lemma}
\begin{proof}%[Proof of Lemma~\ref{lemma-aut-wsta}]
  We show that from every accepting run $\rho$ of $\anAut$ on $\tT=\aA\prd\bB\prd\dD$ there
  exists a finite sequence of transitions in \anWSTSM starting in $\Theta_I^\eqconf$ and
  ending in an accepting configuration.

  Given a position $x \in \tpos(\tT)$, we define the extended configuration of $\tT|_x$ as $\theta_x=(\DataState_x,\DataMonoid_x,r_x,m_x)$, where
  $r_x=\textit{false}$ (unless $x$ is the root position), $m_x=\textit{false}$,
  $(\aA\prd\bB) (x)$ is the label of $\theta_x$, $\dD(x)$ is the data value of $\theta_x$, $\DataMonoid_x$ is consistent with $\tT|_x$, and
$\DataState_x = \set{(q,\bot,d) \mid (q,d) \in \rho(x)}$.
  
  We show by bottom-up induction on $x$ that $\theta_x$ is reachable from $\Theta_I^\eqconf$. By
  definition of $\Theta_I^\eqconf$, for any leaf node $x$ of $\tT$ we do have
  $\theta_x \in \Theta_I^\eqconf$ and hence our inductive process can start.

  Let now $x$ be a node of $\tT$ and let $x \conc 1, \dotsc, x \conc n$ be its
  children. By induction $\theta_{x \conc 1}, \dotsc, \theta_{x \conc n}$ are
  reachable from $\Theta_I^\eqconf$. For each $i$, and thread $(q,\bot,d)$ in $\theta_{x
    \conc i}$ there exists a transition $\tau_{i,q,d} \in \delta$ witnessing
  the fact that $\rho$ is a valid run. We derive $\theta_x$ from there as follows:
  \begin{enumerate}
  \item for each $i$, and each thread $(q,\bot,d)$ let $\theta_i$ be the
    extended configuration obtained from $\theta_{x\conc i}$ using a transition
    $\transdelta$ based on $\tau_{i,q,d}$,
\item for each $i$, we derive the extended configuration $\theta'_i$ from
    $\theta_i$ using the appropriate transitions for each thread newly introduced at the
    previous step. Note that because $\rho$ is accepting, each operation is
    successful and each new thread has $\topng$ or $\topg$ as second component. Hence, 
  \item for each $i$ we can apply $\transup$ to $\theta'_i$ and derive the
    extended configuration $\theta''_i$ using $\dD(x)$ and $(\aA\prd\bB)(x)$ as the data value and label of
    the new extended configuration,
\item we then make a sequence of $(n-1)$ applications of $\transmerge$ adding
  one by one $\theta''_i$ to the previously derived extended configuration,
\end{enumerate}

Repeating this simulation we finally derive the configuration $\theta_\eps$ of
the root of $\tT$. Since $\rho$ is accepting, for every $(q,d) \in \rho(\eps)$
there is $\tau = (t,\opaccept) \in \delta$ so that $\tT, \eps, (q,d) \models
t$. This means that for every $(q,\bot,d) \in \DataState_\eps$ one can apply a
$\transdelta$ transition based on $\tau$ arriving to an extended configuration
$\theta_\eps'
=(\DataState,\DataMonoid,r,m)$ with $r = \textit{true}$ and $\DataState =
\emptyset$ which is hence accepting.
\end{proof}

The other direction requires more care. First, we need to prove the following
lemma.

\begin{lemma}\label{lem:move-inc-forward}
  If we have $\theta_1 \tranaconfe \theta_2 \transpumpSX \theta_3$ then either
\begin{itemize}
\item $\hat \theta_1 \transpumpSpX \theta' \tranaconfe \theta_3$, or
\item 
  $\hat \theta_1 \transpumpSpX \theta' \tranaconfe
  \theta''\tranaconfe \theta_3$ 
\end{itemize}
for some extended configuration $\hat \theta_1, \theta',\theta''$ and set
$S'\subseteq Q$ such that $\hat \theta_1 \eqconf \theta_1$.
\end{lemma}
\begin{proof}
  Suppose that $\theta_1 \tranaconfe \theta_2 \transpumpSX \theta_3$. Notice
  that $\DataMonoid_1 = \DataMonoid_2$ by definition of $\tranaconfe$, and that
  $|\typeof{\theta_2}^{-1}(S,\prof)|\geq 1$ by definition of
  $\transpumpSX$. Let $e$ be a data value so that
  $\typeof{\theta_2}(e)=(S,\prof)$, and let $e' \in\data(\theta_3) \setminus
  \data(\theta_2)$ be the new data value added as a result of $\theta_2
  \transpumpSX \theta_3$. We thus have $\typeof{\theta_3}(e')=(\hat S,\prof)$
  with $\hat S = S \setminus \set{(q, \topg) : (q, \topg) \in S}$. Modulo
  replacing $\theta_1$ by an equivalent extended configuration we can further
  assume without any loss of generality that $e' \not\in \data(\theta_1)$.

Let $(q,\alpha,d)$ be the thread of $\DataState_1$ that triggers
$\tranaconfe$. We will treat all cases of $\tranaconfe$ at once, independently
of which particular transition it is.  Let $H \subseteq \DataState_2$ be the
new threads generated from $(q,\alpha,d)$ by $\theta_1 \tranaconfe
\theta_2$ (note that $H$ may be empty). We then have that:
 \begin{align}
   \DataState_2 &= (\DataState_1 \setminus \set{(q,\alpha,d)}) \cup H\label{eq:moveinc:ds2}\\
    \DataState_3 &= \DataState_2 \cup (\hat S\times\set{e'})\label{eq:moveinc:ds3}\\
    \DataMonoid_2 &= \DataMonoid_1 \label{eq:moveinc:dm2}\\
    \DataMonoid_3 &= \DataMonoid_2 \cup (\prof\times\set{e'}) \label{eq:moveinc:dm3}
 \end{align}

 Let $(S',\prof) = \typeof{\theta_1}(e)$ (recall that
 $\DataMonoid_1(e)=\DataMonoid_2(e)=\prof$). We show that $\transpumpSpX$ can
 be applied to $\theta_1$. In other words we show that
 $|\DataMonoid_1^{-1}(\prof)| \geq 2$. As $\transpumpSX$ was applied to
 $\theta_2$ we have that $|\DataMonoid_2^{-1}(\prof)|
 \geq 2$. Since $\DataMonoid_1 =
 \DataMonoid_2$, this implies $|\DataMonoid_1^{-1}(\prof)| \geq 2$ and we are
 done.

 We then define $\theta'= (\DataState', \DataMonoid',r_1,m_1)$ such that
 $\theta_1 \transpumpSpX \theta'$ with
\begin{align}
  \DataState' = \DataState_1 \cup (\hat S'\times\set{e'})\label{eq:moveinc:ds'}\\
\DataMonoid' = \DataMonoid_1  \cup (\prof\times\set{e'})\label{eq:moveinc:ms'}
\end{align}
where $\hat S' =  S' \setminus \set{(q, \topg) : (q, \topg) \in S'}$.

Notice that by \eqref{eq:moveinc:dm2}, \eqref{eq:moveinc:dm3} and
\eqref{eq:moveinc:ms'} we get $\DataMonoid' = \DataMonoid_3$.

\smallskip

We show that $\theta' \tranaconfe \theta_3$ or $\theta' \tranaconfe \theta''
\tranaconfe \theta_3$ for some $\theta''$. We distinguish between two possibilities: either $\DataState_1(e) = \DataState_2(e)$ or not.
\begin{itemize}
\item The easiest case is when $\DataState_1(e) = \DataState_2(e)$ (in
  particular $S=S'$).  This means that the two transitions of $\theta_1
  \tranaconfe \theta_2 \transpumpSX \theta_3$ do not interact with one
  another. Since a transition $\transpump$ preserves the truth of tests, the
  same transition as the one between $\theta_1$ and $\theta_2$ can be applied
  to $\theta'$. We show that this yields $\theta_3$.  As already mentioned,
  $\DataMonoid_3 = \DataMonoid'$.

  Let $H'$ be the new threads generated from $(q,\alpha,d)$ by applying this
  transition (i.e. the set of threads in the resulting extended configuration
  is $(\DataState' \setminus \set{(q,\alpha,d)}) \cup H'$). If the transition
  was a \transguess then we make sure that the guessed data value is the one
  that was guessed in $\theta_2$. It is therefore immediate to verify that
  $H=H'$ unless the transition was a $\transuniv$. In this latter case, assuming $\alpha = \opuniv(p)$, we have $H'=H\cup
  \set{(p,\top,e')}$. But we also have $(p,\top,e)\in\DataState_2$ because
  $e\in\data(\theta_1)$ and $\theta_1 \transuniv \theta_2$. As
  $\DataState_1(e)=\DataState_2(e)$ we have $(p,\top,e)\in\DataState_1$ and therefore
  $(p,\top,e')\in\DataState'$. In all cases we get  $(\DataState' \setminus
  \set{(q,\alpha,d)}) \cup H' = (\DataState' \setminus \set{(q,\alpha,d)}) \cup
  H$. Further:
  \begin{align*}
    (\DataState' \setminus \set{(q,\alpha,d)}) \cup H' &= (\DataState' \setminus \set{(q,\alpha,d)}) \cup H
    \tag{by the previous remark}\\
 &= \big( (\DataState_1
    \cup (\hat S'\times\set{e'})) \setminus \set{(q,\alpha,d)}
    \big) \cup H
    \tag{by \eqref{eq:moveinc:ds'}}\\
    &=(\DataState_1 \setminus \set{(q,\alpha,d)}) \cup H \cup (\hat S'\times\set{e'}) \tag{since $e' \neq d$}\\
    &=\DataState_2 \cup (\hat S'\times\set{e'}) \tag{by \eqref{eq:moveinc:ds2}}\\
    &=\DataState_2 \cup (\hat S\times\set{e'}) \tag{since $S=S'$}\\
    &=\DataState_3 . \tag{by \eqref{eq:moveinc:ds3}}
  \end{align*}
  Hence, we have that $\theta' \tranaconfe \theta_3$.

\item
  If $\DataState_1(e) \neq \DataState_2(e)$, we
  distinguish again between two possibilities depending on whether $d = e$ or not.
  \begin{itemize}
  \item Assume first that $d = e$. We are in a situation as
    the one depicted in Figure~\ref{fig:move-inc-forward}. 
\begin{figure}
    \centering
    \includegraphics[width=.75 \textwidth]{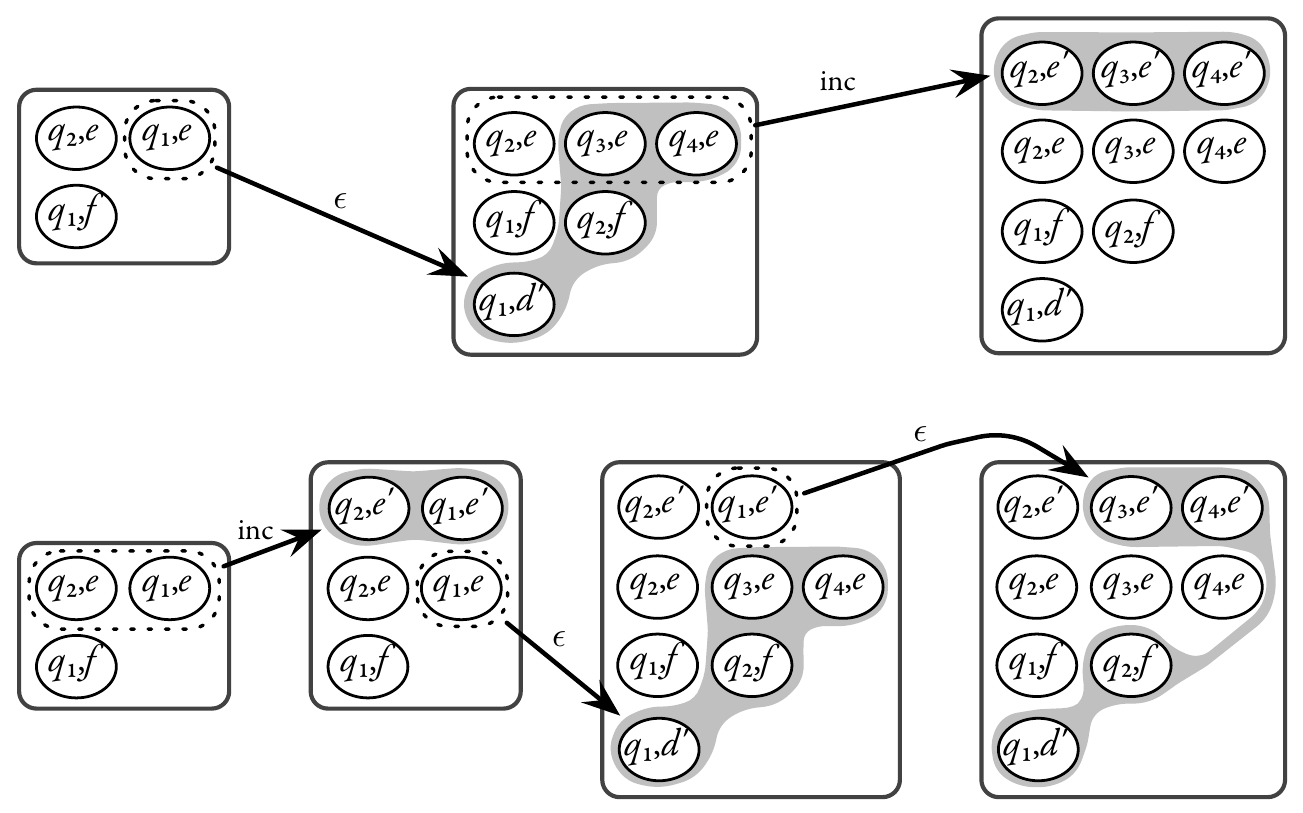}
    \caption{General idea of how to commute $\transpumpSX$ with $\tranaconfe$. In the figure we use a general $\tranaconfe$ which  creates new threads with the same data value and with fresh data value.}
    \label{fig:move-inc-forward}
  \end{figure}
  In this case we simply apply twice the same transition to $\theta'$, the
  first time using the thread $(q,\alpha,e)$ and the second time using
  $(q,\alpha,e')$ (in the case of $\transguess$ we guess twice the same data
  value). The resulting extended configuration is $\theta_3$.
\item The remaining situation is when $d \neq e$ and $\DataState_1(e) \neq
  \DataState_2(e)$. It could come from a transition $\transuniv$, $\transguess$
  or $\transstore$ (the other transitions only affect the data value $d$ and
  therefore $\DataState_1(e)=\DataState_2(e)$). 

  Suppose first that $\alpha =
  \opguess(p)$, and that it produces the thread $(p, \topg, e)$ as a result.
  We show that $\theta' \transguess \theta_3$.

  Notice that since $e$ is with a thread $(p, \topg)$, we have that $\hat S' =
  \hat S$, and we can apply the previous reasoning.  Indeed, by definition of
  $\transguess$ we have that the resulting $\DataState$ is
\begin{align*}
&  (\DataState' \setminus \set{(q,\opguess(p),d)}) \cup
  \set{(p,\topg,e)} \\
&=
((\DataState_1 \cup (\hat S'\times\set{e'})) \setminus \set{(q,\opguess(p),d)}) \cup
  \set{(p,\topg,e)} \tag{by \eqref{eq:moveinc:ds'}}\\
&=(\DataState_1 \setminus \set{(q,\opguess(p),d)}) \cup 
\set{(p,\topg,e)} \cup (\hat S'\times\set{e'}) \tag{since $e' \neq d$}\\
&= \DataState_2 \cup (\hat S'\times\set{e'})  \tag{by \eqref{eq:moveinc:ds2}}\\
&= \DataState_2 \cup (\hat S\times\set{e'})  \tag{since $\hat S' = \hat S$}\\
&= \DataState_3 \tag{by \eqref{eq:moveinc:ds3}}
.
\end{align*}
The fact that the resulting $\DataMonoid$ is $\DataMonoid_3$ is immediate. Hence, $\theta' \transguess \theta_3$.

\medskip

Suppose now that $\alpha = \opuniv(p)$.  We show that $\theta' \transuniv \theta_3$.

 By definition of $\transuniv$ we have that
 $S=S'\cup\set{(p,\top)}$. Furthermore we have:
\begin{align*}
&  (\DataState' \setminus \set{(q,\opuniv(p),d)}) \cup
  \set{(p,\top,d') : d' \in \data(\DataMonoid')} \\
&=
((\DataState_1 \cup (\hat S'\times\set{e'})) \setminus \set{(q,\opuniv(p),d)}) \cup
  \set{(p,\top,d') : d' \in \data(\DataMonoid')} \tag{by \eqref{eq:moveinc:ds'}}\\
&=(\DataState_1 \setminus \set{(q,\opuniv(p),d)}) \cup 
\set{(p,\top,d') : d' \in \data(\DataMonoid')} \cup (\hat S'\times\set{e'}) \tag{since $e' \neq d$}\\
&=(\DataState_1 \setminus \set{(q,\opuniv(p),d)}) \cup 
\set{(p,\top,d') : d' \in \data(\DataMonoid_1)} \cup \set{(p,\top,e')} \cup (\hat S'\times\set{e'}) \tag{since $e' \neq d$}\\
&= \DataState_2 \cup \set{(p,\top,e')} \cup (\hat S'\times\set{e'})  \tag{by
  definition of $\transuniv$}\\
&= \DataState_2 \cup (\hat S\times\set{e'})  \tag{since $S=S'\cup\set{(p,\top)}$}\\
&= \DataState_3 \tag{by \eqref{eq:moveinc:ds3}}
.
\end{align*}
Hence, $\theta' \transuniv \theta_3$.

The case of $\transstore$ is treated similarly. \qedhere
      \end{itemize}
\end{itemize}

\ 
\end{proof}

Finally we show:
\begin{lemma}\label{lemma-wsts-aut}
  Consider a \BUDTA \anAut and its associated transition system \anWSTSM. If \anWSTSM can reach
  an accepting extended configuration from the set of initial extended configurations
  then \anAut has an accepting run.
\end{lemma}
\begin{proof}%[Proof of Lemma~\ref{lemma-wsts-aut}]
  An extended configuration $\theta=(\Delta,\Gamma,r, m)$ is called
  \emph{starting} if for all $(q,\alpha,d)$ in
  $\Delta$ we have $\alpha=\bot$. Note that by definition of the transition
  relations, a starting extended configuration can only be obtained from a non starting
  extended configuration via a $\transup$ transition. Moreover starting extended
  configurations are preserved only by $\transmerge$ and $\transpump$ transitions.

  We show by induction on the length of the derivation that for every starting
  reachable extended configuration $\theta$ (by reachable we mean such that
  $\Theta_I^\eqconf \tranAC^+ \theta$), there exists a tree $\tT_\theta$ and a run
  $\rho_\theta$ of $\anAut$ on $\tT_\theta$, such that $\Gamma$ is consistent
  with $\tT_\theta$ and $\rho(x)$ is $\Delta$ for the root $x$ of $\tT_\theta$.

  This is clearly the case for all extended configurations in $\Theta_I^\eqconf$.

  For the inductive argument it is useful to notice that if a starting
  extended configuration $\theta$  can be associated with a tree $\tT_\theta$ and a run
  $\rho_\theta$ of $\anAut$ on $\tT_\theta$ satisfying the inductive
  hypothesis then, for any bijection $h$ on the data values, $h(\theta')$, $h(\tT_\theta)$
  and $h(\rho_\theta)$ satisfy also the inductive hypothesis. 
  
    Assume $\theta_1$ and $\theta_2$ are both starting reachable extended
  configurations and that $\theta_1,\theta_2 \transmerge \theta_0$. By
  induction we have trees $\tT_{\theta_1}$ and $\tT_{\theta_2}$ and runs
  $\rho_{\theta_1}$ and $\rho_{\theta_2}$ satisfying the induction
  hypothesis. By definition of \transmerge, $\theta_1$ and $\theta_2$ have the
  same label and data value. Hence by consistency of $\DataMonoid_1$ and
  $\DataMonoid_2$, the roots of $\tT_{\theta_1}$ and $\tT_{\theta_2}$ have the
  same label and data value. Let $\tT$ be the tree constructed from the union
  of $\tT_{\theta_1}$ and $\tT_{\theta_2}$ by identifying their roots. We show
  that $\tT$ is the desired $\tT_{\theta_0}$. The reader can easily verify
  that $\rho$, constructed by taking the union of $\rho_{\theta_1}$ and
  $\rho_{\theta_2}$, is a run of $\anAut$ on $\tT$ and that $\DataMonoid_0$ is
  consistent with $\tT$.

  Assume now that $\theta$ is a starting reachable extended configuration and that $\theta \transpumpSX \theta'$. Let $e$ be the data value such that $(S,\prof)=\typeof{\theta}(e)$ and let $e'$ be the new data value added in $\theta'$. By induction we have a tree $\tT_\theta$ and a run $\rho_\theta$ satisfying the induction hypothesis. Notice that by definition of $\transpumpSX$ the data value $e$ duplicated cannot be the one of the root of $\tT_\theta$ (because the associated $\prof'$ occurs only once). Let $\tT'_\theta$ be the tree obtained from $\tT_\theta$ by replacing $e$ with $e'$. Let $\tT$ be the tree constructed from the union of $\tT_\theta$ and $\tT'_\theta$ by identifying their roots. We show that $\tT$ is the desired tree. The reader can easily verify that $\rho$, constructed by taking the union of $\rho_{\theta}$ and $\rho'_{\theta}$, is a run of $\anAut$ on $\tT$, where $\rho'_{\theta}$ is the copy of $\rho_\theta$ on $\tT'_\theta$ and that $\DataMonoid'$ is consistent with $\tT$.

  It remains to show that a starting extended configuration obtained via \transup
  corresponds to a real configuration of $\anAut$. Assume $\theta'$ is a
  reachable extended configuration and that $\theta$ is such that $\theta'
  \transup \theta$. By definition, all threads $(q,\alpha,d)\in\DataState'$ are
  such that $\alpha\in\set{\topng, \topg}$. Consider a derivation witnessing the fact that
  $\theta'$ is reachable. Let $\theta_1$ be a starting extended configuration in
  this derivation such that all other extended configurations between
  $\theta_1$ and $\theta'$ are not starting. Hence we have $\theta_1
  (\tranaconfe|\transpump)^+ \theta'$ (no \transup nor \transmerge can occur in
  this derivation as a \transmerge can only appear right after a \transup,
  because of the $m$ flag, and a \transup would derive a starting extended
  configuration).
 By Lemma~\ref{lem:move-inc-forward}, $\theta'$ can
  equivalently be derived from $\hat \theta_1\eqconf\theta_1$ using a derivation of the form
  $\transpump\!\!{}^* \tranaconfe^+$ (a sequence of $\transpump$ followed by a
  sequence of $\tranaconfe$). As \transpump preserves startingness, this shows that
  there is a reachable starting extended transition $\theta_2$ such that $\theta_2
  \tranaconfe^+ \theta'\transup\theta$. By induction hypothesis we have a tree
  $\tT_{\theta_2}$ and a run $\rho_{\theta_2}$ satisfying the induction
  hypothesis. Let $\tT$ be the tree constructed from $\tT_{\theta_2}$ by adding
  a new node, having for label and data value those of $\theta$, and a unique immediate
  subtree $\tT_{\theta_2}$. Let $\rho$ be constructed from $\rho_{\theta_2}$ as
  follows: If $x$ is a node of $\tT$ occurring in $\tT_{\theta_2}$ then
  $\rho(x)=\rho_{\theta_2}(x)$. If $y$ is the new root of $\tT$ then, for each
  thread $(p,\bot,d)\in\DataState_2$ there must be a transition $\transdelta$
  in the derivation from $\theta_2$ to $\theta'$ using that thread (otherwise
  that thread would never disappear and a \transup transition would not be
  applicable on $\theta'$). Let $\tau=(t,a) \in \delta$ be the corresponding
  transition of $\anAut$. Because of the consistency condition of
  $\DataMonoid_2$ on $\tT_{\theta_2}$ we have that $\tT_{\theta_2},x,(p,d)
  \models t$, where $x$ is the root of $\tT_{\theta_2}$ and we add to $\rho(y)$
  the effect of the action $a$. The reader can now easily verify that this
  gives the desired tree and run.
\end{proof}

By combining the previous Lemmas we immediately obtain the proof of
Proposition~\ref{prop-wsts-aut}. Hence, combining
Proposition~\ref{prop-wsts-aut} and Corollary~\ref{decid-wsts}
Theorem~\ref{thm-budta-decid} is proven.

%%% Local Variables: 
%%% TeX-PDF-mode: t
%%% TeX-master: "main"
%%% ispell-local-dictionary: "american"
%%% End:  

 \section{Satisfiability of vertical {XP}ath.}\label{section-xpath}
\newcommand{\upa}{{\upw}}
\newcommand{\downa}{{\dow}}
\newcommand{\fkey}{\textit{dep}}
\newcommand{\xpathv}{\ensuremath{\xpath(\frak V,=)}\xspace}
\newcommand{\Nonoid}{\+M}
\newcommand{\aNonoid}{\nu}
\newcommand{\norphism}{g}

\newcommand{\llpar}{{\langle}\hspace{-.25em}{\langle}} 
\newcommand{\rrpar}{{\rangle}\hspace{-.25em}{\rangle}} 
\newcommand{\spsubf}{\msf{spsub}}

\newcommand{\moneps}{\varepsilon}

\newcommand{\qtyupdow}[1]{f_{#1}}%
\newcommand{\apath}{\pi}%

%The main results of this section are the following.
%\begin{corollary}
%  The inclusion and equivalence problems for node expressions of \xpathv are decidable.
%\end{corollary}
%\begin{corollary}
%  The satisfiability problem for \xpathv in the presence of unary inclusion dependency constraints is decidable.
%\end{corollary}

In order to conclude the proof of Theorem~\ref{thm:rvxpath-decidable} it
remains to show that \BUDTA can capture emptiness of $\rxpath(\frak V, =)$ expressions.
Given a formula $\eta$ of $\rxpath(\frak V,=)$, we say that a $\budta$ $\anAut$
is \emph{equivalent} to $\eta$ if for every data tree $\tT$, $\tT$ is accepted
by \anAut if{f} $\dbracket{\eta}^{\tT} \neq \emptyset$. The main contribution
of this section is the following.

\begin{proposition}\label{prop:regvxpath-2-buda}
  For every $\eta \in \rxpath(\frak V, =)$ there exists an equivalent
 $\anAut \in \budta$ computable from $\eta$.
\end{proposition}

We first give the general idea of the construction. From
$\eta\in\rxpath(\frak V,=)$ we actually compute an equivalent
$\budta^\eps$. By Proposition~\ref{prop:budaeps-2-buda} this is enough
to prove the result. Note that $\budta^\eps$ can easily simulate any
positive test $\tup{\alpha = \beta}$ or $\tup{\alpha \neq \beta}$ of
$\rxpath(\frak V, =)$ using a $\opguess$ action and tests of the form
$\tup{\rexp}^=$ and $\tup{\rexp}^{\neq}$. In the following examples, we disregard the internal alphabet $\B$ in the expressions $\rexp$ for clarity.  Consider, for example, the property $\tup{\rtdow[a] \neq \upw\dow[b]}$, which states that there is a descendant labeled $a$ with a different data value than a sibling labeled $b$. A $\budta^\eps$ automaton can test this property as follows.
\begin{enumerate}
\item It guesses a data value $d$ and stores it in the register.
\item It tests that $d$ can be reached by $\rtdow[a]$ with a test
  $\tup{\A^*a}^=$.
\item It moves up to its parent.
\item It tests that a data value different from $d$ can be reached in one of its
  children labeled with $b$, using the test $\tup{\A b}^{\neq}$.
\end{enumerate}

On the other hand, the simulation of negative tests ($\neg\tup{\alpha=\beta}$
or $\neg\tup{\alpha\neq\beta}$) is more complex as $\budta^\eps$ is not closed
under complementation. Nevertheless, the automaton has enough universal
behavior (in the operations $\opuniv$, $\overline{\tup{\rexp}^=}$ and
$\overline{\tup{\rexp}^{\neq}}$) in order to do the job. Consider for
example the formula $ \lnot\tup{\rtupw[b]\dow[a] = \rtdow[c]}$, that states
that no data value is shared between a descendant labeled $c$ and any $a$-child
of a $b$-ancestor. To test this property, the automaton behaves as follows.
\begin{enumerate}
\item It creates one thread in state $q$ for every data value in the subtree,
  using $\opuniv(q)$.
\item A thread in state $q$ tests whether the data value of the register
  is reachable by $\rtdow[c]$, using a test $\tup{\A^*c}^=$. If the test is
  successful, it changes to state $p$, otherwise it stops and accepts.
\item A thread in state $p$ moves up towards the root, and each time it
  finds a $b$, it tests that the currently stored data value cannot be reached
  by $\dow[a]$. This is done with a test of the kind
  $\overline{\tup{b a}^{=}}$.
\end{enumerate}

This is essentially what we do. As usual the details are slightly more
complicated. In particular we will have to deal with more complicated regular
expressions involving possibly complex node expressions. As our
automaton is bottom-up, we will need to compute all loops within a
subtree. We will use the internal alphabet \B for this purpose.

\smallskip

\begin{proof}[Proof of Proposition~\ref{prop:regvxpath-2-buda}]
  Let $\eta$ be a node expression in $\rxpath(\frak V,=)$. We construct a
  $\budta^\eps$ $\cl{A}_\eta$ that tests whether $\eta$ holds at all the leaves of the tree. Note that this is without any loss of generality, since it is then easy to test any formula $\eta$ at the root with $\tup{\rtupw[\lnot\tup{\upw} \land \eta]}$.

  We denote by $\nsubf(\eta)$ the node subformulas of $\eta$, and by
  $\psubf(\eta)$ the path subformulas of $\eta$. For any $\varphi \in
  \nsubf(\eta)$ we denote by $\overline\varphi$ its \emph{simple negation},
  that is, $\overline \varphi = \psi$ if $\varphi$ is of the form $\lnot \psi$,
  and $\overline \varphi = \lnot \varphi$ otherwise. By $\nsubf^\lnot(\eta)$ we denote
  the closure of $\nsubf(\eta)$ under simple negations.
% Without loss of
%  generality, we assume that every formula of $\nsubf^{\lnot}(\eta)$ is in
%  negated normal form\luc{negated normal form not defined}.

  For technical reasons, we distinguish between the \emph{nesting levels} of
  the formulas in $\nsubf(\eta)$. Node expressions of  nesting
  level $0$ are those testing node labels and any boolean combination of those. Node expressions of nesting level~$i+1$ are those of
  nesting level~$i$ plus those the form $\tup{\alpha=\beta}$ or
  $\tup{\alpha\neq\beta}$, and any boolean combination of them, where $\alpha$ and $\beta$ are path expressions using only node
  subexpressions of nesting level~$i$. We denote by $\nsubf_i(\eta)$ the
  subset of $\nsubf^{\lnot}(\eta)$ containing node expressions of nesting level~$i$.
  Similarly we denote by $\psubf_i(\eta)$ the path expressions of $\psubf(\eta)$
  using only formulas in $\nsubf_i(\eta)$ as node subexpressions. Note that the
  maximal nesting level $n$ is bounded by $|\eta|$.

  For each $i\leq n$, consider now the finite alphabet
\[ \A_{\eta,i} = \set{\upw, \dow, [\varphi] \mid \varphi \in \nsubf_i(\eta)} .\]
Every path expression $\alpha \in \psubf_i(\eta)$ can then be interpreted as a
regular expression over $\A_{\eta,i}$, and every word $w \in \A_{\eta,i}^*$ can be
interpreted as a path expression.

A \emph{path} $\apath$ of $\tT$ is a non-empty string of node positions of
$\tT$ (\ie, $\apath \in \tpos(\tT)^+$) so that every two consecutive elements
of $\apath$ are in a parent/child relation (\ie, one is the parent of the other). A path $\apath$ is \emph{looping}
if the first and last elements are the same, and it is \emph{non-ascending} if
each of its elements is either a descendant of the first element or equal to it.  We say that a
path $\apath$ of $\tT$ \emph{verifies} $w \in \A_{\eta,i}^*$ if $\apath$ behaves
according to the sequence of letters of $w$. More formally this means:
\begin{itemize}
\item $w=\epsilon$ and $|\apath|=1$;
\item $w=\upw w'$, $\apath = u v \apath'$ and $v$ is the parent of $u$ in
  $\tT$ and $v \apath'$ verifies $w'$;
\item $w=\dow w'$, $\apath = u v \apath'$ and $v$ is a child of $u$ in $\tT$ and $v \apath'$ verifies $w'$; or
\item $w=[\varphi] w'$, $\apath = u \apath'$ and $u \in \dbracket{\varphi}^\tT$ and $\apath$ verifies $w'$.
\end{itemize}

By the characterizations of regular languages, for each $i\leq n$, there exists
a finite monoid\footnote{Unlike in Section~\ref{sec:discussion-BUDA} we work
  here with an
  automata model with $\epsilon$ transitions. Therefore it is more convenient
  to use monoids instead of semigroups.} $\Nonoid_i$ and a
homomorphism $\norphism_i: \A_{\eta,i}^* \to \Nonoid_i$ such that for every
$\alpha \in \psubf_i(\eta)$ there is a set $S_\alpha \subseteq \Nonoid_i$ so
that $w \in \A_{\eta,i}^*$ is recognized by $\alpha$ if{f} $\norphism_i(w) \in
S_\alpha$. Let us denote by $1_{\Nonoid_i}$ the neutral element of $\Nonoid_i$
and by $\aNonoid_i, \aNonoid'_i$ the elements of $\Nonoid_i$.

The internal alphabet of $\anAut_\eta$ is $\B =
\subsets(\Nonoid_0)\times\cdots\times\subsets(\Nonoid_n)$.  Intuitively,
$\anAut_\eta$ accepts trees $\tT \otimes \bB$ so that for each $i\leq n$ and
any node $x$, the $i^{th}$ component of $\bB(x)$, denoted $\bB_i(x)$ in the
sequel, contains the set of all $\aNonoid_i \in \Nonoid_i$ so that there is $w
\in \A_{\eta,i}^*$ where
\begin{enumerate}
\item\label{cond:internal:1}  $\norphism_i(w) = \aNonoid_i$, and
\item\label{cond:internal:2} there is a non-ascending looping path $\apath$ of $\tT$ so that $\apath$ starts and ends in $x$ and  verifies $w$.
\end{enumerate}
In other words, $\bB(x)$ contains all the information about the non-ascending
looping paths at $x$ and one of the chief tasks of $\anAut_\eta$ is to ensure
that the $\bB(x)$ are properly set. For this, we have in the set of states $Q$
of $\anAut_\eta$ a state $\llpar \varphi \rrpar$ for any
$\varphi\in\nsubf^\lnot(\eta)$. We will design $\anAut_\eta$ such that in an
accepting run, if a thread in state $\llpar \varphi \rrpar$ is started at a node
$x\in\tT$ then $x \in \dbracket{\varphi}^\tT$.
% \diego{I changed this part to include the name $q_\B$ for the state
% responsible of checking the correctness of the internal alphabet. Otherwise,
% the conditions (a)--(d) would not be true, due to the definition of
% ``accepting run from $y$'' that does not take threads that verify $\bB_i$
% into account.}
When this is the case, properties \ref{cond:internal:1} and
\ref{cond:internal:2} are enforced by starting a thread at each leaf of $\tT$
with a state $q^i_\B$. A thread in state $q^i_\B$ moves up
in the tree while performing the following tests and actions at any node $x$,
where $\Lambda_i(x)$ is the set of all $\varphi\in\nsubf_i(\eta)$ such that
$x\in\dbracket{\varphi}^\tT$:
\begin{itemize}
\item If $x$ is a leaf, then $\bB_i(x)$ is the submonoid of $M_i$ generated by $\norphism_i(\Lambda_i(x))$. This property can be enforced by guessing a maximally consistent set $L_i \subseteq \nsubf_i(\eta)$ of formulas that hold at the node, testing whether $\bB_i(x)$ has the desired form (\ie, that it is the submonoid generated by $\norphism_i(L_i)$), and verifying that $L_i = \Lambda_i(x)$ by starting a thread with state $\llpar \varphi \rrpar$ for every $\varphi \in L_i$.

\item If $x$ is not a leaf then $\bB_i(x)$ is the submonoid of $M_i$ generated by $\norphism_i(\Lambda_i(x))\cup S_i$ where $S_i=\bigcup_{y \text{ child of } x} \norphism_i(\dow) \bB_i(y) \norphism_i(\upw)$.  This is done by guessing $L_i$ and $S_i$, performing the same tests and actions as above concerning $L_i$ and testing that $S_i$ is correct by testing the internal labels of the children of the current nodes using tests of the form $\tup{(\A\times\B) \conc (\A\times \set b)}$ and $\overline{\tup{(\A\times\B) \conc (\A\times \set b)}}$ for appropriate $b\in\B$.
\end{itemize}

\medskip

Thus, the initial state $q_0$ of $\anAut_\eta$ will simply create $n+2$  threads,  with states $q^i_\B$ for every $i$, and one with state $\llpar \eta \rrpar$.

Intuitively,  the automaton that verifies $\varphi$ starts with one thread in state $\llpar \varphi \rrpar$ at every leaf. The boolean connectives $\land, \lor$ are treated by the alternation/nondeterminism. If $\varphi$ is $\tup{\alpha  = \beta}$, the automaton guesses the witness data value and verifies that we can reach that value through $\alpha$ (resp.\ $\beta$) by going to a state $\llpar \alpha \rrpar^{=}_{1_\Nonoid}$. To verify this, the state $\llpar \alpha \rrpar^{=}_{1_\Nonoid}$ checks that either the data in the register is reachable through $\alpha$ in the subtree and through a test $\tup{e_\alpha}^=$ for a suitable regular expression $e_\alpha$, or it is elsewhere as a result of $\alpha$ starting with an upward axis, and the automaton recursively moves up to some $\llpar \alpha \rrpar^{=}_{\aNonoid}$ remembering the `history' of labels when going up in $\aNonoid$. Tests of the form $\tup{\alpha  = \beta}$ are treated similarly. Finally, for a test $\lnot \tup{\alpha = \beta}$ (resp.\ $\lnot \tup{\alpha \neq \beta}$), the automaton uses the power of unbounded alternation: it creates a new thread for each data value in the subtree reachable through $\alpha$ and passes the control to the states $\llpar \beta \rrpar_{\aNonoid}^{\lnot =}$ (resp.\ $\llpar \beta \rrpar_{\aNonoid}^{\lnot \neq}$), which are in charge of testing that no node can be reached through $\alpha$ with $=$ (resp.\ $\neq$) data value. Again, this has to be repeated for every ancestor as well, taking into account the type of the path in $\aNonoid$. Symmetrical conditions are also requested on $\beta$.

More concretely, the automaton $\anAut_\eta$ uses states of the from $\llpar \psi \rrpar$ or
$\llpar \alpha \rrpar^\odot_\aNonoid$, where $\psi$ is a subformula of
$\varphi$, $\alpha$ is a path expression of $\varphi$, and $\odot \in \set{=,
  \neq, \lnot\!\!=, \lnot\!\!\neq}$ and $\aNonoid \in \bigcup_i \Nonoid_i$.  It now
remains to explain how $\anAut_\eta$ can test whether
$x\in\dbracket{\varphi}^\tT$ by starting a thread with state $\llpar \varphi
\rrpar$ at node $x$. We explain this by induction on the structure of
$\varphi$. Using alternation of the automata model we can easily simulate
disjunctions and conjunctions. Testing node labels is also a simple
task. Altogether it remains to explain how formulas of the form $\tup{\alpha =
  \beta}$ and $\tup{\alpha \neq \beta}$, or their negations can be tested by
$\anAut_\eta$. In order to do so, for any $i\leq n$, $\aNonoid_i \in
\Nonoid_i,\alpha,\beta \in \psubf_i(\eta), \circledast \in
\set{=,\not=,{\lnot}{=},{\lnot}{\not=}},\odot \in
\set{{\lnot}{=},{\lnot}{\not=}}$, we have in $Q$ the states $\llpar \alpha
\rrpar^{\circledast}_{\aNonoid_i}$ and $\llpar \alpha, \beta
\rrpar^{\odot}_{\aNonoid_i}$, where $\alpha, \beta \in \psubf_i(\eta)$.

  We say that a thread $(q,d) \in Q\times\D$ has
an \emph{accepting run from $y$} if there is an accepting run $\rho$ where
instead of requiring that for every leaf $x$ of $\tT$, $\rho(x)$ is initial, we
require that 
\begin{itemize}
\item for every leaf $x$ of $\tT$, $\rho(x)$ contains a thread in state
  $q^i_\B$, where $i$ is the maximum nesting level of the formulas in the state
  $q$ and 
%\diego{changed to include    $q^i_\B$. Also, as you can see unfortunately I
%think we need to% talk about the nesting level of the formula described in the
%symbol for $q$... Do you have a simpler solution?}
\item $\rho(y)$ contains the thread $(q,d)$.
\end{itemize}
The transitions of $\anAut_\eta$ will be built so that the following conditions
are met, for every $d \in \D$.
\begin{enumerate}[label=(\alph*),ref=\alph*]
\item\label{cond:run:a} A thread $(\llpar \varphi \rrpar,d)$ has an accepting
  run from $x$ implies $x \in \dbracket{\varphi}^\tT$.
\item\label{cond:run:b} A thread $(\llpar\alpha\rrpar^=_{\aNonoid_i},d)$
  (resp. $(\llpar\alpha \rrpar^{\neq}_{\aNonoid_i},d)$) has an accepting run from
  $x$ implies there is a path $\apath$ and a word $w\in\A_{\eta,i}^*$ so that
  \begin{itemize}
  \item $\apath$ starts at $x$ and ends at a node with
    the same (resp. different) data value as $d$, and
  \item $\apath$ verifies $w$ and $\aNonoid_i\conc \norphism_i(w) \in S_\alpha$.
  \end{itemize}
\item\label{cond:run:c} A thread $(\llpar\alpha\rrpar^{\lnot =}_{\aNonoid_i},d)$
  (resp. $(\llpar\alpha \rrpar^{\lnot \neq}_{\aNonoid_i},d)$) has an accepting
  run from $x$ implies there is no path $\apath$ and no word $w\in\A_{\eta,i}^*$ so that
\begin{itemize}
  \item $\apath$ starts at $x$ and ends at a node with
    the same (resp. different) data value as $d$, and
  \item $\apath$ verifies $w$ and $\aNonoid_i\conc \norphism_i(w) \in S_\alpha$.
  \end{itemize}
\item\label{cond:run:d} A thread $(\llpar\alpha,\beta\rrpar^{\lnot
    =}_{\aNonoid_i},d)$ (resp. $(\llpar\alpha,\beta\rrpar^{\lnot
    \neq}_{\aNonoid_i},d)$) has an accepting run from $x$ implies there are
  no paths $\apath, \apath'$ and no words $w,w'\in\A_{\eta,i}^*$ so that
\begin{itemize}
\item both $\apath$ and $\apath'$ start at $x$ and end at nodes with the same
  (resp. different) data value, and
\item $\apath$ verifies $w$ where $\aNonoid_i\conc \norphism_i(w) \in S_\alpha$,
  and $\apath'$ verifies $w'$ where $\aNonoid_i\conc \norphism_i(w') \in S_\beta$.
  \end{itemize}
\end{enumerate}

Note that condition~\eqref{cond:run:a} is enough to conclude the correctness of
the construction. The other conditions are here to help enforcing this
condition.

It now remains to set the transition function of $\anAut_\eta$ in order to
ensure all the properties stated above. First note the following 
observation.  For every $i\leq n$ and $\aNonoid_i \in \Nonoid_i$ and $\alpha \in \psubf_i(\eta)$
the set of all words $(a_1,b_1), \dotsc, (a_m,b_m)$ of $(\A \times \B)^+$ such
that there are $\aNonoid_{i,1}\in b_1, \dotsc, \aNonoid_{i,m} \in b_m$ with \[\aNonoid_i
\conc (\aNonoid_{i,1}\conc \norphism_i(\dow) \conc \aNonoid_{i,2} \conc
\norphism_i(\dow)\conc \dotsb \conc \aNonoid_{i,m-1} \conc \norphism_i(\dow)\conc
\aNonoid_{i,m}) \in S_\alpha\] is a regular language and we denote by
$\rexp(\alpha,\aNonoid_i)$ the corresponding regular expression.

We now describe the behavior of each thread at a node $x$ depending on
its state.  In this description we assume that $\alpha$ and $\beta$ are in
$\psubf_i(\eta)$.

\begin{description}
\item [$\bullet$ $\llpar \tup{\alpha = \beta} \rrpar$] In this case $\anAut_\eta$ guesses
  a data value, stores it in its register using $\opguess$ and continues the
  execution with both states $\llpar \alpha\rrpar^=_{1_{\Nonoid_i}}$ and $\llpar
  \beta \rrpar^=_{1_{\Nonoid,i}}$, that will test if there exist two nodes
  accessible by $\alpha$ and $\beta$ respectively such that both carry the data
  value of the register.

\item [$\bullet$ $\llpar \tup{\alpha \neq \beta}\rrpar$] Similarly as above,
  $\anAut_\eta$ guesses a data value, stores it in its register and continues
  the execution with both states $\llpar \alpha\rrpar^=_{1_{\Nonoid_i}}$ and $\llpar
  \beta \rrpar^{\neq}_{1_{\Nonoid_i}}$, which are responsible of testing that there
  is a $\alpha$ path leading to a node with the same data value as the one in
  the register and a $\beta$ path leading to a node with a different data
  value.

\item [$\bullet$ $\llpar \alpha \rrpar^=_{\aNonoid_i}$  or $\llpar \beta
  \rrpar^{\neq}_{\aNonoid_i}$] We denote by $\odot$ the symbol $=$ or
  $\neq$ occurring in superscript. In this case $\anAut_\eta$ chooses
  non-deterministically between one of the following actions.
  \begin{itemize}
  \item It checks that  the
    required data value is already in the subtree, making the test
    $\tup{\rexp(\alpha,\aNonoid_i)}^\odot$.

  \item It moves up and switches state to $\llpar \alpha
    \rrpar^\odot_{\aNonoid_i\conc \aNonoid'_i \conc\norphism_i(\upw)}$ for some $\aNonoid'_i\in\bB_i(x)$.
  \end{itemize}

\medskip

\item [$\bullet$ $\llpar \lnot \tup{\alpha = \beta} \rrpar$ or $\llpar \lnot
  \tup{\alpha \neq \beta} \rrpar$] We denote by $\odot$ the symbol $=$ or
  $\neq$ occurring in the middle. In this case $\anAut_\eta$ continues the
  execution with state $\llpar \alpha, \beta \rrpar^{\lnot \odot}_{1_{\Nonoid_i}}$.

\item [$\bullet$ $\llpar \alpha, \beta \rrpar^{\lnot =}_{\aNonoid_i}$ or $\llpar
  \alpha, \beta \rrpar^{\lnot \neq}_{\aNonoid_i}$] We denote by $\odot$ the symbol
  $=$ or $\neq$ occurring in superscript. In this case $\anAut_\eta$ performs
  all the following actions using alternation:
  \begin{itemize}
  \item If the test $\optestnroot$ succeeds, it moves up, and creates a thread in state $\llpar  \alpha, \beta
  \rrpar^{\lnot \odot}_{\aNonoid_i\conc\aNonoid'_i\conc \norphism_i(\upw)}$ for each $\aNonoid'_i\in\bB_i(x)$.
\item For every data value $d$ in the subtree that can be reached via
  $\rexp(\alpha,\aNonoid_i)$ it moves to state $\llpar \beta \rrpar^{\lnot
    \odot}_{\aNonoid_i}$ with data value $d$. This can be achieved by performing a
  $\opuniv$ operation and then by choosing non-deterministically to perform one
  of the following transitions
    \begin{itemize}
    \item test $\tup{\rexp(\alpha,\aNonoid_i)}^=$ and move to state $\llpar \beta
      \rrpar^{\lnot \odot}_{\aNonoid_i}$, or
    \item test $\overline{\tup{\rexp(\alpha,\aNonoid_i)}^=}$ and accept.
    \end{itemize}
  \item Analogously, for every data value $d$ in the subtree that can be
    reached via $\rexp(\beta,\aNonoid_i)$, it moves to state $\llpar \alpha
    \rrpar^{\lnot \odot}_{\aNonoid_i}$ with data value $d$.
  \end{itemize}

\item [$\bullet$ $\llpar \alpha \rrpar^{\lnot =}_{\aNonoid_i}$ or $\llpar \alpha
  \rrpar^{\lnot \neq}_{\aNonoid_i}$] We denote by $\odot$ the symbol $=$ or $\neq$
  occurring in superscript. In this case $\anAut_\eta$ performs all the
  following actions using alternation:
\begin{itemize}
\item if the test $\optestnroot$ succeeds, for all $\aNonoid'_i\in\bB_i(x)$, it
  starts a new thread in state to $\llpar \alpha \rrpar^{\lnot
    \odot}_{\aNonoid_i\conc\aNonoid'_i\conc\norphism_i(\upw)}$ at the parent of
  the current node.
\item tests that
  $\overline{\tup{\rexp(\alpha,\aNonoid_i)}^\odot}$ holds.
\end{itemize}

\end{description}

\subsection*{Correctness}

We show, for every position $x$ and state that conditions
\eqref{cond:run:a}--\eqref{cond:run:d} hold. We proceed by induction on the
nesting level of the expressions involved. Recall that once the node
expressions of nesting level $i$ are correctly enforced by a thread of the form
$\llpar \varphi \rrpar$ then the behavior of $\anAut_\eta$ enforces that
$\bB_i(x)$ contains the correct information for all $x$ (i.e verify
\ref{cond:internal:1} and \ref{cond:internal:2}). Therefore
$\rexp(\alpha,\aNonoid_i)$ does find a downward path evaluating to $\aNonoid'_i$
such that $\aNonoid_i\conc\aNonoid'_i \in S_\alpha$.

The base case is the nesting level $0$. At this level node expressions only
test the labels of the current node and this is exactly what the automaton
does. Hence $\bB_0(x)$ verifies \ref{cond:internal:1} and
\ref{cond:internal:2}.

We now assume a correct behavior for nesting level $i$ and show the same for
level $i+1$. For this we need to show that~\eqref{cond:run:a} holds for
$\varphi\in\nsubf_{i+1}(\eta)$ and that~\eqref{cond:run:b}--\eqref{cond:run:d} holds
for $\alpha,\beta\in\psubf_{i}(\eta)$ and $\aNonoid_{i}\in\Nonoid_{i}$. 

\begin{itemize}
\item[] We prove~\eqref{cond:run:b} by induction on the depth of $x$ starting
  from the root.

  Suppose that there is an accepting run of $(\llpar\alpha\rrpar_{\aNonoid_i}^=
  , d)$ from $x$ on $\aA \otimes \bB \otimes \dD$. By definition of the
  transition set, one of the following must hold.
\begin{itemize}
\item The test $\tup{\rexp(\alpha,\aNonoid_i)}^=$ succeeded. Then, by inductive
    hypothesis on $i$, there is a downward path $\apath$ starting at $x$ and ending at some
    position with data value $d$ so that the path verifies some $w \in
    \A_{\eta,i}^*$ with $\aNonoid_i \cdot \norphism_i(w) \in S_\alpha$ and we
    are done. Note that this is the only possible case at the root of $\tT$,
    hence proving the base case.

  \item $(\llpar \alpha \rrpar^=_{\aNonoid_i \cdot \aNonoid'_i\cdot
      \norphism_i(\upw)},d)$ has an accepting run from the parent $y$ of
    $x$. By inductive hypothesis on $i$ there is a non-ascending looping path
    $\apath''$ that starts and ends at $x$, and verifying a word $w''\in\A_{\eta,i}$ so that
    $\norphism_i(w'')=\aNonoid'_i$.  By inductive hypothesis on the depth of
    $x$, this means that there is a path $\apath$ that starts at $y$ and ends
    at a node with data value $d$, so that $\apath$ verifies some $w \in
    \A_{\eta,i}^*$ where $\aNonoid_i \cdot \aNonoid'_i \cdot \norphism_i(\upw)
    \cdot \norphism_i(w) \in S_\alpha$. Therefore the path $\apath' = \apath''
     \cdot
    \apath$ starts at $x$ and ends at a node with data value $d$, and $\apath'$
    verifies $w'= w'' \upw w$, where $\aNonoid_i \cdot \norphism_i(w')
    \in S_\alpha$.
\end{itemize}

Suppose now there is a path $\apath$ starting in $x$ and ending at a node with
a data value $d$, verifying some $w \in \A_{\eta,i}^*$ so that $\aNonoid_i
\cdot \norphism_i(w) \in S_\alpha$. We want to show that the thread $(\llpar \alpha
\rrpar^=_{\aNonoid_i},d)$ has an accepting run starting at $x$.
Then either
\begin{itemize}
\item $\apath$ is non-ascending, then by induction on $i$, $w
  \in\rexp(\alpha,\aNonoid_i)$ and the test $\tup{\rexp(\alpha,\aNonoid_i)}^=$ succeeds.

\item $\apath=\apath'\cdot \apath''$ where $\apath'$ is the maximal prefix of
  $\apath$ that is looping and non-ascending and $\apath''$ starts at the
  parent $y$ of $x$. Let $w'$ be such that $w=w'w''$, $\apath'$ verifies $w'$
  and $\apath''$ verifies $w''$. Let $\aNonoid'_i=\norphism_i(w')$.  By
  induction on the depth of $x$, we know that the thread $(\llpar \alpha
  \rrpar^=_{\aNonoid_i\conc\aNonoid'_i\conc \norphism_i(\upw)},d)$ has an accepting run starting at
  $y$.
\end{itemize}

The case of  $\llpar \alpha \rrpar^{\neq}_{\aNonoid_i}$ is treated similarly.

\medskip

The cases~\eqref{cond:run:c} and~\eqref{cond:run:d} are treated similarly.

\medskip

\item [] Consider for example~\eqref{cond:run:c}.  The proof is again by induction on
  the depth of $x$ starting from the root. We do only the case
  $\llpar\alpha\rrpar_{\aNonoid_i}^{\neg =}$ as the other one is proven identically.

  Suppose that there is an accepting run of
  $(\llpar\alpha\rrpar_{\aNonoid_i}^{\neg =}
  , d)$ from $x$ on $\aA \otimes \bB \otimes \dD$. By definition of the
  transition set, all of the following must hold.
\begin{itemize}
\item The test $\overline{\tup{\rexp(\alpha,\aNonoid_i)}^=}$ succeeded. Then,
  by inductive hypothesis on $i$, there is no downward path $\apath$ starting
  at $x$ and ending at some position with data value $d$ so that the path
  verifies some $w \in \A_{\eta,i}^*$ with $\aNonoid_i \cdot \norphism_i(w) \in
  S_\alpha$. Note that this is the only possible case at the root of $\tT$,
  hence proving the base case.

\item $(\llpar \alpha \rrpar^{\lnot =}_{\aNonoid_i \cdot \aNonoid'_i\cdot
    \norphism_i(\upw)},d)$ has an accepting run from the parent $y$ of $x$ for
  any $\aNonoid'_i\in\bB_i(x)$. By inductive hypothesis on the depth of $x$,
  this means that there is no path $\apath$ that starts at $y$ and ends at a
  node with data value $d$, so that $\apath$ verifies some $w \in
  \A_{\eta,i}^*$ where $\aNonoid_i \cdot \aNonoid'_i \cdot \norphism_i(\upw)
  \cdot \norphism_i(w) \in S_\alpha$. Assume now that there was a path
  $\apath'$ starting from $x$ and ending at some position with data value $d$
  so that the path verifies some $w \in \A_{\eta,i}^*$ with $\aNonoid_i \cdot
  \norphism_i(w) \in S_\alpha$. If this path is not ascending then it is a
  contradiction with the previous case. If this path is ascending then it
  starts with a non-ascending loop at $a$ and the continue from $y$. By
  induction on $i$ the non-ascending loop part evaluates to
  $\aNonoid\in\bB(x)$ and therefore we also get a contradiction with the the
  fact that $(\llpar \alpha \rrpar^{\lnot =}_{\aNonoid_i \cdot \aNonoid'_i\cdot
    \norphism_i(\upw)},d)$ has an accepting run from $y$.
\end{itemize}

The converse direction is treated similarly.

\medskip

\item [] A similar reasoning shows the case~\eqref{cond:run:d} and is omitted here.

\medskip

\item [] Property~\eqref{cond:run:a} is now a simple consequence of~\eqref{cond:run:b}--\eqref{cond:run:d}.\qedhere
\end{itemize}
\end{proof}

%%% Local Variables: 
%%% TeX-PDF-mode: t
%%% TeX-master: "main"
%%% End: 

\section{Concluding remarks}

We have exhibited a decidable class of automata over data trees. This automaton
model is powerful enough to code node expressions of $\rxpath(\frak V, =)$. Therefore, since
these expressions are closed under negation, we have shown
decidability of the satisfiability, containment and equivalence problems for node
expressions of $\rxpath(\frak V, =)$.

Consider the containment problem for \emph{path} expressions as the problem of, given two path expressions $\alpha$ and $\beta$ of our logic, whether $\dbracket{\alpha}^\tT \subseteq \dbracket{\beta}^\tT$ for all data trees $\tT$.
It is straigtforward that the technique used in \cite[Proposition 4]{CateL09} to reduce containment of path expresions into the satisfiability problem for node expressions works also in our context. Indeed, this technique is independent of having data equality tests and only requires that the logic be closed under boolean connectives. We therefore obtain a decision procedure for this problem as a corollary of Theorem~\ref{thm:rvxpath-decidable}.
\begin{proposition}
The containment problem for path expressions of $\rxpath(\frak V, =)$ is decidable.
\end{proposition}
 
Our decision algorithm relies heavily on the fact that we work
with {\bf unranked} data trees. As already shown in~\cite{FS09} without this assumption
\vxpath would be undecidable. In particular if we further impose the presence
of a DTD, \vxpath becomes undecidable, unless the DTD is simple enough for
being expressible as a \budta.

% It is worth noting that \xpathv, contrary to, e.g.,
% $\xpath(\frak{F},=)$, can express unary \emph{inclusion dependency} constraints. That
% is, whether for some symbols $a,b$ every $a$-node carries the data value of
% some $b$-node.
% \begin{lemma}\label{lem:foreign-key-vertical-xpath}
%   For $a,b \in \A$ let $\fkey(a,b)$ be the property over a data tree $\tT=\aA
%   \otimes \dD$: ``\textsl{For every position $x \in \tpos(\tT)$ such that
%     $\aA(x)=a$ there exists a position $x' \in \tpos(\tT)$ with $\aA(x')=b$ and
%     $\dD(x)=\dD(x')$}''. Then, $\fkey(a,b)$ is expressible in \xpathv for any
%   $a,b$.
% \end{lemma}
% \begin{proof}
% $\fkey(a,b)$ is expressed by the following formula: $\lnot \tup{\upa^*\downa^*[a \land \lnot \tup{\varepsilon = \upa^*\downa^*[b] }]}$.
% \end{proof}

% However, as shown in \citep{DL-tocl08,FS09}, if we are
% able to restrict the tree to have a word-like structure, the satisfiability
% problem for \xpathv becomes undecidable, and we then have the
% following.
% \begin{proposition}
%   The satisfiability problem for \xpathv in the presence of DTDs is undecidable.
% \end{proposition}

%%% Local Variables: 
%%% TeX-PDF-mode: t
%%% TeX-master: "main"
%%% End: 

 \bibliographystyle{alpha}
 \bibliography{biblio}

%\newpage\appendix
%\input{appendix.tex}

 \end{document}